\documentclass[12pt]{article}

\usepackage{amstext,amsmath,amssymb,amsfonts,bbm,slashed}
\usepackage[latin1]{inputenc}
\usepackage{epsfig}
\usepackage{hyperref}
\usepackage{amsthm}
\usepackage{subfigure}
\usepackage{color}
\usepackage{multirow}

\newcommand{\kin}{\text{kin\,}}

\newcommand{\inter}{\text{int\,}}

\def\C{{\mathbbm C}}

\newtheorem{theorem}{Theorem}
\newtheorem{proposition}{Proposition}

\newcommand{\bea}{\begin{eqnarray}}
\newcommand{\eea}{\end{eqnarray}}
\newcommand{\beq}{\begin{equation}}
\newcommand{\eeq}{\end{equation}}

\newcommand{\vp}{\varphi}
\newcommand{\bvp}{\bar{\varphi}}

\makeatletter

\begin{document}

\begin{center}

{\LARGE\bf Closed equations of the two-point functions\\ 
\vspace{5pt} for tensorial group field theory }
\vspace{15pt}

{\large  Dine Ousmane Samary}

\vspace{15pt}

{\sl Perimeter Institute for Theoretical Physics\\
31 Caroline St. N.
Waterloo, ON N2L 2Y5, Canada\\

\vspace{0.5cm}

International Chair in Mathematical Physics and Applications\\ (ICMPA-UNESCO Chair), University of Abomey-Calavi,\\
072B.P.50, Cotonou, Republic of Benin\\
}
\vspace{5pt}
E-mails:  {\sl dsamary@perimeterinstitute.ca\quad 
}

\vspace{10pt}

\begin{abstract}
In this paper we  provide the closed equations that satisfy two-point correlation functions of the rank 3 and 4 tensorial group field theory. The formulation of the present problem extends the method used by Grosse and Wulkenhaar in [arXiv 0909.1389] to the tensor case. Ward-Takahashi identities and Schwinger-Dyson equations are combined to establish a nonlinear integral equation for the two-point functions. In the 3D case
 the solution of this equation is given perturbatively at second order of the coupling constant.  
\end{abstract}

\end{center}

\noindent  Pacs numbers:  11.10.Gh, 04.60.-m
\\

\setcounter{footnote}{0}

%%%%%%%%%%%%%%%%%%%%%%%%%%%%

%%%%%%%%%%%%%%%%%%%%%%%%%%%%%

\section{Introduction} 
Random Tensor Models \cite{Rivasseau:2013uca}\cite{Rivasseau:2012yp}\cite{Rivasseau:2011hm} extends Matrix Models \cite{Di Francesco:1993nw} as  promising candidates to understand Quantum Gravity in higher dimension, $D 
\geq 3$. The formulation of such models is based on a Feynman path integral generating randomly graphs representing simplicial pseudo manifolds of dimension $D$. The equivalent of the t'Hooft large $N$ limit \cite{'tHooft:1982tz}\cite{'tHooft:1983wm}  for these Tensor Models      has been recently discovered by Gurau \cite{Gurau:2010ba}\cite{Gurau:2011aq}\cite{Gurau:2011xq}. The large $N$ limit behaviour is a powerful tool which  allows to understand the continuous limit of these models through, for instance, the study of critical exponents and phase transitions \cite{Dartois:2013sra}\cite{Bonzom:2011ev}\cite{Bonzom:2011zz}. 

With the advent of the field theory formulation of Random Tensor Models, henceforth called Tensorial Group Field Theory (TGFT) \cite{BenGeloun:2011rc}--\cite{Samary:2013xla}, one addresses several different questions such as Renormalizability (for removing divergences) and the study UV behaviour of these models. It turns out that Renormalization can be consistently defined for TGFTs and most of them, for the higher rank $D 
\geq 3$ are UV asymptotically free \cite{Geloun:2013saa}. This is of course very encouraging for the Geomeogenesis Scenario \cite{BenGeloun:2012pu}\cite{BenGeloun:2012yk}\cite{Samary:2013xla}.

It becomes more and more convincing that  Random Tensors and TGFT's will take a growing role  for giving answers for the Quantum Gravity conundrum. Despite all these results, a lot of questions (both conceptual and technical) arise in this framework for obtaining a final and emergent theory of General Relativity \cite{Rivasseau:2011hm}. Among other goals, it would be strongly desirable to establish more connections with other studies and important results around Gravity. This is the purpose of this paper which provides the first glimpses of the extension of the recent full resolution of the correlation functions in the Grosse-Wulkenhaar (GW) model \cite{Grosse:2009pa}\cite{Grosse:2013iva}
\cite{Grosse:2012uv}.

One of the main purpose of a field theory is to find  the exact value of the Green's functions  also called correlation functions. Obviously, this can be a highly nontrivial task.  In almost scarce cases where this is successfully done, one calls the model exactly solvable.  In a recent work,  the renormalizable noncommutative scalar field theory called the GW model was solved  \cite{Grosse:2013iva}--\cite{Grosse:2003aj}.  This particular noncommutative field theory projects on a matrix model and then can be seen a model for QG in $2D$. Let us review this model arising in Noncommutative Geometry. Grosse and Wulkenhaar   modified the propagator of the noncommutative  field theory by adding a harmonic term and showed that the resulting   functional action is renormalizable at all orders of pertubation.
  The proof of this claim   was given using the matrix basis dual to the Moyal space of functions.  In \cite{Rivasseau:2005bh}\cite{Rivasseau:2007rz}\cite{Gurau:2005gd}  a new proof of the renormalizability was given in direct space using  multiscale analysis \cite{Riv1}. The GW propagator breaks the $U(N)$ symmetry invariance  in the
infrared regime, but is asymptotically safe in the ultraviolet regime  \cite{Disertori:2006nq}\cite{Disertori:2006uy}\cite{Grosse:2004by}.
The model is also non invariant under translation and rotation of spacetime. The only known invariance satisfied by the model is the so-called  Langman Szabo duality \cite{Langmann:2002cc}. At the perturbative level, the associated Feynman graphs are ribbon graphs.   In a recent remarkable contribution, Grosse and Wulkenhaar  solve successfully all correlators in this model.  Using both Ward-Takahashi identities and the Schwinger-Dyson equation, these authors provide, via Hilbert transform, a nonlinear integral equation for the two-point functions \cite{Grosse:2012uv}. From this result, they were able to generate solutions for all correlators. Thus, the GW model is exactly nonperturbatively solvable. The question is whether or not this method may apply to other models, in particular to TGFTs dealing with higher rank tensors. We give a partial positive answer of this question. Indeed, as we will show in the following, the resolution method can be applied to find nonlinear equations for the correlations here as well. Due to the highly nontrivial equations and combinatorics, the full resolution of all correlators deserves more work which should be addressed elsewhere. 

The present paper is organized as follows.
In the section \ref{sec2},  we  derive the Ward-Takahashi identities of arbitrary rank $D$ TGFT.  In section \ref{sec3} we give the closed equation of the two-point correlation functions for the rank $3$ TGFT.  We also give the solution of this equation at second order of perturbation. In section \ref{sec4} we provide the closed equation of rank $4$ tensor field.  We give a summary of our results and outlook of the paper in section \ref{sec5}.

%%%%%%%%%%%%%%%%%%%%%%%%%%%%%%

%%%%%%%%%%%%%%%%%%%%%%%%%%%%%

\section{Ward-Takahashi identities for arbitrary $D$-tensor field model}\label{sec2}
TGFT's are generally defined by an action  $S[\bvp,\vp]$, that depends on the field $\vp$ and its conjugate $\bvp$ defined on the compact Lie group $G$  i.e.
$
\vp:   G^D\longrightarrow \C;\,\,\,
(g_1,\cdots,g_D)\longmapsto \vp(g_1,\cdots,g_D).
$ 
 For simplicity, we will always consider $G=U(1)$.
We are using the Fourier transformation of the field and are defining the momentum variable  associated to the group elements $[g]=(g_1,g_2,\cdots, g_D)\in U(1)^D$ as $[p]=(p_{1},p_{2},\cdots ,p_{D})\in \mathbb{Z}^D$. Using the parametrization  $g_k=e^{i\theta_k}$  we write
\bea
\vp({g_1,\cdots,g_D})=\sum_{p_i\in\mathbb{Z}} \vp({p_1,\cdots, p_D})e^{i\sum_k\theta_k p_k},\quad \theta_i\in[0,2\pi).
\eea
The Fourier transform of the field $\vp$  is denoted by  $\vp_{12\cdots D}=:\vp(p_1,\cdots,p_D)=:\vp_{[D]}$ for simplicity. 
The functional action $S[\bvp,\vp]$ is written in general case as
\bea\label{actiontensor}
S[\bvp,\vp]=\sum_{p_i}\bvp_{12\cdots D}C^{-1}(p_1,p_2,\cdots,p_D;p'_1,p'_2,\cdots p'_D)\vp_{12\cdots D}\prod_{i=1}^D\delta_{p_ip'_i}+S^{\inter}
\eea 
where $C$ stands for the propagator and $S^{\inter}$ collects all  vertex contributions of the interaction.
Let $d\mu_C$ be the field measure associated with  the covariance $C$, we have the relation
\beq
C([p];[p'])=\int\,d\mu_C\,\vp_{[p]}\bvp_{[p']},\quad d\mu_C=\prod_{[p]}d\bvp_{[p]}\,d\vp_{[p]}
e^{-\bvp_{[p]}\,C^{-1}([p],[p])\,\vp_{[p]}}.
\eeq
The Green's functions or $N$-point correlation functions are defined by  
 the relation
\bea\label{Green}
G([p]_1,[p]_2,\cdots [p]_N)=\frac{1}{\mathcal Z}\int\,d\mu_C\,\,\vp_{[p]_1}\bvp_{[p]_1}\cdots\vp_{[p]_N}\bvp_{[p]_N}e^{-S^{\inter}},
\eea
 where $\mathcal Z$ is the normalization factor also called partition function given by
\bea
\mathcal Z=\int\,d\mu_C\, e^{-S^{\inter}}.
\eea
Let us write the interaction term of the action \eqref{actiontensor} as $S^{\inter}=\lambda V[\bvp,\vp]=:\sum_k \lambda_k V_k[\bvp,\vp]$. The main idea of the pertubative theory is to expand the Green's functions \label{Green} as 
\bea
G([p]_1,[p]_2,\cdots [p]_N)&=&\sum_{n=0}^\infty\,\frac{(-\lambda)^n}{n!}\int\,d\mu_C\, V^n[\bvp,\vp] \,\vp_{[p]_1}\bvp_{[p]_1}\cdots\vp_{[p]_N}\bvp_{[p]_N}\cr
&=&\sum_{n=0}^\infty\, \lambda^n\,G^{(n)}_{ N}.
\eea 
Using this formula, the Green's functions can be computed order by order using Dyson's theorem.

We consider the rank $D$ tensor
 field $\vp$ and its conjugate $\bvp$, which 
  are transformed under 
the tensor product of  $D$ fundamental representations of the unitary group 
$\mathcal U_{\otimes}^{N_D}:=\otimes_{i=1}^D U(N_i)$.   Let $U^{(a)}\in U(N_a)$, $a=1,2,\cdots,D$. The field $\vp$ and its conjugate $\bvp$ are transformed under $U(N_a)$ as
\bea
\vp_{12\cdots D}\rightarrow [U^{(a)}\vp]_{12\cdots a\cdots D}=\sum_{p'_a\in Z} U^{(a)}_{ p_ap'_a}\vp_{12\cdots  a'\cdots D},\\
\bvp_{12\cdots D}\rightarrow [\bvp U^{\dag(a)}]_{12\cdots  a\cdots D}=\sum_{p'_a\in Z} \bar U^{(a)}_{p_ap'_a}\bvp_{12\cdots  a'\cdots D}.
\eea 
$p'_a$ or simply $a'$ is the momentum index at the position $a$ in the expression $\vp_{12\cdots a'\cdots D}$. The kinetic action in \eqref{actiontensor} is re-expressed as follows
\bea\label{actionintddd}
S^{\kin}[\bvp,\vp]=\sum_{p_1,\cdots,p_D}\vp_{12\cdots D}
M_{12\cdots D}\bvp_{12\cdots D},\quad M_{12\cdots D}=C^{-1}_{12\cdots D}.
\eea
$M_{12\cdots D}$ is the inverse of propagator associated to the model.  Rank $D$ tensor fields are represented by half lines made with $D$ segments called strands. A propagator is a $D$ stranded line and as usual connects vertices. 
The variation of the action $S^{\kin}$  under infinitesimal $U(N_a)$ transformation  is given by
\bea
\delta^{(a)} [S^{\kin}]=-i\sum_{p_1,\cdots,p_D}\Big[M\Big(\vp[\bar B\bvp]^{(a)}-[B\vp]^{(a)}\bvp\Big)
\Big]_{12\cdots D}
\eea
where $B$ is the infinitesimal  Hermitian operator corresponding to the generator of unitary group  $U(N_a)$ i.e.
\bea
U_{pp'}^{(a)}=\delta_{pp'}^{(a)}+iB_{pp'}^{(a)}
+O(B^2),\quad
\bar{U}_{pp'}^{(a)}=\delta_{pp'}^{(a)}-i\bar{B}_{pp'}^{(a)}+O(\bar B^2),
\eea
with $\bar{B}_{pp'}^{(a)}={B}_{p'p}^{(a)}$. Consider now the  theory defined with external source $F[\vp,\bvp;\eta,\bar\eta]$ as
\bea
F[\eta,\bar\eta]=\sum_{p_1,\cdots,p_D} \bar{\vp}_{12\cdots D}\eta_{12\cdots D}+\bar{\eta}_{12\cdots D}\vp_{12\cdots D}.
\eea 
 The partition function of the model is re-expressed as
\bea
\mathcal Z[\eta,\bar\eta]=\int\,d\vp d\bvp\, e^{-S[\vp,\bvp]+F[\vp,\bvp;\eta,\bar\eta]}.
\eea
Under $U(N_a)$ infinitesimal transformation 
\bea
\delta^{(a)} [F]=i\sum_{p_1,\cdots,p_D}\Big[\bar\eta [B\vp]^{(a)}-[\bar B\bvp]^{(a)}\eta\Big]_{12\cdots D}.
\eea
 Let $\delta^{(\otimes)}$ be the total variation under the action of the group element $U^{(1)}\otimes U^{(2)}\otimes \cdots \otimes U^{(D)}\in \mathcal U_{\otimes}^{N_D}$. Then we get the following proposition
\begin{proposition} \label{propcol}The kinetic term of the action \eqref{actiontensor}, i.e.  $S^{\kin}$ and  $F$ are respectively  transformed linearly as
\bea
\delta^{(\otimes)} S^{\kin}=\sum_{a=1}^D\delta^{(a)} S^{\kin},\quad 
\delta^{(\otimes)} F=\sum_{a=1}^D\delta^{(a)} F.
\eea
Then $\delta^{(\otimes)}S=0$ if and only if  $\delta^{(a)}S=0$ for all functional quantity $S$, which depends on $\vp$, $\bvp$, $\eta$ and $\bar{\eta}$.
\end{proposition}
We assume that $N_i=N,\,\,i=1,2 \cdots D$, and we  take the interaction terms such that there are invariant under the transformation $U^{(a)}$ i.e. 
$
\delta^{(a)} S^{\inter}=0.
$ This is the new input in TGFT's: the $U(N_a)$ tensor invariance must be the one defining the interaction \cite{Rivasseau:2013uca}.    
 Note that the measure $d\vp d\bvp$ is also invariant under 
$U^{(a)}$. The variation of the partition function can be performed for $a=1$  and the results for all value of $a\in\{1,2,\cdots,D\}$ may be deduced using  proposition \eqref{propcol}. We write
\bea
\frac{\delta^{(1)}\ln \mathcal Z[\eta,\bar\eta]}{\delta B_{p_mp_n}}
&=&\frac{1}{\mathcal Z[\eta,\bar\eta]}\int\,d\vp d\bvp\,\Big\{i\sum_{p_2,\cdots,p_D}\Big(M_{n\,2\cdots D}\vp_{n\,2\cdots D}\bvp_{ m\,2\cdots D }-M_{m\,2\cdots D}\bvp_{m\,2\cdots D}\vp_{ n\,2\cdots D}\Big)
\cr
&+&
i\sum_{p_2,\cdots,p_D}\Big(\bar\eta_{m\,2\cdots D} \vp_{n\,2\cdots D}-\bvp_{m\,2\cdots D}\eta_{n\,2\cdots D}\Big)\Big\}e^{-S[\vp,\bvp]+F[\vp,\bvp;\eta,\bar\eta]}=0.
\eea
Now take $\partial_{\bar\eta}\partial_{\eta}$ of the above expression,
we get only the connected components of the correlation functions as
\bea
&&\sum_{[p]}\Big(M_{m\,2\cdots D}-M_{n\,2\cdots D}\Big)\langle\Big[\frac{\partial (\bar\eta\vp)}{\partial\bar\eta}\frac{\partial (\bvp\eta)}{\partial\eta}\Big]\vp_{n\,2\cdots D}\bvp_{ m\,2\cdots D }
\rangle_c\cr
&&=\sum_{[p]}\langle\frac{\partial( \bar{\eta}_{m\,2\cdots D}\vp_{n\,2\cdots D})}{\partial \bar\eta}\Big[\frac{\partial (\bvp\eta)}{\partial \eta}\Big]-\frac{\partial(\bvp_{m\,2\cdots D}\eta_{n\,2\cdots D})}{\partial\eta}\Big[\frac{\partial( \bar{\eta}\vp)}{\partial\bar\eta}\Big]\rangle_c,
\eea
which can be simply written as
\bea\label{gros}
\sum_{[p]}\Big(M_m-M_n\Big)\langle\Big[\frac{\partial (\bar\eta\vp)}{\partial\bar\eta}\frac{\partial (\bvp\eta)}{\partial\eta}\Big]\vp_{n}\bvp_{ m}
\rangle_c\cr
=\sum_{[p]}\langle\frac{\partial(\bar{\eta}_m\vp_n)}{\partial\bar\eta}\frac{\partial (\bvp\eta)}{\partial \eta}\rangle_c-\sum_{[p]}\langle\frac{\partial (\bvp_m\eta_n)}{\partial\eta}\frac{\partial( \bar{\eta}\vp)}{\partial\bar\eta}\rangle_c.
\eea
Note that the equation \eqref{gros} is valid for all positions indices $a=1,2,\cdots,D$. Let us also remark   that for $m=n$ the left hand side (lhs) of the equation \eqref{gros}  vanishes. In the double derivative $\partial_{\bar\eta} \partial_\eta$, we fix the indices such that $\bar\eta_{[\alpha]} \eta_{[\beta]}$. Then comes  the following proposition:
\begin{proposition}
 For  index $a=1$ (corresponding to  $U^{(1)}$),   we get the Ward-Takahashi identity
\bea
&&\sum_{p_2,\cdots,p_D}\big(M_{m2\cdots D}-M_{n2\cdots D}\big)\langle \vp_{[\alpha]}\bvp_{[\beta]}\vp_{n2\cdots D}\bvp_{ m2\cdots D}\rangle_c\cr
&&=\delta_{m\alpha_1}\langle
\vp_{n\alpha_2\cdots\alpha_D}\bvp_{\beta_1\cdots\beta_D}
\rangle_c-\delta_{n\beta_1}\langle\bvp_{m\beta_2\cdots\beta_D}
\vp_{\alpha_1\cdots\alpha_D}\rangle_c,
\eea
which can be re-expressed for arbitrary position $a$ taking any value in $\{1,2,\cdots, D\}$ as
\bea\label{Ward1}
\big(M_m-M_n\big)\langle [\vp_m\bvp_n]\vp_{n}\bvp_{ m}\rangle_c=\langle
\vp_n\bvp_n\rangle_c-\langle\bvp_m\vp_m\rangle_c,\quad [\vp_m\bvp_n]=\sum_{p_2,\cdots,p_D} \vp_{n2\cdots D}\bvp_{ m2\cdots D}.
\eea
\end{proposition}

We  emphasize that the position taken by the indices $m$ and $n$ in the relation \eqref{Ward1} are the position of the momentum index $p_a$ used  in the transformation $U^{(a)}$. In conclusion, there are exactly $D$ Ward-Takahashi identities for the 
rank $D$ TGFT's associated with this type
of invariance. Note that the Ward-Takahashi identities for Boulatov model can be found  in reference \cite{BenGeloun:2011xu}. The result obtained therein radically differs from the present identities found in \eqref{Ward1}.      Furthermore, we mention that  we are not considering the
TGFT with gauge invariance condition on the fields like in the works \cite{Carrozza:2013wda}\cite{Carrozza:2012uv}. We consider here  the simplest the  TGFT as treated in \cite{BenGeloun:2011rc}\cite{Geloun:2012fq}. Most of  the result of this work might be
extended to this different framework with not much work since only the propagator will be modified. Thus, one expects  similar
Ward identities in that gauge invariant framework .

%%%%%%%%%%%%%%%%%%%%%%%%%%%

%%%%%%%%%%%%%%%%%%%%%%%%%%%

\section{Two-point functions of rank $3$ TGFT}\label{sec3}
In this section we consider the just renormalizable rank $3$ TGFT on compact $U(1)$ group, addressed firstly in \cite{BenGeloun:2012pu}. 
The  rank $3$ tensor field is defined by $\varphi:U(1)^3\longrightarrow\mathbb{C}$,  and we expand in Fourier modes as
\beq
\varphi(g_1, g_2, g_3)=\sum_{p_j\in\mathbb{Z}}\varphi_{123}e^{ip_1\theta_1}e^{ip_2\theta_2}e^{ip_3\theta_3},\quad \theta_i\in [0,2\pi).
\eeq
We write as usual $ \varphi_{123}:= \varphi_{p_1p_2p_3}$.
The renormalizable $3D$ tensor model is defined by the action $S_{3D}$, in which the  kinetic term take's the form
\beq
S^{\kin}_{3D}=\sum_{[p]}  \bar{\varphi}_{123}\,C^{-1}_{123}\,\varphi_{123},
\eeq
where  $C_{123}$ is the propagator.  
We  write  the resulting action  for the bare quantities which involves the bare mass $m_{bar}$ and  the three  wave functions renormalizations
 $Z_{\rho=1,2,3}$, each of which is associated with a strand index  $a=1,2,3$.  The field strength
can be modified as follows:
\beq
\varphi \longrightarrow \left(Z_1Z_2Z_3\right)^{\frac16} \varphi=Z^{1/2}\vp,\quad Z_{\rho} = 1 -   \partial_{b_{\rho}} 
\Gamma_{b_1b_2b_3} \Big|_{b_{1,2,3} = 0},
\eeq
where $\Gamma_{b_1b_2b_3}$ is  the self-energy or one particle irreducible (1PI) two-point functions. Then,   the renormalized propagator takes the form
\beq
C_{abc}=Z^{-1}(|a|+|b|+|c|+m^2)^{-1},\quad a,b,c\in\mathbb{Z}.
\eeq
$m$ is the renormalized mass parameter. 
The interaction of the model is defined by the three contributions $V_1$, $V_2$, and $V_3$ expressed in momentum space as
\bea
S^{\inter}_{3D}&=&\lambda_1Z^2
\sum_{\stackrel{1,2,3}{1',2',3'}}\vp_{123}\bvp_{321'}\vp_{1'2'3'}
\bvp_{3'2'1}+\lambda_2 Z^2\sum_{\stackrel{1,2,3}{1',2',3'}}\vp_{123}\bvp_{32'1}\vp_{1'2'3'}
\bvp_{3'21'}
\cr
&+&\lambda_3 Z^2\sum_{\stackrel{1,2,3}{1',2'3'}}\vp_{123}\bvp_{3'21}\vp_{1'2'3'}
\bvp_{32'1'}
=\lambda_1V_1+\lambda_2V_2+\lambda_3V_3,
\eea
and  are represented in the figure \ref{fig:Vertex4}.

\begin{figure}[htbp]
\begin{center}
 $V_1$
\includegraphics[scale=0.10]{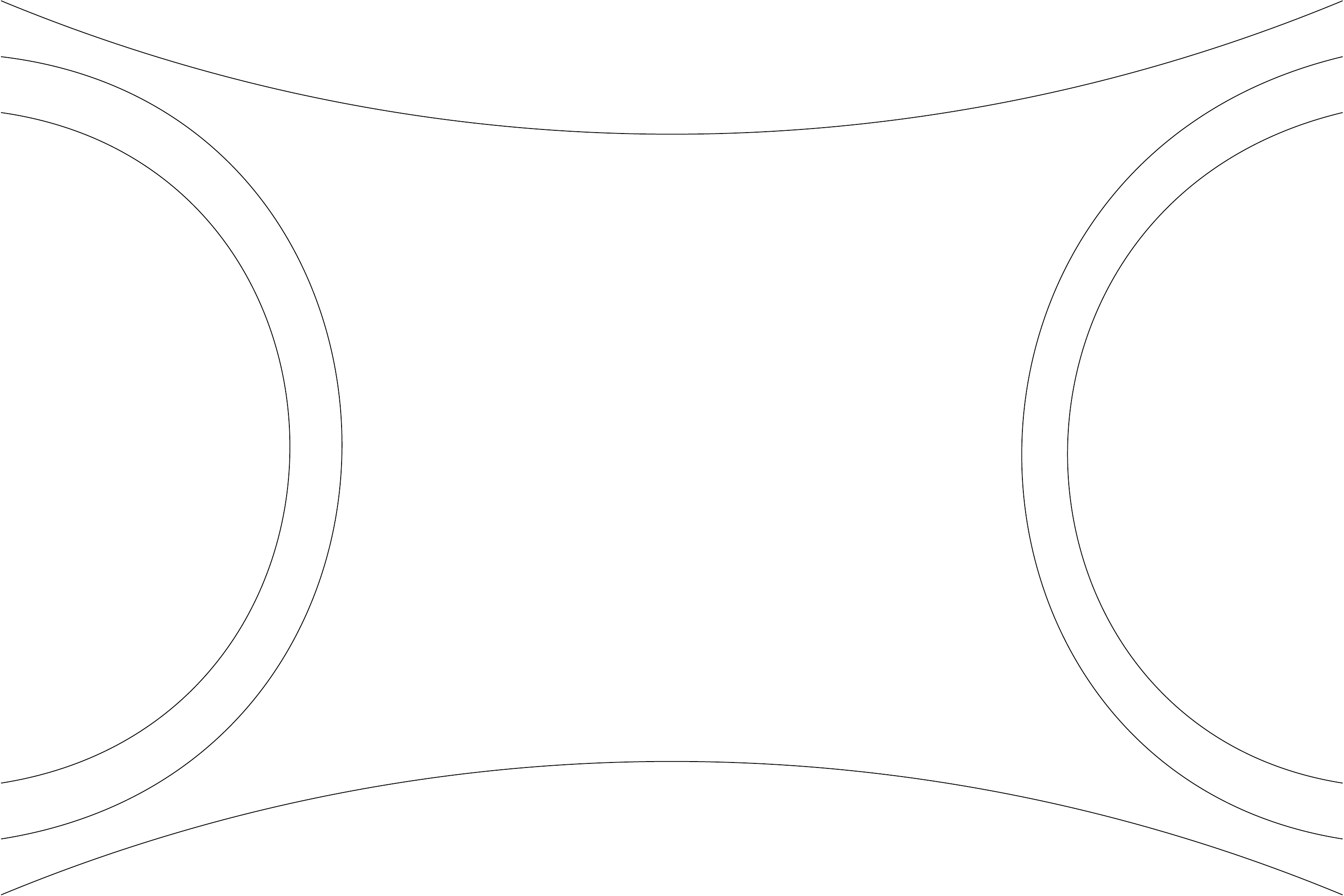}\hspace{0.5cm}
 $V_2$
\includegraphics[scale=0.10]{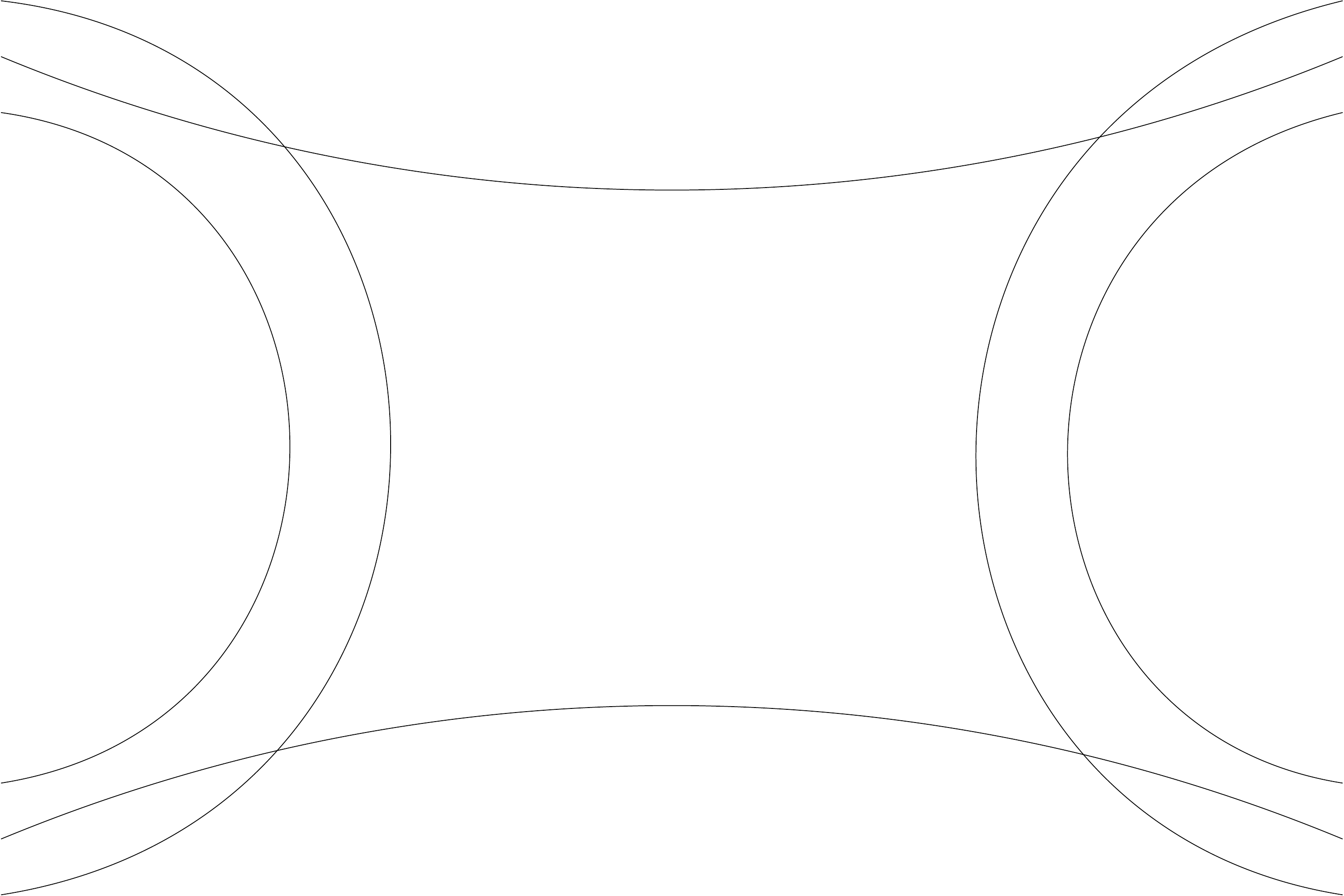}\hspace{0.5cm}
 $V_3$
\includegraphics[scale=0.10]{ward-45.pdf}\hspace{0.5cm}
\end{center}
 \caption{The vertices  of rank 3  tensor model}
  \label{fig:Vertex4} 
\end{figure}
\begin{figure}[htbp]
\begin{center}
\includegraphics[scale=0.25]{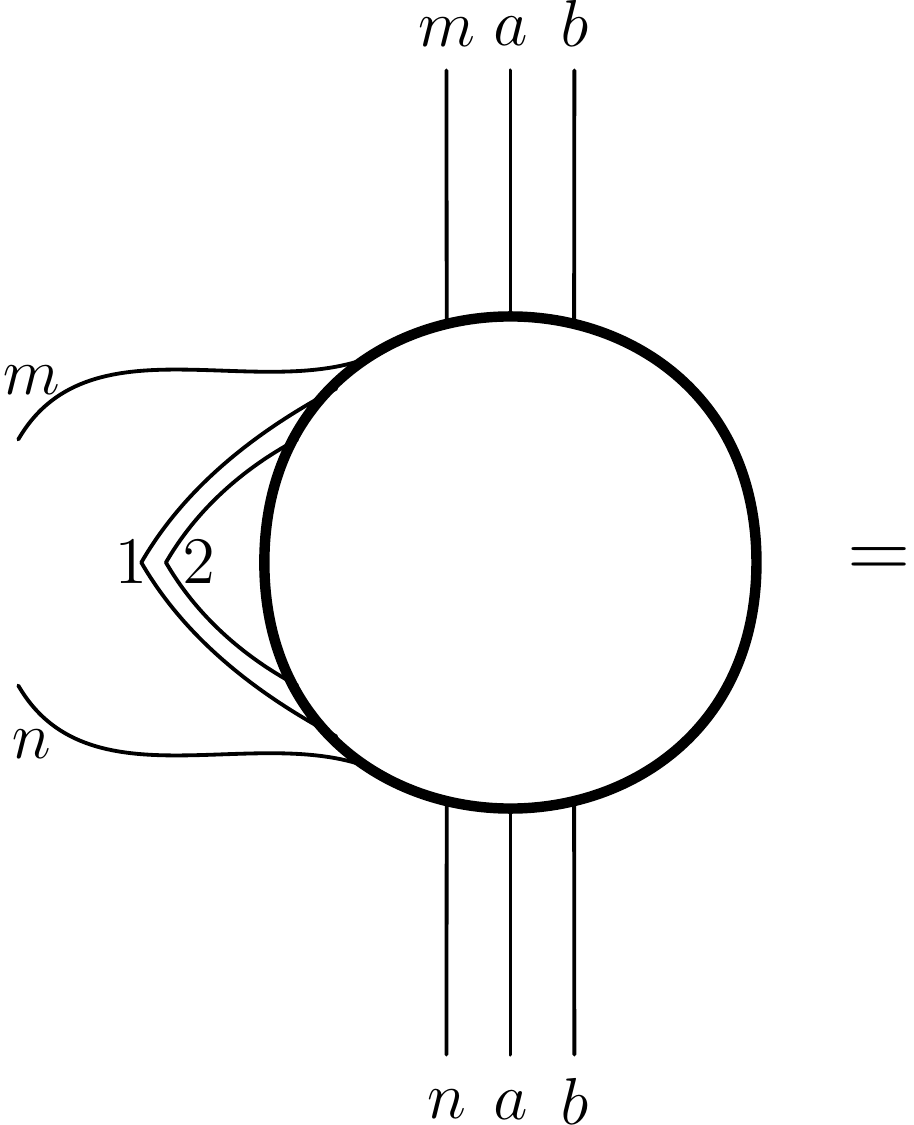}
\includegraphics[scale=0.25]{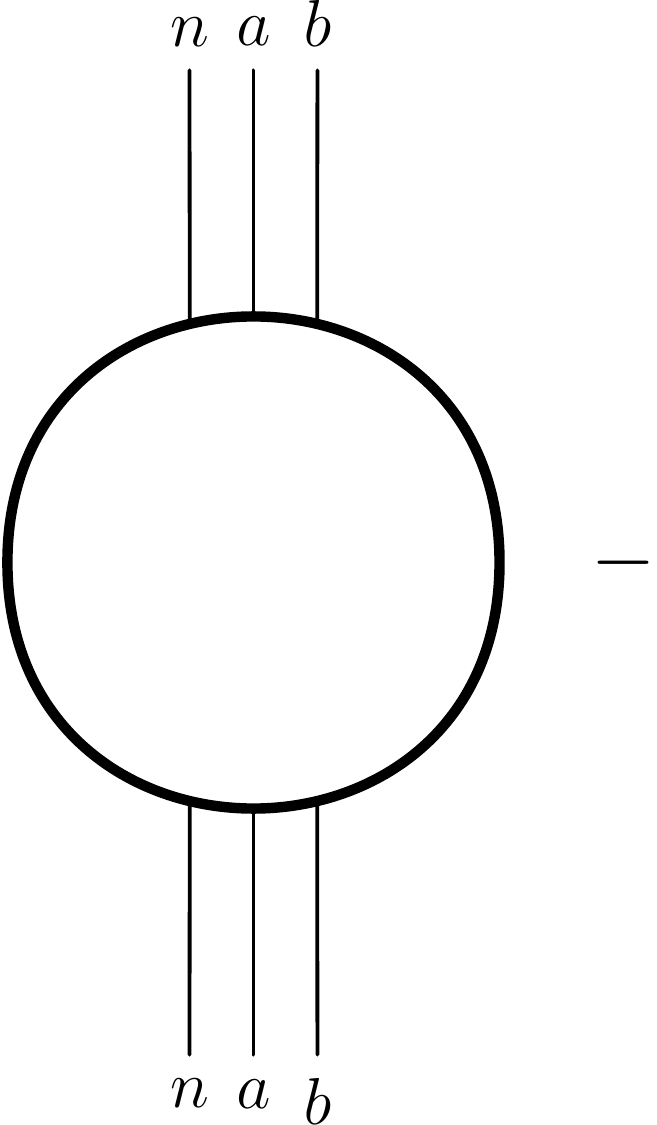}
\includegraphics[scale=0.25]{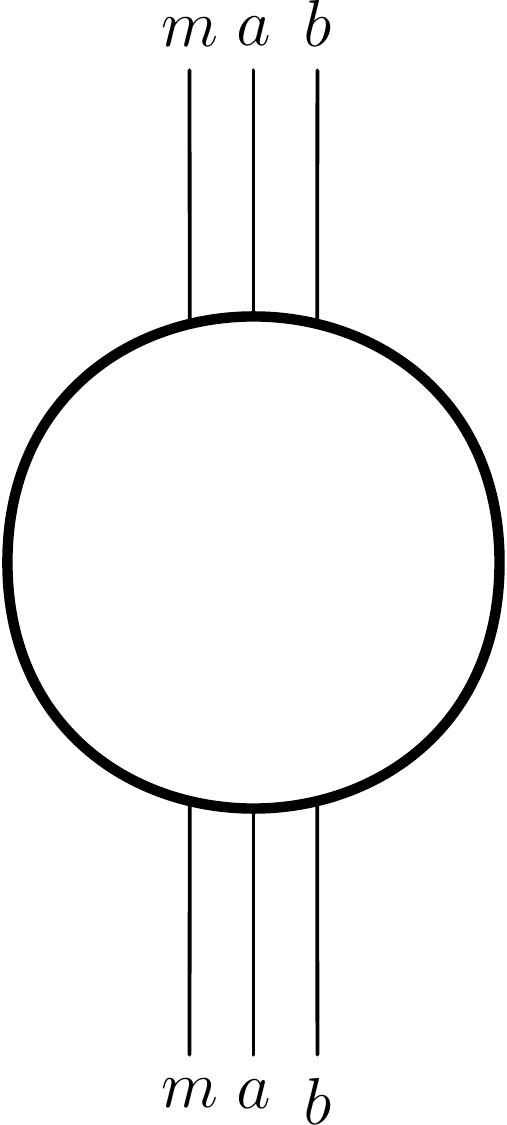}\hspace{1cm}\\
\includegraphics[scale=0.25]{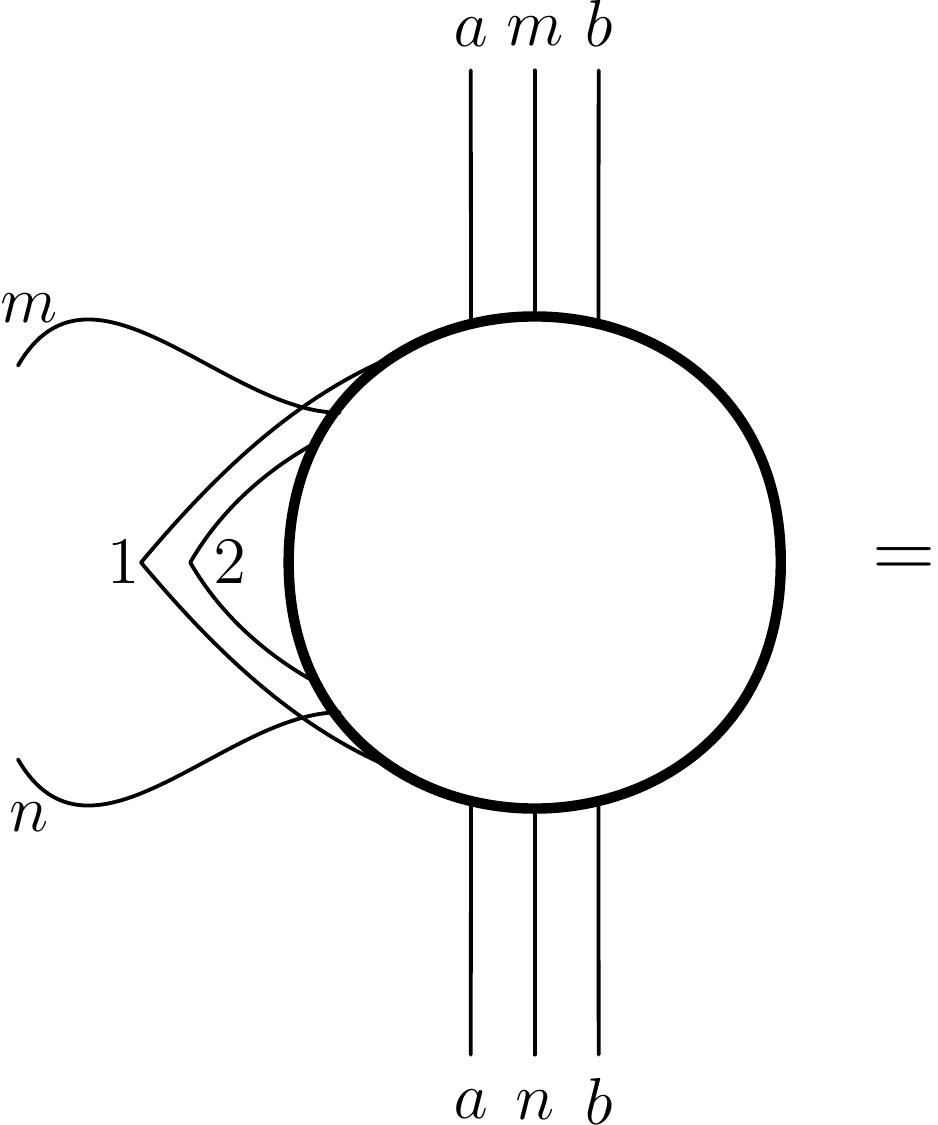}
\includegraphics[scale=0.25]{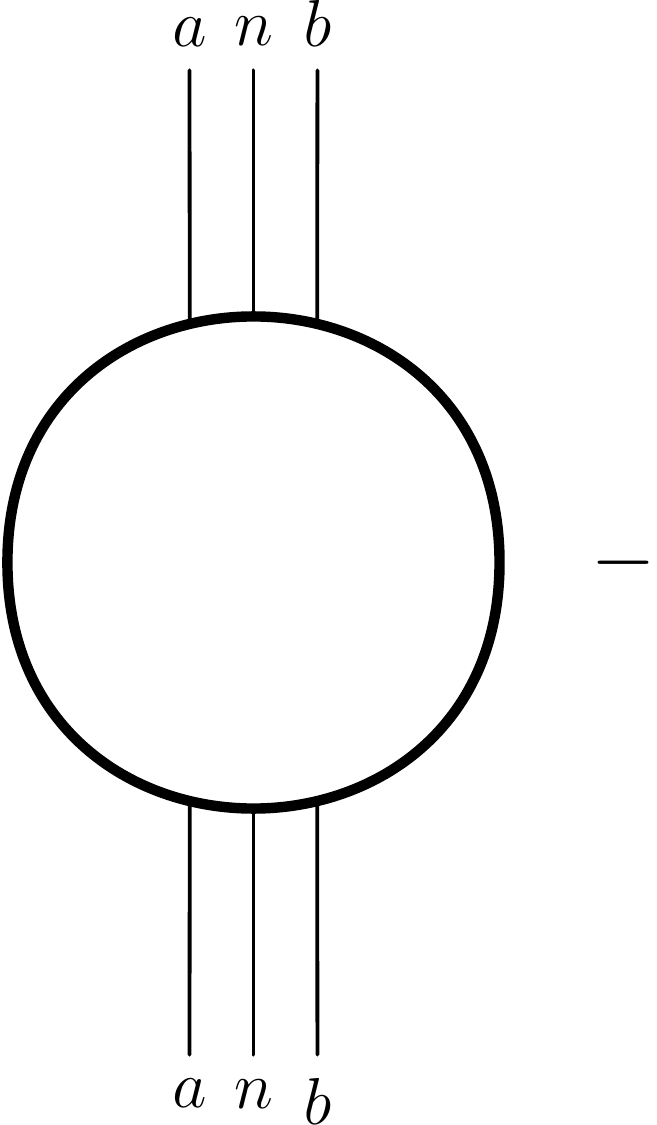}
\includegraphics[scale=0.25]{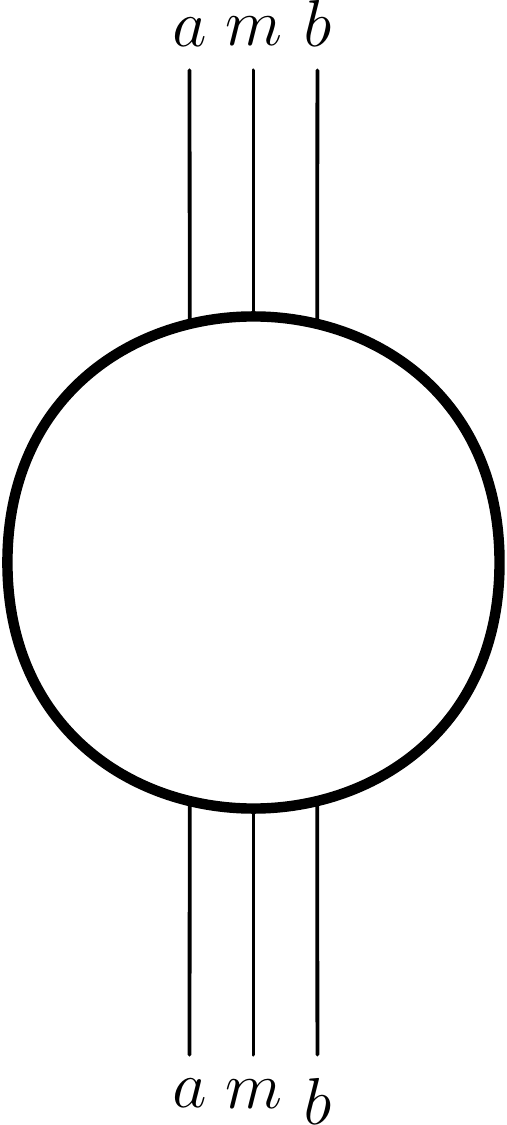}\\
\includegraphics[scale=0.25]{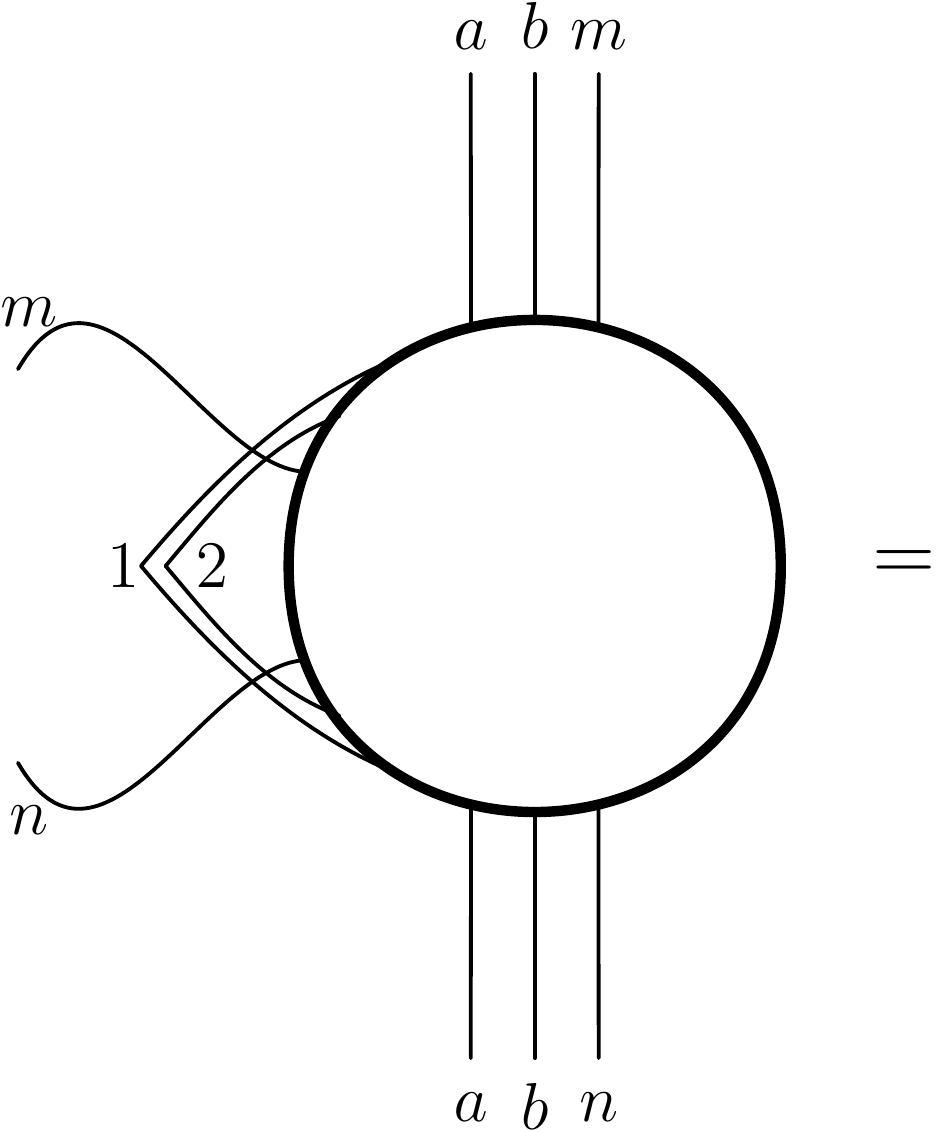}
\includegraphics[scale=0.25]{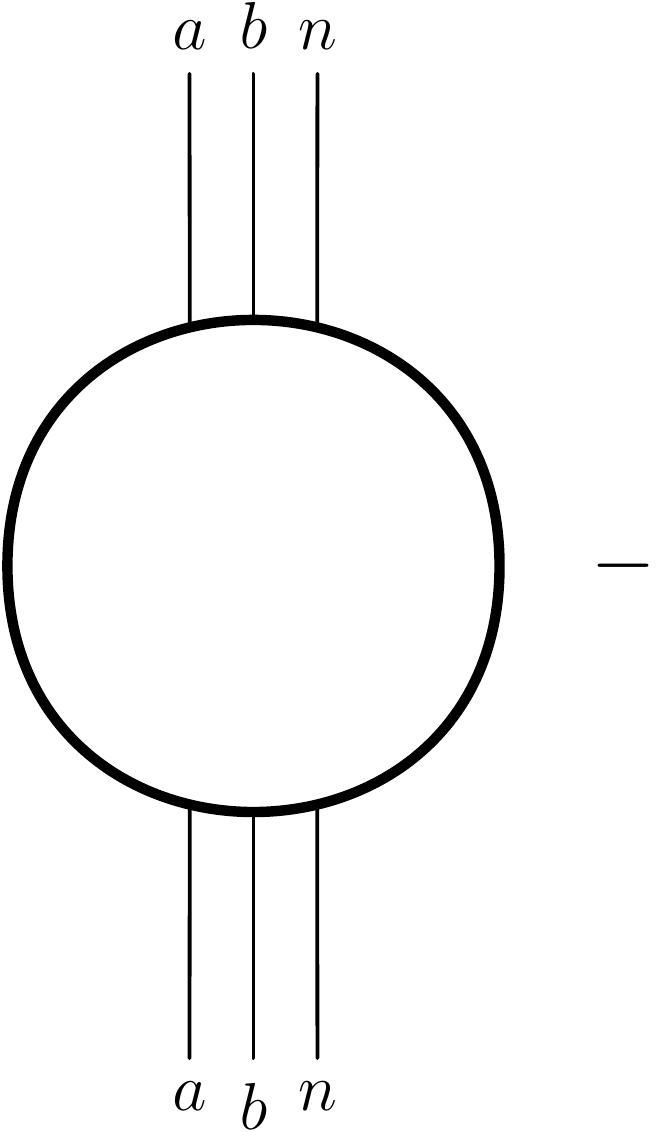}
\includegraphics[scale=0.25]{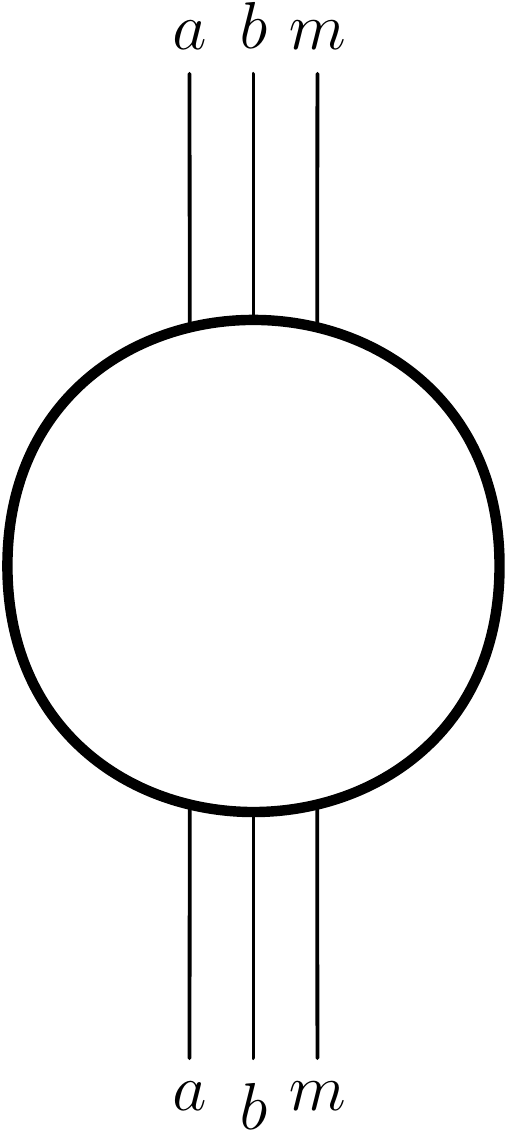}
\end{center}
 \caption{Ward-Takahashi identities}
  \label{fig:Ward} 
\end{figure}

The Ward-Takahashi identities \eqref{Ward1} 
 now find the form after reducing some constraints
\bea\label{Ward123}
\sum_{p_2,p_3}\big(M_{m23}-M_{n23}\big)\langle \vp_{m23}\bvp_{n23}\vp_{nab}\bvp_{ mab}\rangle_c=\langle
\vp_{nab}\bvp_{nab}\rangle_c-\langle\bvp_{mab}
\vp_{mab}\rangle_c\\
\label{Ward231}
\sum_{p_1,p_3}\big(M_{1m3}-M_{1n3}\big)\langle \vp_{1m3}\bvp_{1n3}\vp_{anb}\bvp_{ amb}\rangle_c=\langle
\vp_{anb}\bvp_{anb}\rangle_c-\langle\bvp_{amb}
\vp_{amb}\rangle_c
\\
\label{Ward312}
\sum_{p_1,p_2}\big(M_{12m}-M_{12n}\big)\langle \vp_{12m}\bvp_{12n}\vp_{abn}\bvp_{ abm}\rangle_c=\langle
\vp_{abn}\bvp_{abn}\rangle_c-\langle\bvp_{abm}
\vp_{abm}\rangle_c
\eea
with $M_{abc}=C_{abc}^{-1}$. Graphically the equations \eqref{Ward123}, \eqref{Ward231} and \eqref{Ward312} are given in figure \ref{fig:Ward}. Let  $G^{ins}_{[mn]ab}$ be  the two-point functions with insertion $(2,3)$ i.e.
\bea
G^{ins}_{[mn]ab}=\sum_{p_2,p_3}\langle \vp_{m23}\bvp_{n23}\vp_{nab}\bvp_{ mab}\rangle_c.
\eea
The rest of this section is devoted to find pertubatively, the exact value of  the  renormalizable two- and four-point  functions. We will use the Schwinger-Dyson equation, and  then combine it with Ward-Takahashi identities to yield the closed equation that satisfies the connected two- and four-point functions.   
The  Schwinger-Dyson  equation is represented graphically in figure  \ref{fig:Dyson}. In this figure the quantities $T^\rho_{abc}$ and 
$\Sigma^\rho_{abc}$ for $\rho=1,2,3,$ are given  in the figures \ref{fig:A} and \ref{fig:B}.
\begin{figure}[htbp]
\begin{center}
\includegraphics[scale=0.45]{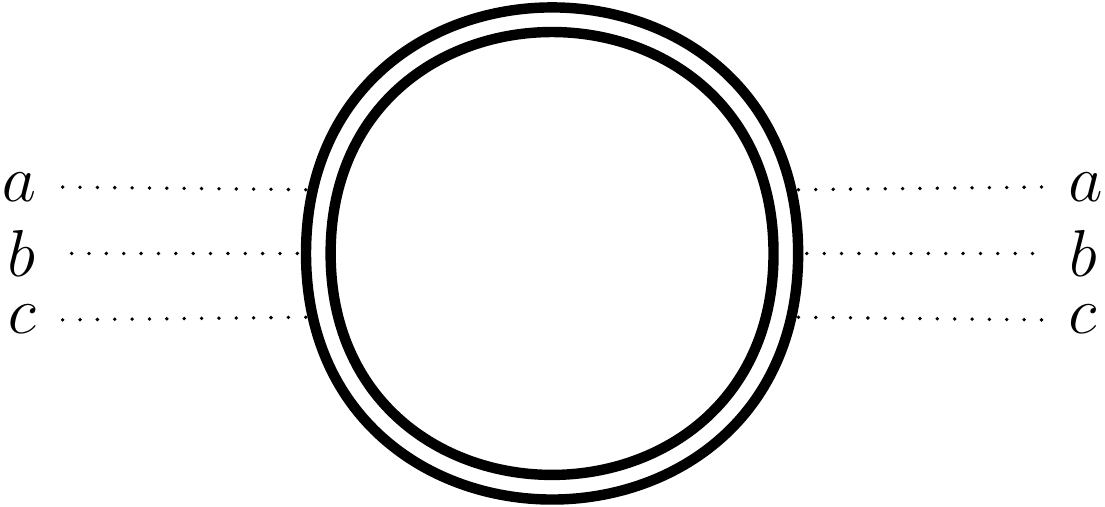}\hspace{0.2cm}
\put(10,30){$=\,\sum_{\rho=1}^3\Big(T^\rho_{abc}+\Sigma_{abc}^\rho\Big)$}
\put(-190,30){$\Gamma_{abc}=$}
\hspace{2.4cm}
\end{center}
 \caption{Schwinger-Dyson equation for   1PI two-point functions}
  \label{fig:Dyson} 
\end{figure}
\begin{figure}[htbp]
\begin{center}
\includegraphics[scale=0.32]{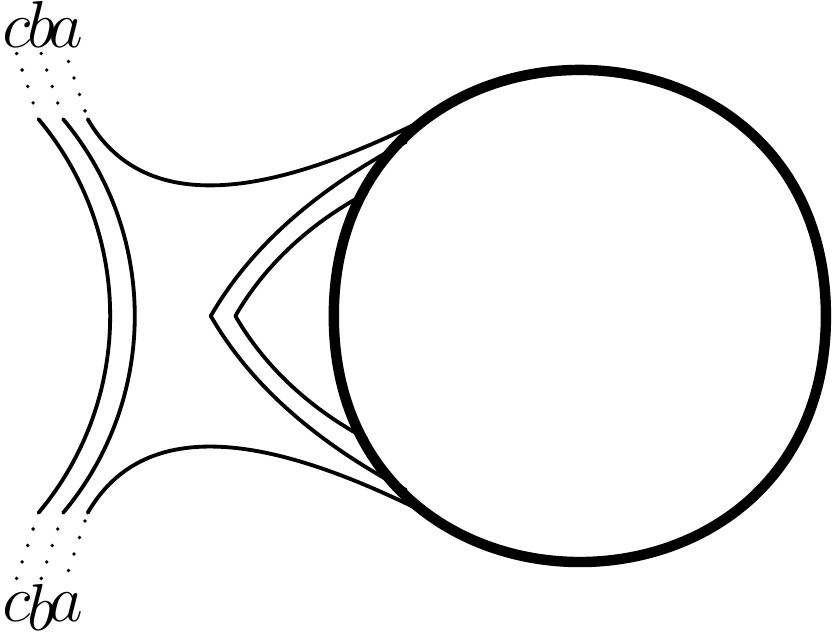}
\put(-110,26){$T^1_{abc}=$}
\hspace{2.4cm}
\includegraphics[scale=0.32]{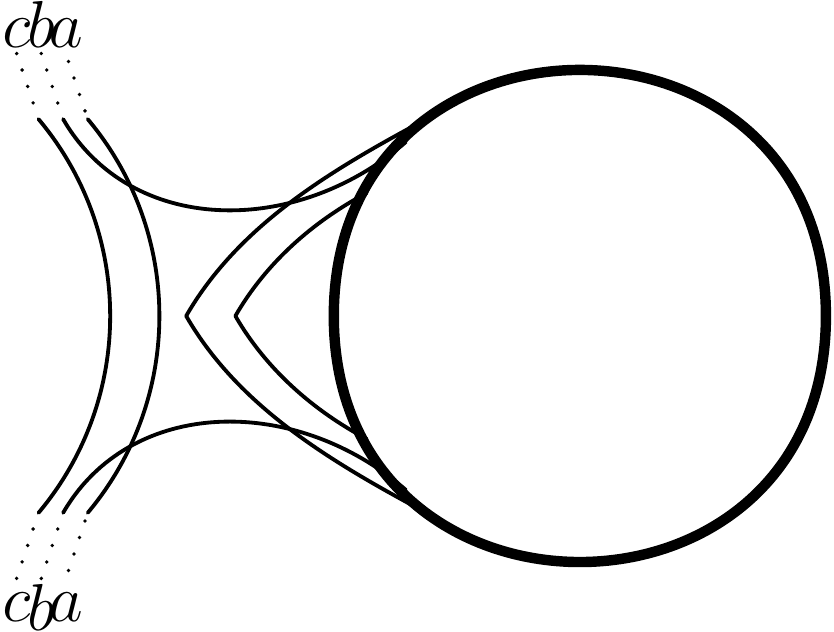}
\put(-110,26){$T^2_{abc}=$}
\hspace{2.4cm}
\includegraphics[scale=0.32]{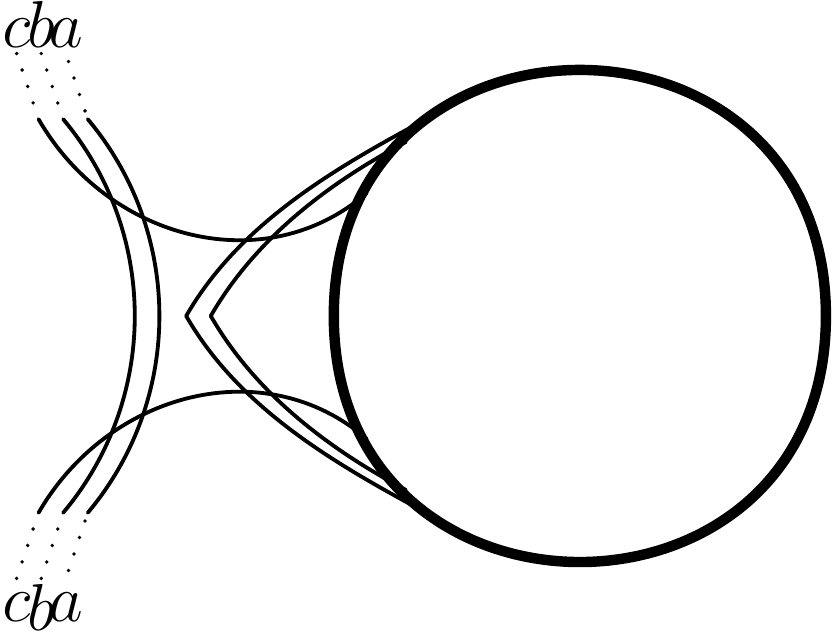}
\put(-110,26){$T^3_{abc}=$}
\end{center}
 \caption{}
  \label{fig:A} 
\end{figure}

\begin{figure}[htbp]
\begin{center}
\includegraphics[scale=0.32]{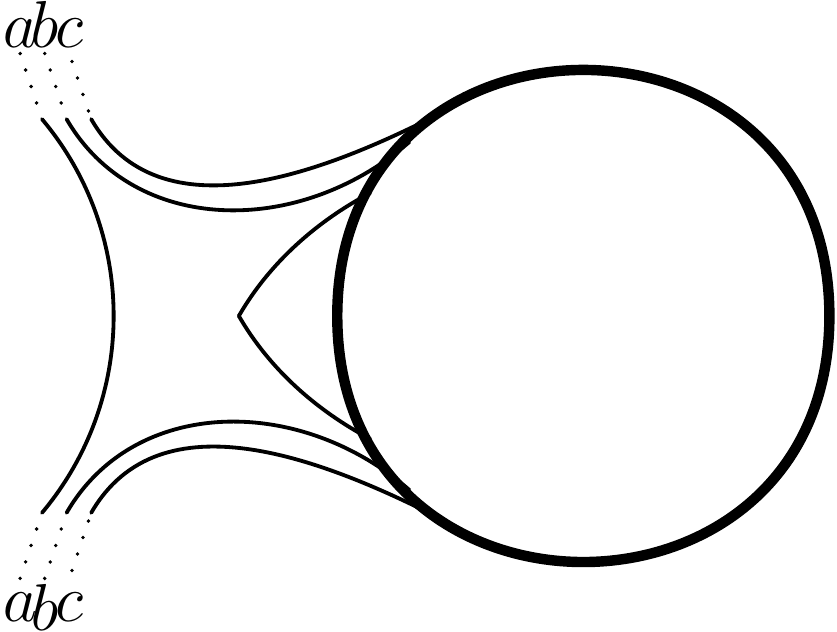}
\hspace{1.2cm}
\includegraphics[scale=0.32]{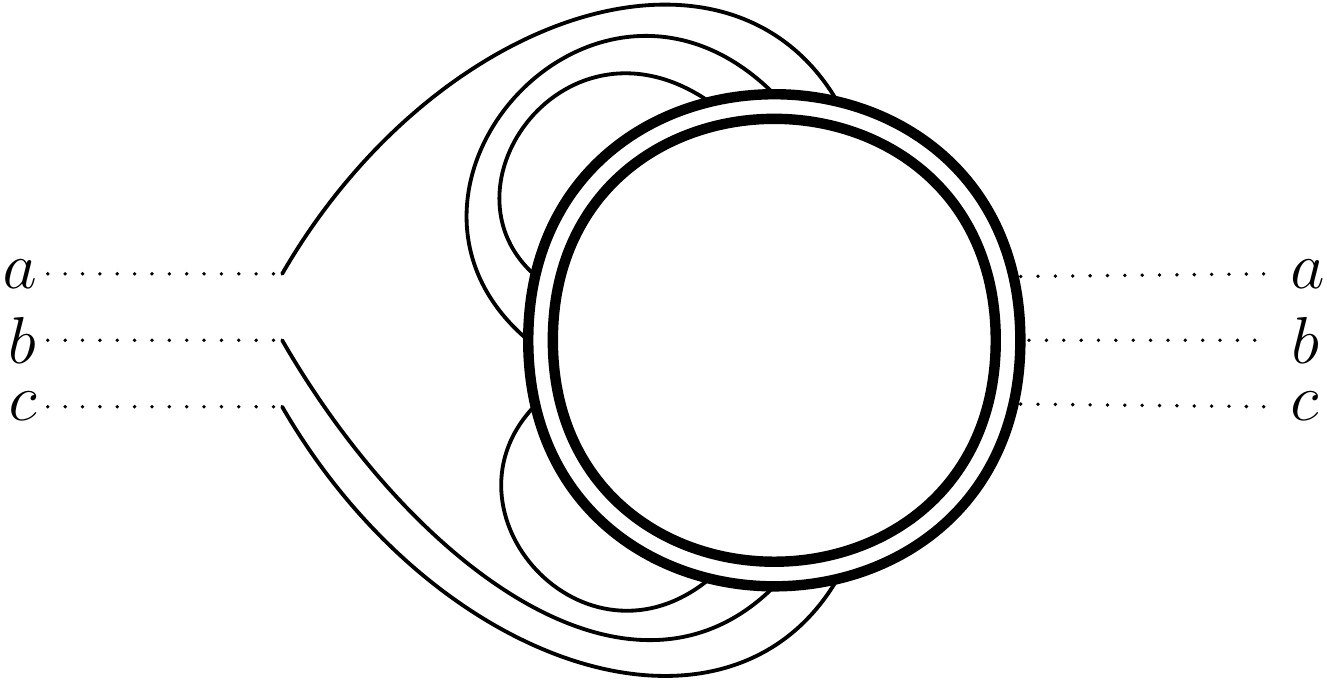}
\put(-280,25){$\Sigma^1_{abc}=$}
\put(-142,25){$+$}
\\
\includegraphics[scale=0.32]{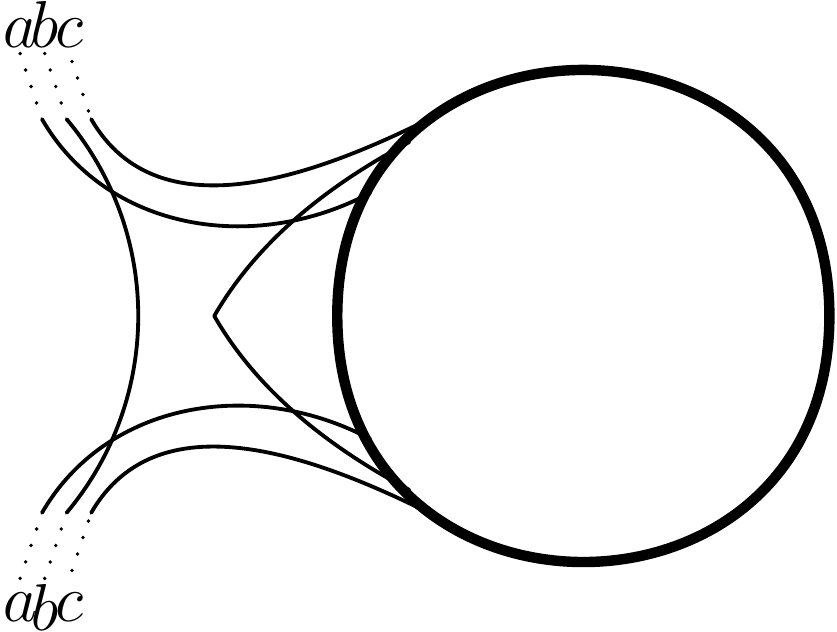}
\hspace{1.2cm}
\includegraphics[scale=0.32]{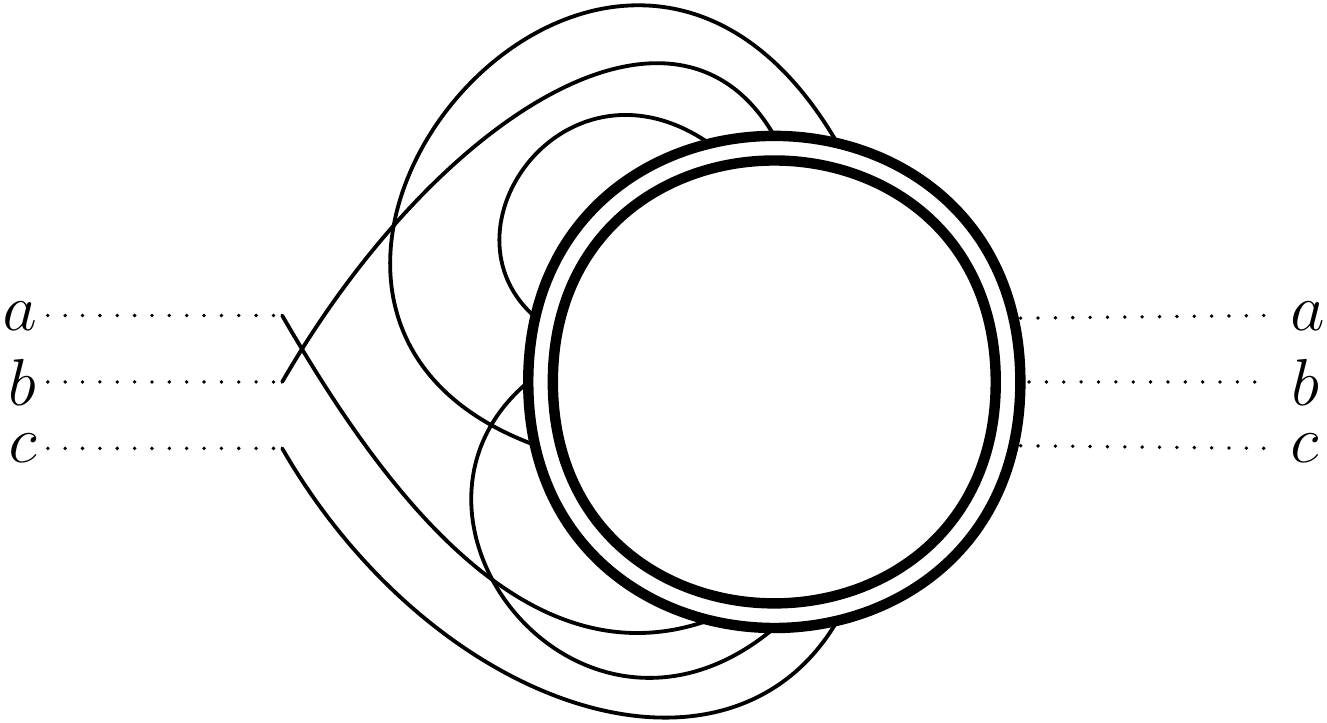}
\put(-280,25){$\Sigma^2_{abc}=$}
\put(-142,25){$+$}\\
\includegraphics[scale=0.32]{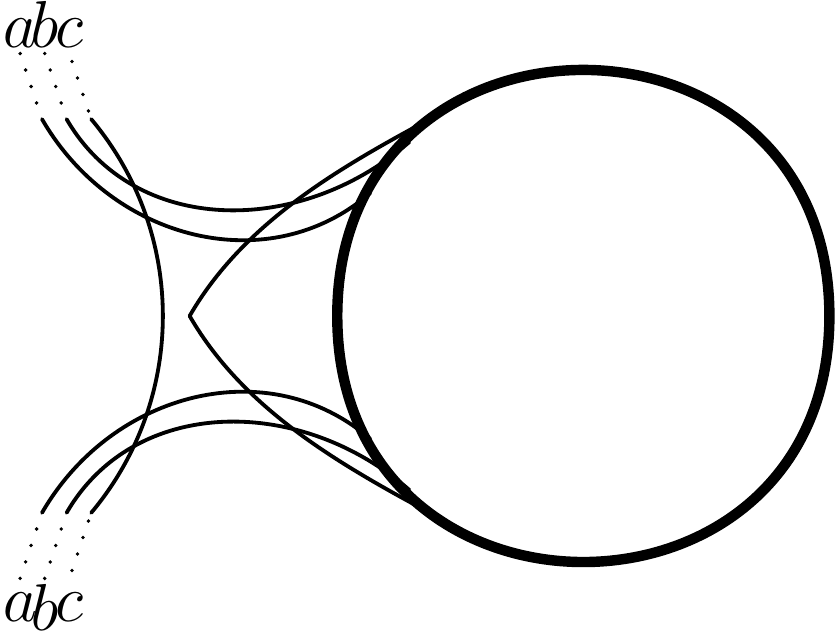}
\hspace{1.2cm}
\includegraphics[scale=0.32]{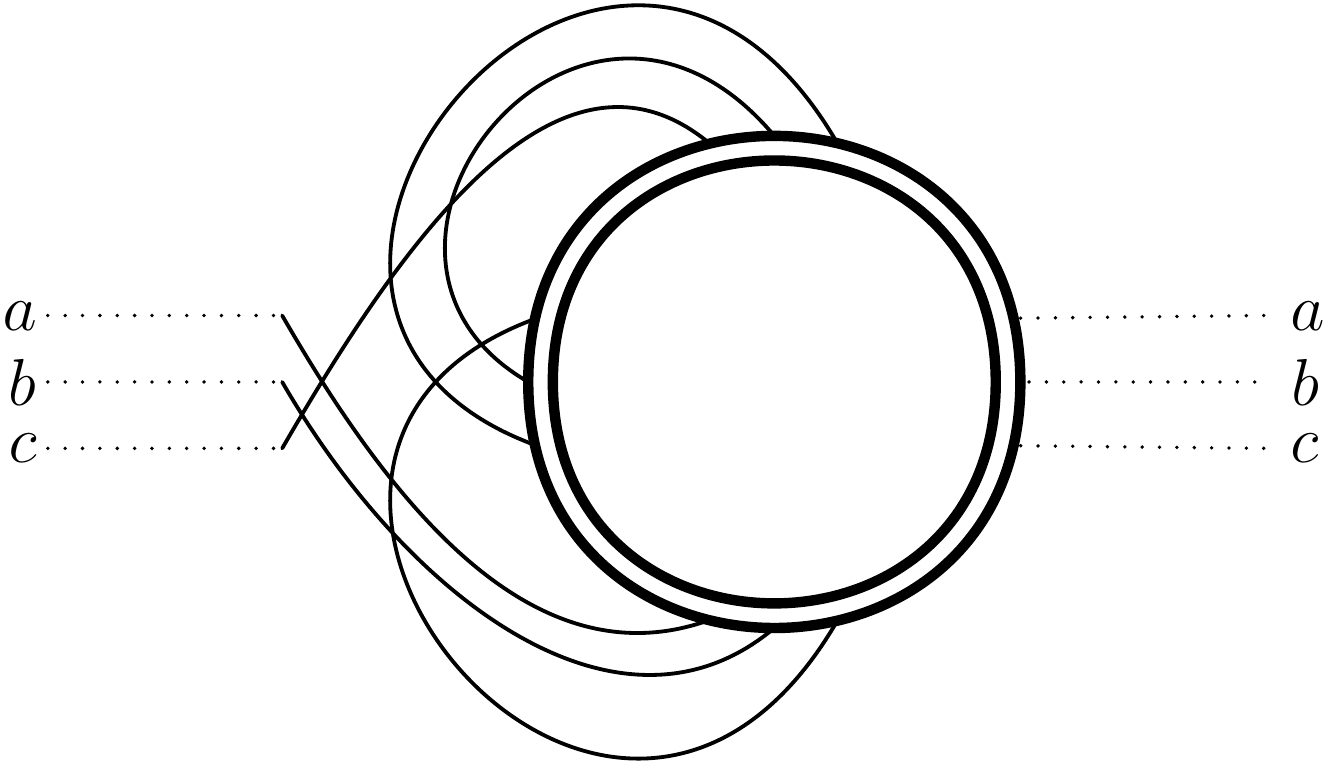}
\put(-280,25){$\Sigma^3_{abc}=$}
\put(-142,25){$+$}
\end{center}
 \caption{}
  \label{fig:B} 
\end{figure}
In the figure \ref{fig:Dyson}   the quantity $\Gamma_{abc}$ is the
 self-energy  or 1PI two-point functions that expresses as
\bea\label{rares}
\Gamma_{abc}=\sum_{\rho=1}^3\Gamma_{abc}^\rho,\quad \mbox{where}\quad \Gamma_{abc}^\rho=T^\rho_{abc}+\Sigma^\rho_{abc}.
\eea
Also, in figures  \ref{fig:Dyson}, \ref{fig:A} and \ref{fig:B}  a single circle represents a connected graph and a double circle stands for  a 1PI subgraph.
\begin{figure}[htbp]
\begin{center}
\includegraphics[scale=0.3]{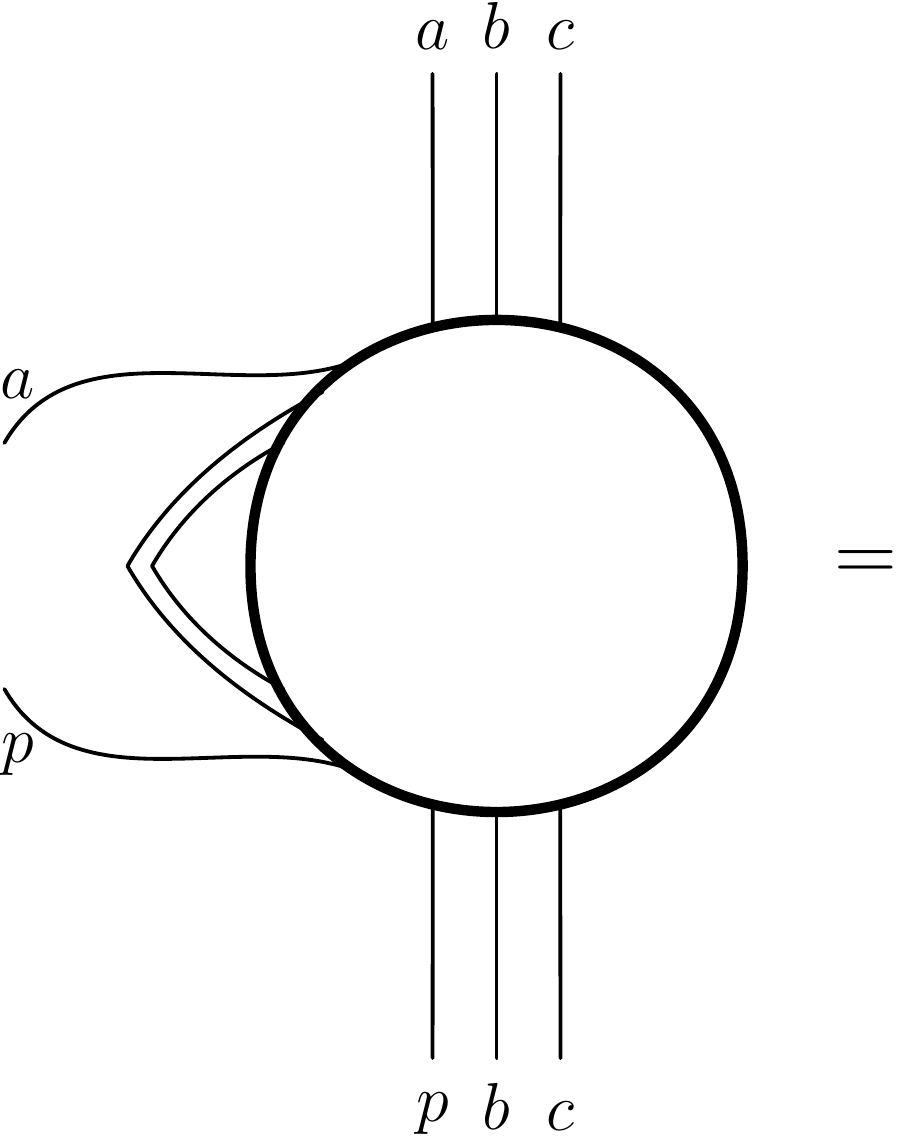}\hspace{0.2cm}
\includegraphics[scale=0.4]{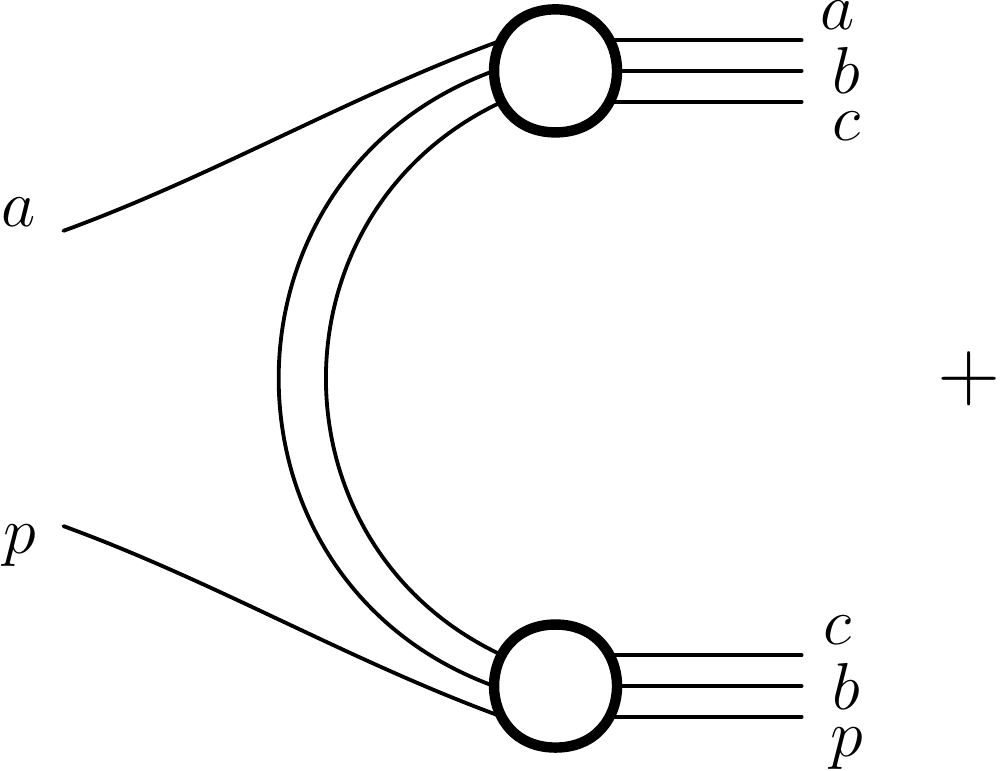}\hspace{0.2cm}
\includegraphics[scale=0.3]{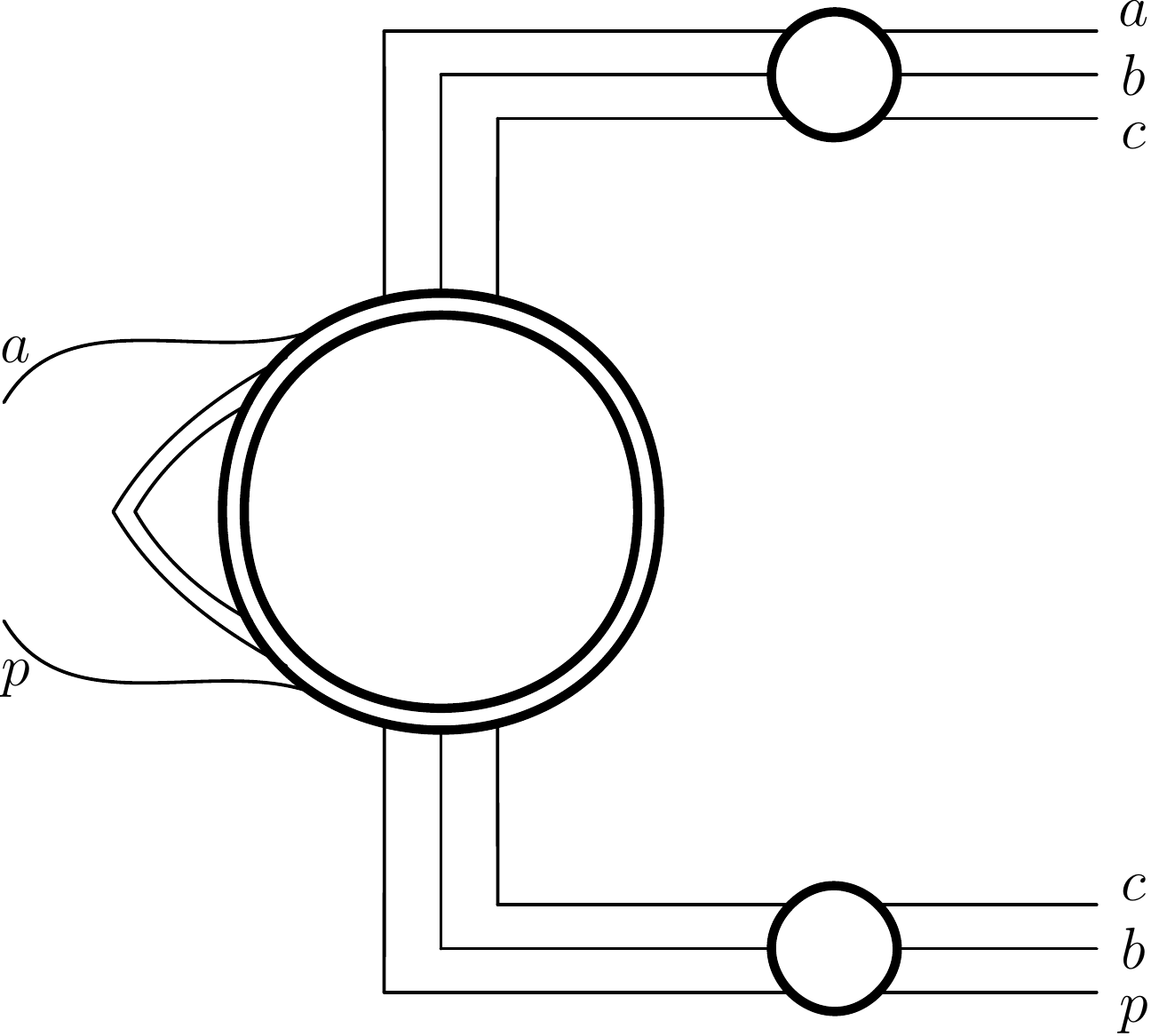}\hspace{0.2cm}
\end{center}
 \caption{Decomposition of the two-point functions with insertion: Case where $\rho=1$}
  \label{fig:Decomp1} 
\end{figure}
Let us consider now the decomposition given in figure \ref{fig:Decomp1}. The  lhs of this equation collects all connected graphs having the vertex insertion. Cutting this vertex out one gets a four-point functions, but the four-point functions can either be disconnected (first graph on the right hand side (rhs)), or connected (second graph on the rhs). The connected four-point functions must somewhere have a 1PI four-point functions as its core and then full connected two-point functions attached to its four legs. 
Now, multiplying this equation by $G_{abc}^{-1}$ means on the rhs to remove in the first graph the upper (bc)-branch attached to the insertion vertex and in the second graph the $(abc)$-branch attached to the 1PI four-point functions. If one now sums over $p$ and uses the fact that the newly created vertex is $\lambda_1 Z^2$ one gets precisely the function $\Sigma^\rho_{abc}$.
Then the equation \eqref{rares} can be written  explicitly using the decomposition of figure \ref{fig:Decomp1} as
\bea\label{1}
&&\Sigma^1_{abc}=Z^2\lambda_1\sum_p G_{abc}^{-1}G_{[ap]bc}^{ins}, \quad T^1_{abc}=Z^2\lambda_1\sum_{p,q}G_{apq}.
\eea
In the same manner we can obtain the decomposition of figure \ref{fig:Decomp23}, 
\begin{figure}[htbp]
\begin{center}
\includegraphics[scale=0.3]{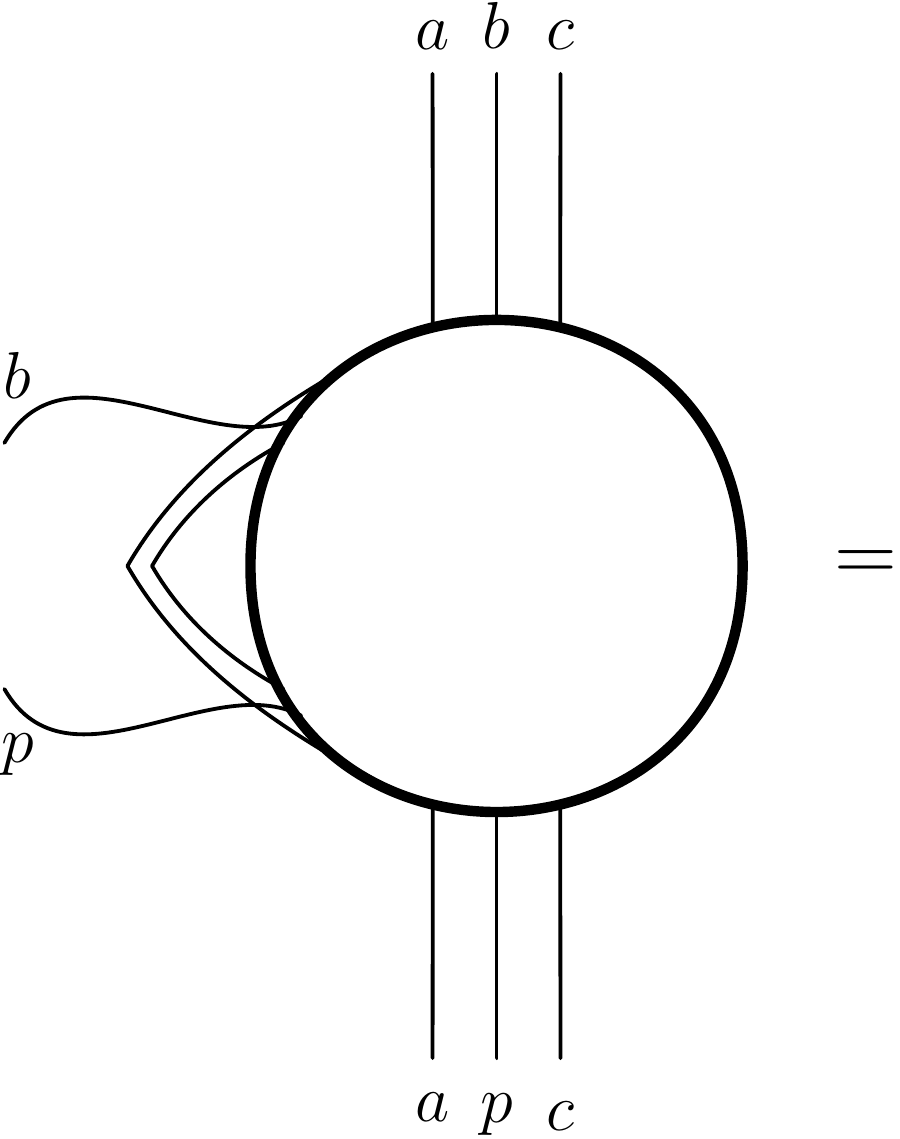}\hspace{0.2cm}
\includegraphics[scale=0.4]{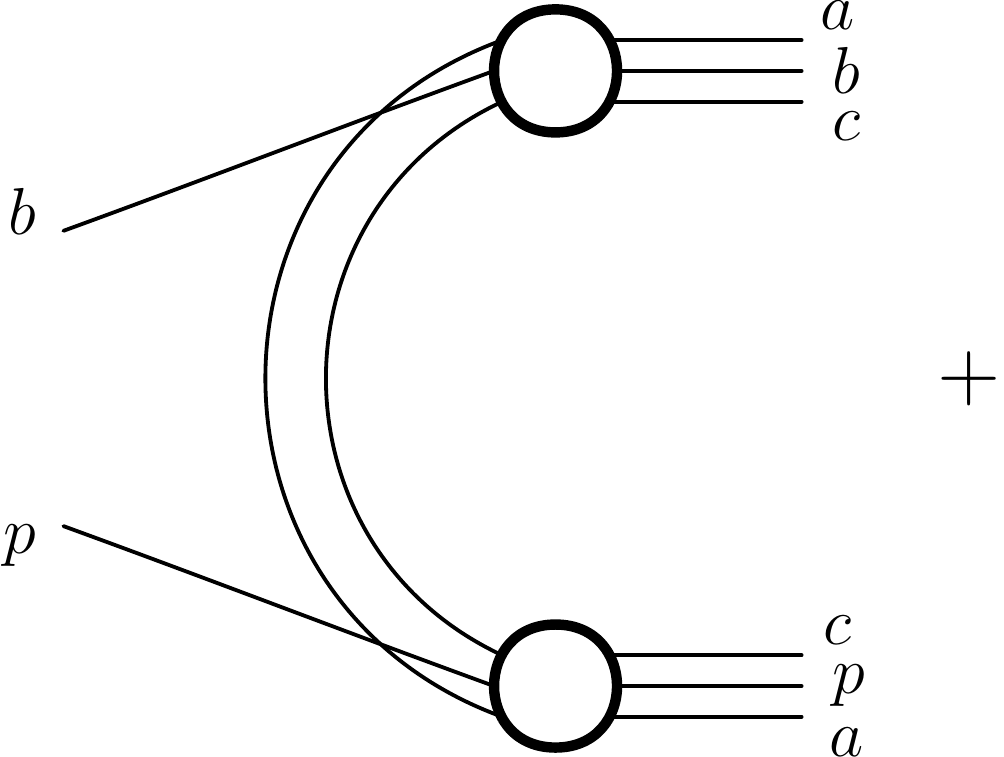}\hspace{0.2cm}
\includegraphics[scale=0.3]{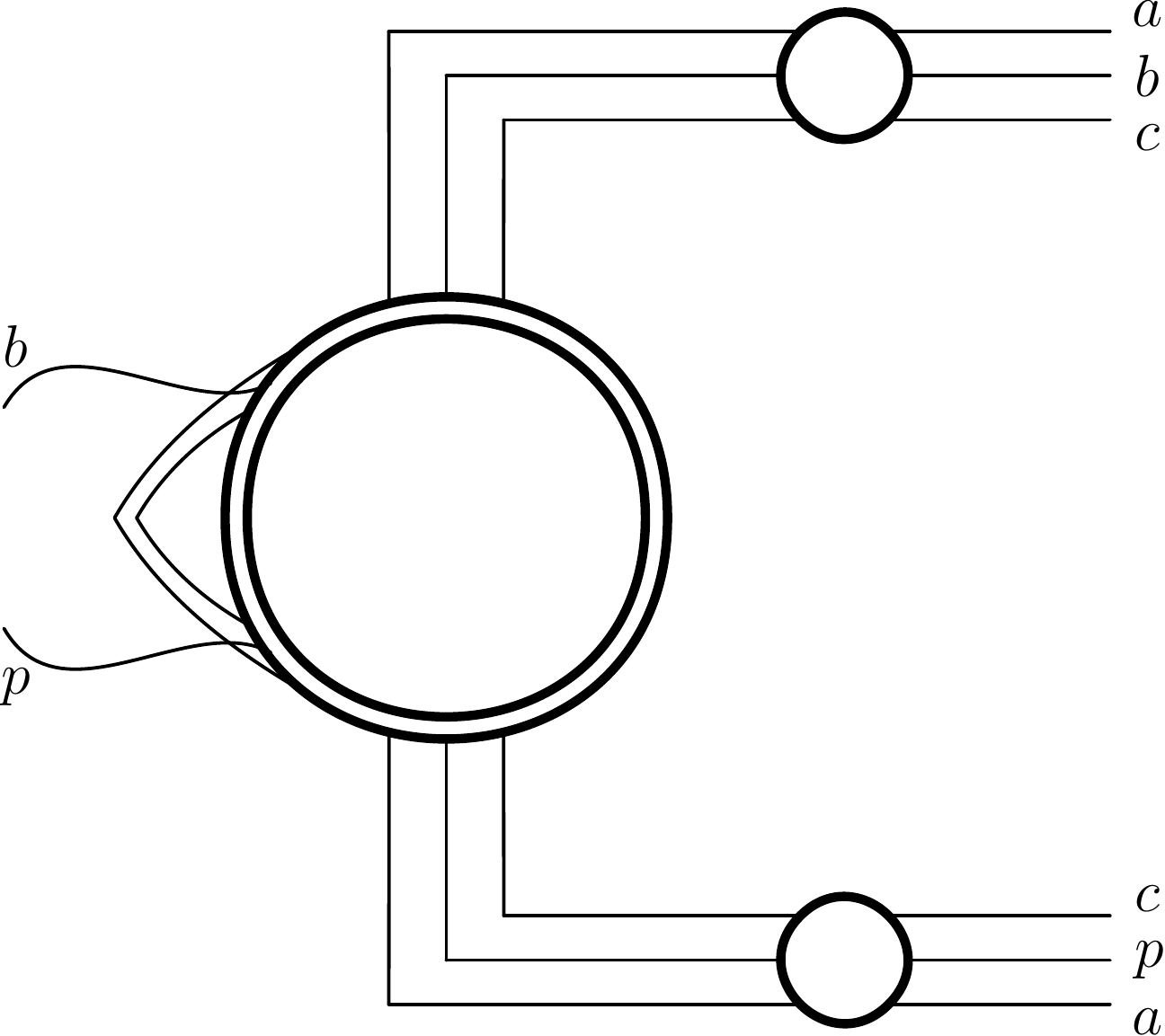}\hspace{0.2cm}\\
\includegraphics[scale=0.3]{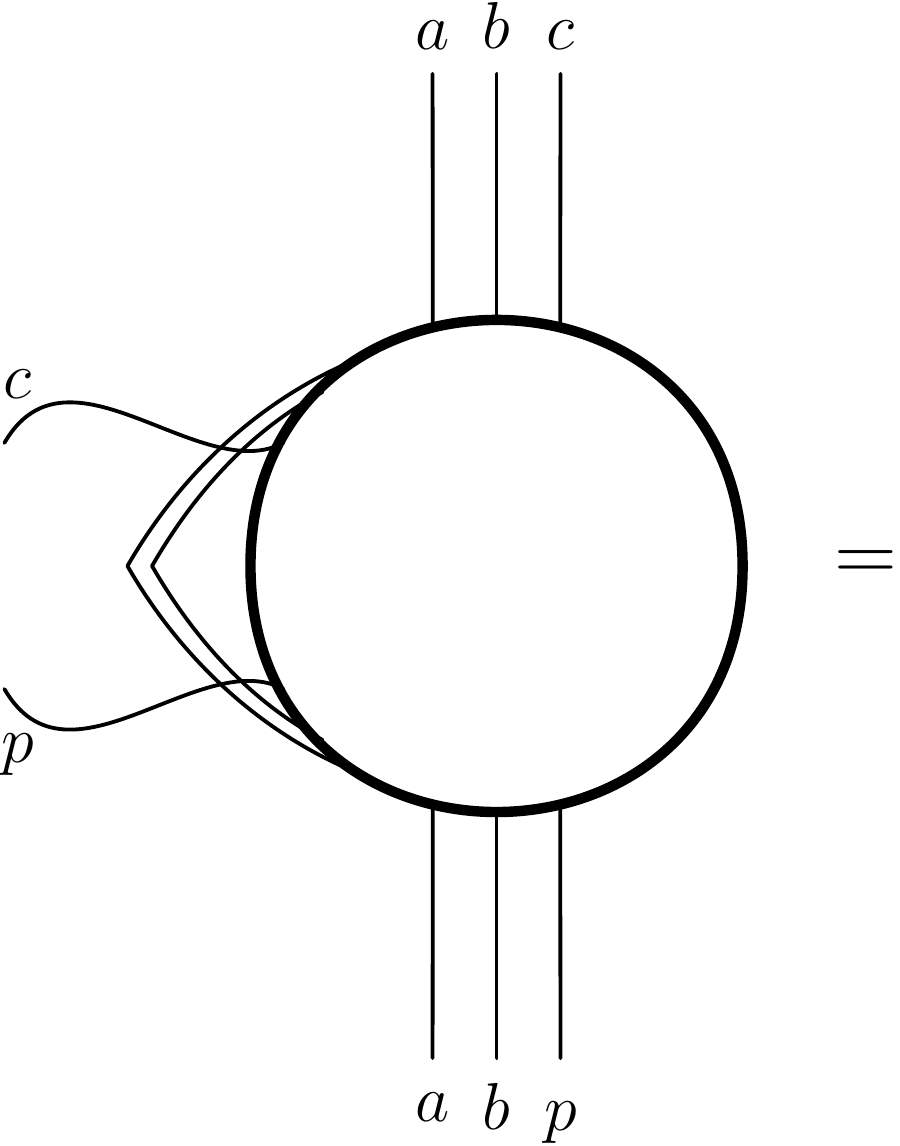}\hspace{0.2cm}
\includegraphics[scale=0.4]{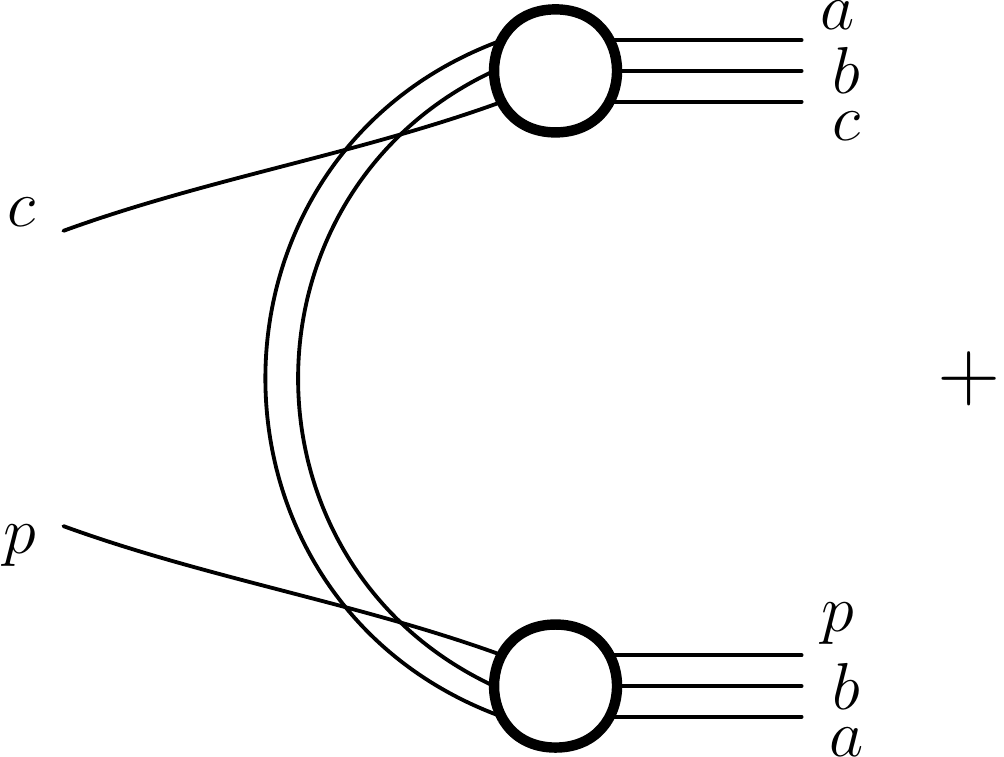}\hspace{0.2cm}
\includegraphics[scale=0.3]{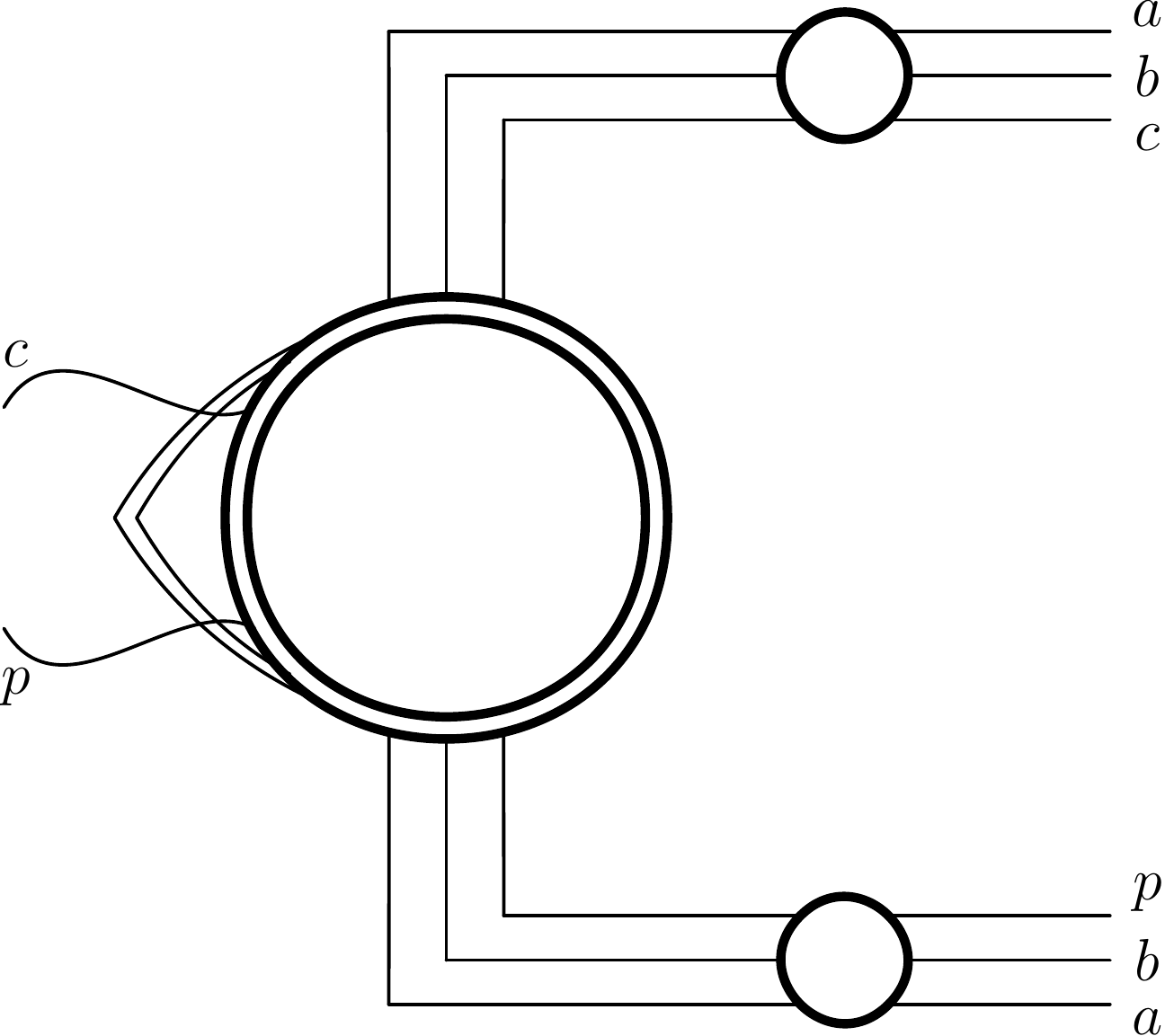}\hspace{0.2cm}
\end{center}
 \caption{Decomposition of the two-point functions with insertion: Case where $\rho=2$ and $\rho=3$}
  \label{fig:Decomp23} 
\end{figure}
which allows to obtain the relations
\beq\label{2}
\Sigma^2_{abc}=Z^2\lambda_2\sum_p G_{abc}^{-1}G_{[bp]ca}^{ins}, \quad T^2_{abc}=Z^2\lambda_2\sum_{p,q}G_{pbq}
\eeq
and
\beq
\label{3}\Sigma^3_{abc}=Z^2\lambda_3\sum_p G_{abc}^{-1}G_{[cp]ab}^{ins}, \quad T^3_{abc}=Z^2\lambda_3\sum_{p,q}G_{pqc}.
\eeq
Therefore  using the last expressions \eqref{1}, \eqref{2} and \eqref{3}, the 1PI two-point functions take the form
\bea\label{gam}
\Gamma_{abc}&=&Z^2\lambda_1
\sum_{p,q}G_{apq}+Z^2\lambda_2
\sum_{p,q}G_{pbq}+Z^2\lambda_3
\sum_{p,q}G_{pqc}\cr
&+&Z^2\lambda_1\sum_p G_{abc}^{-1}G_{[ap]bc}^{ins}+Z^2\lambda_2\sum_p G_{abc}^{-1}G_{[bp]ca}^{ins}+Z^2\lambda_3\sum_p G_{abc}^{-1}G_{[cp]ab}^{ins}\cr
&=&Z^2\lambda_1\sum_{p,q}G_{apq}+Z^2
\lambda_2\sum_{p,q}G_{pbq}+
Z^2\lambda_3\sum_{p,q}G_{pqc}+Z\lambda_1\sum_p G_{abc}^{-1}\frac{G_{abc}-G_{pbc}}{|p|-|a|}\cr
&+&Z\lambda_2\sum_p G_{abc}^{-1}\frac{G_{bca}-G_{pca}}{|p|-|b|}+Z\lambda_3\sum_p G_{abc}^{-1}\frac{G_{cab}-G_{pab}}{|p|-|c|}.
\eea
We assume now that the function $G_{abc}$ satisfy the condition 
\bea\label{ben}
G_{abc}=G_{bca}=G_{cab}
\eea
and  then, we get the following proposition:
\begin{proposition}\label{symnewcal}{Symmetry properties:\,\,}
The connected two-point functions $\Gamma_{abc}^2$  can be obtained using $\Gamma_{abc}^1$ and  replace respectively $a\rightarrow b$ and  $b\rightarrow c$ and $c\rightarrow a$. In the same manner $\Gamma_{abc}^3$  can be obtained using $\Gamma_{abc}^1$ and  replacing respectively $a\rightarrow c$, $b\rightarrow a$ and $c\rightarrow b$. 
\end{proposition} 
  Now 
using the relation $G_{abc}^{-1}=M_{abc}-\Gamma_{abc}$    , we  get
\bea\label{zon}
\Gamma_{abc}^1=Z^2\lambda_1\Big[\sum_{pq}\frac{1}{M_{apq}-\Gamma_{apq}}+
\sum_p\frac{1}{M_{pbc}-\Gamma_{pbc}}-\sum_p\frac{1}{M_{pbc}-\Gamma_{pbc}}
\frac{\Gamma_{abc}-\Gamma_{pbc}}{Z(|a|-|p|)}\Big],
\eea
\bea\label{zon1}
\Gamma_{abc}^2=Z^2\lambda_2\Big[\sum_{pq}\frac{1}{M_{pbq}-\Gamma_{pbq}}+
\sum_p\frac{1}{M_{pca}-\Gamma_{pca}}-\sum_p\frac{1}{M_{pca}-\Gamma_{pca}}
\frac{\Gamma_{bca}-\Gamma_{pca}}{Z(|b|-|p|)}\Big],
\eea
\bea\label{zon2}
\Gamma_{abc}^3=Z^2\lambda_3\Big[\sum_{pq}\frac{1}{M_{pqc}-\Gamma_{pqc}}+
\sum_p\frac{1}{M_{pab}-\Gamma_{pab}}-\sum_p\frac{1}{M_{pab}-\Gamma_{pab}}
\frac{\Gamma_{cab}-\Gamma_{pab}}{Z(|c|-|p|)}\Big].
\eea

For the rest of this section we consider the connected two-point functions $\Gamma_{abc}^{1}$ and finally $\Gamma_{abc}^{2}$ and $\Gamma_{abc}^{3}$ will be deduced using the proposition \eqref{symnewcal}.  Then  we  pass to renormalized quantities using the Taylor expansion as
\bea
\Gamma_{abc}^1=ZM_{abc}^{bar}-M_{abc}^{phys}+\Gamma_{abc}^{phys},\quad \Gamma_{000}^{phys}=0=
\partial\Gamma_{000}^{phys}
\eea
 such that
\beq
M_{abc}^{phys}=|a|+|b|+|c|+m^2,\quad M_{abc}^{bar}=|a|+|b|+|c|+m_{bar}^2.
\eeq
We get  after replacing the expression of $M_{abc}$, 
\bea
\Gamma_{abc}^1=(Z-1)(|a|+|b|+|c|)+Zm^2_{bar}-m^2+\Gamma_{abc}^{phys},
\eea
which expresses the relation between renormalized and bare quantities.
The equation \eqref{gam} takes the form (we set $\lambda_1=\lambda$)
\bea\label{reply}
&&Zm_{bar}^2-m^2+(Z-1)(|a|+|b|+|c|)+\Gamma_{abc}^{phys}=Z^2\lambda
\sum_{p,q}\frac{1}{|p|+|q|+|a|+m^2-\Gamma^{phys}_{pqa}}\cr
&&+Z\lambda\Big[\sum_{p}\frac{1}{|p|+|b|+|c|+m^2-\Gamma^{phys}_{pbc}}-
\frac{1}{|p|+|b|+|c|+m^2-\Gamma^{phys}_{pbc}}
\frac{\Gamma_{abc}^{phys}-\Gamma_{pbc}^{phys}}{(|a|-|p|)}\Big].
\eea
For $a=b=c=0$  the relation of mass variation after renormalization is written as
\bea\label{reply1}
Zm^2_{bar}-m^2&=&Z^2\lambda\sum_{p,q}\frac{1}{|p|+|q|+m^2-\Gamma_{pq0}^{phys}}+Z\lambda\sum_{p}
\frac{1}{|p|+m^2-\Gamma_{p00}^{phys}}\cr
&-&Z\lambda\sum_{p}
\frac{1}{|p|+m^2-\Gamma_{p00}^{phys}}\frac{
\Gamma_{p00}^{phys}}{|p|}.
\eea
Inserting the equation \eqref{reply1} in \eqref{reply}, we get the closed equation of the two-point functions of  renormalizable rank $3$ TGFT as
\begin{multline}\label{eqself}
(Z-1)(|a|+|b|+|c|)+\Gamma_{abc}^{phys}=Z^2\lambda
\sum_{p,q}\Big[\frac{1}{|p|+|q|+|a|+m^2-\Gamma^{phys}_{pqa}}-\frac{1}{|p|+|q|+m^2-\Gamma_{pq0}^{phys}}\Big]\cr
+Z\lambda\sum_{p}\Big[\frac{1}{|p|+|b|+|c|+m^2-\Gamma^{phys}_{pbc}}-
\frac{1}{|p|+|b|+|c|+m^2-\Gamma^{phys}_{pbc}}
\frac{\Gamma_{abc}^{phys}-\Gamma_{pbc}^{phys}}{(|a|-|p|)}\cr
-\frac{1}{|p|+m^2-\Gamma_{p00}^{phys}}
+
\frac{1}{|p|+m^2-\Gamma_{p00}^{phys}}\frac{
\Gamma_{p00}^{phys}}{|p|}\Big].
\end{multline}

The equation \eqref{eqself} is still very complicated  compared to an equivalent one in \cite{Grosse:2009pa}. To simplify it and get explicit solution, we pass to  the  integral transforms. The process is to  set
\beq
\sum_{p\in\mathbb{Z}}=2\int_0^\infty\,d|p|,\quad\sum_{p,q\in\mathbb{Z}}=2\int_0^\infty\,|p|d|p|.
\eeq
We also assume that  $\Gamma_{abc}=\Gamma_{|a||b||c|}$. Then we get the integral equation of \eqref{eqself} as
\begin{align}\label{kyle}
 &(Z-1)(|a|+|b|+|c|)+\Gamma_{abc}^{phys}\cr&=2Z^2\lambda
\int_0^\infty\,|p|d|p|\Big[\frac{1}{2|p|+|a|+m^2-\Gamma^{phys}_{ppa}}-\frac{1}{2|p|+m^2-\Gamma_{pp0}^{phys}}\Big]\cr
&+2Z\lambda\int_0^\infty\,d|p|\Big[\frac{1}{|p|+|b|+|c|+m^2-\Gamma^{phys}_{pbc}}-\frac{1}{|p|+m^2-\Gamma_{p00}^{phys}}\cr
&-
\frac{1}{|p|+|b|+|c|+m^2-\Gamma^{phys}_{pbc}}
\frac{\Gamma_{abc}^{phys}-
\Gamma_{pbc}^{phys}}{(|a|-|p|)}
+
\frac{1}{|p|+m^2-\Gamma_{p00}^{phys}}\frac{
\Gamma_{p00}^{phys}}{|p|}\Big]
\end{align}
with $p\in\mathbb{R^+}.$ We introduce a change of variables
\bea\label{t1}
&&|a|=m^2\frac{\alpha}{1-\alpha},\quad |b|=m^2\frac{\beta}{1-\beta},\quad |c|=m^2\frac{\gamma}{1-\gamma},\quad |p|=m^2\frac{\rho}{1-\rho},\\
&&\label{t2}
\Gamma_{abc}^{phys}=m^2\frac{
\Gamma_{\alpha\beta\gamma}}{(1-\alpha)(1-\beta)(1-\gamma)}.
\eea
We also  take  the cutoff $\Lambda$ such that
$
p_\Lambda=m^2\frac{\Lambda}{1-\Lambda}.
$
Let us now define the quantity $G_{\alpha\beta\gamma}$ as \beq\label{relationsa}
1-\alpha\beta-\alpha\gamma-\beta\gamma+
2\alpha\beta\gamma-\Gamma_{\alpha\beta\gamma}
=\frac{1-\alpha\beta-\alpha\gamma-\beta\gamma+
2\alpha\beta\gamma}{G_{\alpha\beta\gamma}}.
\eeq
Let $\mathcal{J}_{\alpha\beta\gamma}$, $\mathcal{L}_{\alpha\beta\gamma}$ and $\mathcal{K}_\alpha$ are three integrals relation  given by
\bea 
&&\mathcal{J}_{\alpha\beta\gamma}=\int_0^1 \,\frac{d\rho}{(\alpha-\rho)}\frac{G_{\rho\beta\gamma}}{(1-\beta\rho-\gamma\rho-\gamma\beta+
2\rho\gamma\beta)},
\\
&&\mathcal{L}_{\alpha\beta\gamma}=\int_0^1\,\frac{d\rho}{(1-\rho)}\frac{G_{\rho \beta\gamma}-1}{(\alpha-\rho)},
\\
&&
\mathcal{K}_\alpha=m^2\frac{
\int_0^1\,\frac{\rho d\rho}{(1-\rho)}\Big(\frac{(1-\alpha)G_{\rho\rho\alpha}}{1-\rho^2-2\alpha\rho+2\alpha\rho^2}-\frac{G_{\rho\rho 0}}{1-\rho^2}\Big)}{1+\frac{2\lambda}{m^2}\int_0^1\,d\rho\Big(\frac{G'_{\rho 00}}{\rho}+G_{\rho 00}\Big) }.
\eea
Then
we get the following theorem
\begin{theorem}\label{theo1}
The connected two-point functions $G_{\alpha\beta\gamma}$    of the renormalizable rank  $3$  TGFT on $U(1)$  satisfies the closed integral equation
\bea\label{master3}
G_{\alpha\beta\gamma}&=&1+\lambda'\Big\{\mathcal Y +\int_0^1\,d\rho\,G_{\rho 00}+(1-\alpha)(1-\beta)(1-\gamma)\mathcal J_{\alpha\beta\gamma}\cr
&&+\frac{(1-\alpha)(1-\beta)(1-\gamma)G_{\alpha\beta\gamma}}{1-\alpha\beta-\alpha\gamma-\beta\gamma+
2\alpha\beta\gamma}\Big[-\mathcal Y-\int_0^1\,d\rho\,G_{\rho 00}+\mathcal K_{\alpha}\cr
&+&\int_0^1\,d\rho\,\frac{G_{\rho \beta\gamma}-G_{\rho 00}}{1-\rho}-\int_0^1\,d\rho\,\frac{G_{\rho \beta\gamma}}{\alpha-\rho}+(1-\alpha)\mathcal L_{\alpha\beta\gamma}-\mathcal{L}_{000}\Big]\Big\}
\eea
where
\bea
\mathcal{Y}=\lim_{\epsilon\rightarrow 0}\int_0^1\,d\rho\,\frac{G_{\rho\epsilon 0}-G_{\rho 00}}{\epsilon\rho},\qquad \lambda'=\frac{2\lambda}{m^2}.
\eea
\end{theorem} 
\begin{proof}
 Using the transformations given in the equations
\eqref{t1} and \eqref{t2}, the expression  \eqref{kyle} takes the form
\bea\label{zoro}
&&(Z-1)\Big(\frac{\alpha}{1-\alpha}+\frac{\beta}{1-\beta}+\frac{\gamma}{1-\gamma}\Big)
+\frac{\Gamma_{\alpha\beta\gamma}}{(1-\alpha)(1-\beta)(1-\gamma)}\cr
&&=2Z^2\lambda
\int_0^\Lambda\,\frac{\rho d\rho}{(1-\rho)}\Big[\frac{(1-\alpha)}{1-\rho^2-2\alpha\rho+2\alpha\rho^2-
\Gamma_{\rho\rho\alpha}}-\frac{1}{1-\rho^2-\Gamma_{\rho\rho 0}}\Big]\cr
&&+\frac{2Z\lambda}{m^2} \int_0^\Lambda \,\frac{d\rho}{(1-\rho)}\Big[\frac{(1-\beta)(1-\gamma)}{1-\beta\rho-\gamma\rho-\gamma\beta+
2\rho\gamma\beta-\Gamma_{\rho\beta\gamma}}-\frac{1}{1-\Gamma_{\rho 00}}\cr
&&-\frac{1}{1-\beta\rho-\gamma\rho-\gamma\beta+
2\rho\gamma\beta-\Gamma_{\rho\beta\gamma}}
\frac{(1-\rho)\Gamma_{\alpha\beta\gamma}-(1-\alpha)
\Gamma_{\rho\beta\gamma}}{\alpha-\rho}\cr
&&+\frac{1}{1-\Gamma_{\rho 00}}\frac{\Gamma_{\rho 00}}{\rho}\Big].
\eea
Noting that $\beta$ and $\gamma$ are symmetric parameters in the equation \eqref{zoro}. This implies that $\Gamma_{\alpha\beta\gamma}=\Gamma_{\alpha\gamma\beta}$. Let us now take $\frac{\partial}{\partial\alpha}\Big|_{\alpha=\beta=\gamma=0}$ and $\frac{\partial}{\partial\beta}\Big|_{\alpha=\beta=\gamma=0}$ of the above equation. We come to the relations that satisfies the   renormalized wave function $Z$: 
\bea
Z-1=2Z^2\lambda
\int_0^\Lambda\,\frac{\rho d\rho}{(1-\rho)}\frac{(-1+2\rho-\rho^2+\Gamma'_{\rho\rho 0}+\Gamma_{\rho\rho 0})}{(1-\rho^2-\Gamma_{\rho\rho 0})^2}-\frac{2Z\lambda}{m^2} \int_0^\Lambda \,d\rho\frac{\Gamma_{\rho 00}}{\rho^2(1-\Gamma_{\rho 00})},
\eea
and
\bea\label{moimeme}
Z-1&=&\frac{2Z\lambda}{m^2} \int_0^\Lambda \,\frac{d\rho}{(1-\rho)}\Big[\frac{-1+\rho+\Gamma_{\rho 00}+\Gamma'_{\rho 00}}{(1-\Gamma_{\rho 00})^2}-\frac{(\rho+\Gamma'_{\rho 00})\Gamma_{\rho 00}}{\rho(1-\Gamma_{\rho 00})^2}- \frac{\Gamma'_{\rho 00}}{\rho(1-\Gamma_{\rho 00})}\Big],
\eea
where we take  $\Gamma'_{\rho 00}=:\frac{\partial\Gamma_{\rho\beta\gamma}}{\partial \beta}\Big|_{\beta=\gamma=0}$ or $\Gamma'_{\rho 00}=:\frac{\partial\Gamma_{\rho\beta\gamma}}{\partial \gamma}\Big|_{\beta=\gamma=0}$ and $\Gamma'_{\rho \rho 0}=:\frac{\partial\Gamma_{\rho\rho\alpha}}{\partial \alpha}\Big|_{\alpha=0}$. Now let us pass to the new function  $G_{\alpha\beta\gamma}$ given in  \eqref{relationsa}.
We find the following relations
\bea
\rho+\Gamma'_{\rho 00}=\frac{\rho}{G_{\rho 00}}+\frac{G'_{\rho 00}}{G_{\rho 00}^2},\quad\,\, 2\rho-2\rho^2+\Gamma'_{\rho\rho 0}=\frac{2\rho(1-\rho)}{G_{\rho\rho 0}}+\frac{(1-\rho^2)G'_{\rho\rho 0}}{G_{\rho\rho 0}}.
\eea
Therefore the equation  \eqref{moimeme}  reduces to
\bea\label{Z-1}
 \label{Z}Z^{-1}=1+\frac{2\lambda}{m^2}\int_0^\Lambda\,d\rho\Big[\frac{G'_{\rho 00}}{\rho}+G_{\rho 00}\Big],
\eea
 and \eqref{zoro} takes the form 
\bea\label{eqnew}
&&ZG_{\alpha\beta\gamma}-1-(Z-1)\frac{(1-\alpha)(1-\beta)(1-\gamma)}{1-\alpha\beta-\alpha\gamma-\beta\gamma+
2\alpha\beta\gamma}G_{\alpha\beta\gamma}\cr
&&=\frac{(1-\alpha)(1-\beta)(1-\gamma)}{1-\alpha\beta-\alpha\gamma-\beta\gamma+
2\alpha\beta\gamma}G_{\alpha\beta\gamma}\Big\{2Z^2\lambda
\int_0^\Lambda\,\frac{\rho d\rho}{(1-\rho)}\Big[\frac{(1-\alpha)G_{\rho\rho\alpha}}{1-\rho^2-2\alpha\rho+2\alpha\rho^2}-\frac{G_{\rho\rho 0}}{1-\rho^2}\Big]\cr
&&+\frac{2Z\lambda}{m^2} \int_0^\Lambda \,\frac{d\rho}{(1-\rho)}\Big[\frac{(1-\beta)(1-\gamma)G_{\rho\beta\gamma}}{1-\beta\rho-\gamma\rho-\gamma\beta+
2\rho\gamma\beta}-G_{\rho 00}+
\frac{(1-\alpha)(G_{\rho\beta\gamma}-1)
}{(\alpha-\rho)}+\frac{G_{\rho 00}-1}{\rho}\cr
&&-\frac{G_{\rho\beta\gamma}}{1-\beta\rho-\gamma\rho-\gamma\beta+
2\rho\gamma\beta}
\frac{(1-\rho)(1-\alpha\beta-\alpha\gamma-
\beta\gamma+2\alpha\beta\gamma)(G_{\alpha\beta\gamma}-1)}{(\alpha-\rho)G_{\alpha\beta\gamma}}\Big]\Big\}.
\eea
Inserting \eqref{Z-1} into the left hand side of \eqref{eqnew} and dividing by $Z$, one gets
\bea\label{sadnew}
G_{\alpha\beta\gamma}&=&Z^{-1}-\frac{2\lambda}{m^2}\frac{(1-\alpha)(1-\beta)(1-\gamma)}{1-\alpha\beta-\alpha\gamma-\beta\gamma+
2\alpha\beta\gamma}G_{\alpha\beta\gamma}
\int_0^\Lambda\,d\rho\Big(\frac{G'_{\rho 00}}{\rho}+G_{\rho 00}\Big)\cr
&-&\frac{2\lambda}{m^2} \int_0^\Lambda \,d\rho\frac{(1-\alpha)(1-\beta)(1-\gamma)G_{\rho\beta\gamma}}{1-\beta\rho-\gamma\rho-\gamma\beta+
2\rho\gamma\beta}\cdot
\frac{(G_{\alpha\beta\gamma}-1)}{(\alpha-\rho)}\cr
&+&\frac{(1-\alpha)(1-\beta)(1-\gamma)G_{\alpha\beta\gamma}}{1-\alpha\beta-\alpha\gamma-\beta\gamma+
2\alpha\beta\gamma}\Big
\{\frac{2\lambda}{Z^{-1}}
\int_0^\Lambda\,\frac{\rho d\rho}{(1-\rho)}\Big[\frac{(1-\alpha)G_{\rho\rho\alpha}}{1-\rho^2-2\alpha\rho+2\alpha\rho^2}-\frac{G_{\rho\rho 0}}{1-\rho^2}\Big]\cr
&+&\frac{2\lambda}{m^2} \int_0^\Lambda \,\frac{d\rho}{(1-\rho)}\Big[\frac{(1-\beta)(1-\gamma)G_{\rho\beta\gamma}}{1-\beta\rho-\gamma\rho-\gamma\beta+
2\rho\gamma\beta}-G_{\rho 00}
+
\frac{(1-\alpha)(G_{\rho\beta\gamma}-1)
}{(\alpha-\rho)}\cr
&+&\frac{G_{\rho 00}-1}{\rho}\Big]\Big\}.
\eea
Replacing \eqref{Z}  \eqref{sadnew} yields
\bea\label{lama}
G_{\alpha\beta\gamma}&=&1+\frac{2\lambda}{m^2}\Big\{\int_0^\Lambda\,d\rho\Big(\frac{G'_{\rho 00}}{\rho}+G_{\rho 00}\Big) -\frac{(1-\alpha)(1-\beta)(1-\gamma)}{1-\alpha\beta-\alpha\gamma-\beta\gamma+
2\alpha\beta\gamma}G_{\alpha\beta\gamma}\cr
&\times&
\int_0^\Lambda\,d\rho\Big(\frac{G'_{\rho 00}}{\rho}+G_{\rho 00}\Big)
- \int_0^\Lambda \,d\rho\frac{(1-\alpha)(1-\beta)(1-\gamma)G_{\rho\beta\gamma}}{1-\beta\rho-\gamma\rho-\gamma\beta+
2\rho\gamma\beta}\cdot
\frac{(G_{\alpha\beta\gamma}-1)}{(\alpha-\rho)}\cr
&+&\frac{(1-\alpha)(1-\beta)(1-\gamma)G_{\alpha\beta\gamma}}{1-\alpha\beta-\alpha\gamma-\beta\gamma+
2\alpha\beta\gamma}\Big
[m^2\frac{
\int_0^\Lambda\,\frac{\rho d\rho}{(1-\rho)}\Big(\frac{(1-\alpha)G_{\rho\rho\alpha}}{1-\rho^2-2\alpha\rho+2\alpha\rho^2}-\frac{G_{\rho\rho 0}}{1-\rho^2}\Big)}{1+\frac{2\lambda}{m^2}\int_0^\Lambda\,d\rho\Big(\frac{G'_{\rho 00}}{\rho}+G_{\rho 00}\Big) }\cr
&+& \int_0^\Lambda \,\frac{d\rho}{(1-\rho)}\Big(\frac{(1-\beta)(1-\gamma)G_{\rho\beta\gamma}}{1-\beta\rho-\gamma\rho-\gamma\beta+
2\rho\gamma\beta}-G_{\rho 00}
+
\frac{(1-\alpha)(G_{\rho\beta\gamma}-1)
}{(\alpha-\rho)}\cr
&+&\frac{G_{\rho 00}-1}{\rho}\Big)\Big]\Big\}. 
\eea
Simplifying identical terms we get the result of  Theorem \ref{theo1}.
\end{proof}
Note that $G_{000}=1$ and $\partial G_{000}=0$.  The equation \eqref{lama} shows
the occurrence of the singular integral kernel $\int_0^\Lambda
\frac{d\rho}{\rho-\alpha}$, $\int_0^\Lambda
\frac{d\rho}{1-\rho}$ and $\int_0^\Lambda
\frac{d\rho}{\rho}$ for $\Lambda=1$, which needs to be removed.  We will use the Cauchy principal value of the divergent integrals  and  also take the limit value at points $0$ and $1$ i.e.
\beq
\int_0^1=\lim_{\epsilon\rightarrow 0}\Big[\int_0^{a-\epsilon}+\int_{a+\epsilon}^1\Big],\quad a\in(0,1),\quad
\int_0^1=\lim_{\epsilon\rightarrow 0,\epsilon'\rightarrow 1}\int_\epsilon^{\epsilon'}
\eeq
The nonlinear  integral equation \eqref{master3} is  of the form
\bea
G_{\alpha\beta\gamma}=1+\lambda\int_0^1 \,f(G_{\alpha\beta\gamma},G_{\rho\beta\gamma},G_{\rho\alpha 0},G_{\rho 00},\mathcal{Y},\alpha,\beta,\gamma)d\rho.
\eea
Now we can easily see that \eqref{master3} suffers for the lack of symmetry. This inconvenience is due to the position of parameter $\alpha$. So taken $\alpha=0$  we get the symmetric solution given in the following proposition  
\begin{proposition}
At first order in $\lambda$ the solution of the equation \eqref{master3} for $\alpha=0$ is given by
\bea
G_{0\beta\gamma}&=&
1+\lambda'\Big[1+\frac{(1-\beta)(1-\gamma)}{1-\beta\gamma}\Big(\ln\frac{1+\beta\gamma-\beta-\gamma}{1-\beta\gamma}-1\Big)\Big]=1+\lambda'\mathcal K_{0\beta\gamma}.
\eea
 Then, using the symmetry properties of proposition \ref{symnewcal}  we get the symmetric solution $G_{\alpha\beta\gamma}^{sym} $  as
\bea
G^{sym}_{\alpha\beta\gamma}=1+\lambda'_1\mathcal K_{0\beta\gamma}+\lambda'_2\mathcal K_{0\gamma\alpha }+\lambda'_3\mathcal K_{0\alpha\beta }%+O(\lambda_1^2,\lambda_2^2,\lambda_3^2)
\eea
with $\lambda'_\rho=\frac{2\lambda_\rho}{m^2};\,\, \rho=1,2,3$ and  $\alpha,\beta,\gamma\in[0,1).$
\end{proposition}
We use the symmetry relation \eqref{ben}  and
 get the following result
\begin{theorem}
The  closed equation of the symmetric two-point functions $G_{\alpha\beta\gamma}$ satisfies the nonlinear integral equation
\bea\label{lama123}
G_{\alpha\beta\gamma}&=&1+\lambda'\Big\{\mathcal Y+\int_0^1d\rho\, G_{\rho 00}+\frac{(1-\alpha)(1-\beta)(1-\gamma)}{1-\alpha\beta-\alpha\gamma-\beta\gamma+
2\alpha\beta\gamma}\Big[\int_0^1\, d\rho\Big(\frac{G_{\rho\beta\gamma}}{\alpha-\rho}\cr
&+&\frac{(2\beta\gamma-\beta-\gamma)
G_{\rho\beta\gamma}}{1-\beta\rho-\gamma\rho-\gamma\beta+
2\rho\gamma\beta}\Big)
+G_{\alpha\beta\gamma}\Big
[\int_0^1\,d\rho\,\frac{G_{\rho\alpha 0}-G_{\rho 0 0}}{1-\rho}
+\int_0^1\,d\rho\,\frac{G_{\rho 00}}{\rho}-\mathcal Y\cr
&-&\int_0^1d\rho\, G_{\rho 00}-G_{0\alpha0}^{-1}\Big(\int_0^1\,d\rho\,\frac{G_{\rho\alpha 0}}{\rho}+\alpha\int_0^1\,d\rho\,\frac{G_{\rho\alpha 0}}{1-\alpha\rho}
+\int_0^1\,d\rho\,\frac{G_{\rho 00}}{\alpha-\rho}\Big)
\Big]\Big\}.
\eea
\end{theorem}
\begin{proof}
Using the relation \eqref{ben} we can extract the quantity $\mathcal{K}_{\alpha}$ 
 after simplification as
\bea
\mathcal K_{\alpha}&=&-G_{0\alpha\beta}^{-1}\Big[\int_0^1\,d\rho\,\frac{G_{\rho\alpha\beta}}{\rho}-(2\alpha\beta-\alpha-\beta)\int_0^1\,d\rho\,\frac{G_{\rho\alpha\beta}}{1-\alpha\rho-\beta\rho-\alpha\beta+2\alpha\beta\rho}\cr
&+&\int_0^1\,d\rho\,\frac{G_{\rho 0\beta}}{\alpha-\rho}-\beta\int_0^1\,d\rho\,\frac{G_{\rho 0\beta}}{1-\beta\rho}\Big]+\int_0^1\,d\rho\,\frac{G_{\rho\alpha\beta}-G_{\rho 0\beta}}{1-\rho}+\int_0^1\,d\rho\,\Big(\frac{1}{\rho}+\frac{1}{\alpha-\rho}\Big).\cr
&&
\eea
Then remark that $\mathcal{K}_\alpha$ is function of only the parameter $\alpha$.  We then take $\beta=0$ in the last equation and  we get 
\bea\label{efa}
\mathcal K_{\alpha}&=&-G_{0\alpha0}^{-1}\Big(\int_0^1\,d\rho\,\frac{G_{\rho\alpha 0}}{\rho}+\alpha\int_0^1\,d\rho\,\frac{G_{\rho\alpha 0}}{1-\alpha\rho}+\int_0^1\,d\rho\,\frac{G_{\rho 00}}{\alpha-\rho}\Big)\cr
&+&\int_0^1\,d\rho\,\frac{G_{\rho\alpha 0}-G_{\rho 0 0}}{1-\rho}+\int_0^1\,d\rho\,\Big(\frac{1}{\rho}+\frac{1}{\alpha-\rho}\Big).
\eea
By replacing the relation \eqref{efa} in expression \eqref{lama}  we get the desired result.
\end{proof}

Now we are reaching the point where it is possible to give the solution of the equation \eqref{lama123}. Let us write the solution of this equation  as 
\beq
G_{\alpha\beta\gamma}=1+ \sum_{n=1}^\infty\mathcal (\lambda')^nX_{\alpha\beta\gamma}^{(n)}
\eeq
The $n$ order terms $X_{\alpha\beta\gamma}^{(n)}$ can be deduced by iteration. We give here the quantities $X_{\alpha\beta\gamma}^{(1)}$ and $X_{\alpha\beta\gamma}^{(2)}$ in the following statement
\begin{proposition}\label{propsol123}
Pertubatively, at second order in $\lambda$ the symmetry solution of the equation \eqref{lama123}  using the Cauchy principal value is given by
\bea\label{sol123}
G_{\alpha\beta\gamma}&=&1+
\lambda'\mathcal{X}^{(1)}_{\alpha\beta\gamma}
+\lambda'^2\Big\{
\frac{\pi^2}{6}-\frac{3}{2}+\frac{(1-\alpha)(1-\beta)(1-\gamma)}{1-\alpha\beta-\alpha\gamma-\beta\gamma+
2\alpha\beta\gamma}\Big[\mathcal X^{(1)}_{\alpha\beta\gamma}\Big(\ln\frac{(1-\alpha)^2}{\alpha}-1\Big)\cr
&+&\int_0^1\,
d\rho\frac{(2\beta\gamma-\beta-\gamma)
\mathcal{X}^{(1)}_{\rho\beta\gamma}}{1-\beta\rho-\gamma\rho-\beta\gamma
+2\beta\gamma\rho}+
\int_0^1\,d\rho\frac{\mathcal{X}^{(1)}_{\rho
\alpha 0}-\mathcal{X}^{(1)}_{\rho 00}}{1-\rho}-\alpha\int_0^1\,d\rho\,\frac{\mathcal{X}_{\rho\alpha 0}^{(1)}}{1-\alpha\rho}\cr
&
-&\frac{\pi^2}{6}+\frac{3}{2}-\int_0^1\,d\rho\frac{\mathcal X^{(1)}_{\rho \alpha 0}-\mathcal X^{(1)}_{\rho 00}+\mathcal{X}^{(1)}_{0\alpha 0}}{\rho}
-\mathcal{X}^{(1)}_{0\alpha 0}\ln\frac{(1-\alpha)^2}{\alpha}\Big]\Big\}
+O(\lambda'^3),
\eea
where $G_{000}=1$ and where the first order term $\mathcal X^{(1)}_{\alpha\beta\gamma}$ is
\bea
\mathcal X^{(1)}_{\alpha\beta\gamma}&=&1+
\frac{(1-\alpha)(1-\beta)(1-\gamma)}{1-\alpha\beta-\alpha\gamma-\beta\gamma+
2\alpha\beta\gamma}\Big(\ln(1-\alpha)-1+\ln\frac{\beta\gamma-\beta-\gamma+1}{1-\beta\gamma}\Big).
\eea
\end{proposition}
The  exact value of the integrals in the rhs of \eqref{sol123}  are given using the following relations
\bea
\int_0^1\,d\rho\,\frac{\mathcal{X}^{(1)}_{\rho\beta
\gamma}}{a-\rho}&=&\ln\frac{a}{1-a}+\frac{(1-\beta)(1-\gamma)}{1-a\beta-a\gamma-\beta\gamma+
2a\beta\gamma}\Big(-1+\ln\frac{\beta\gamma
-\beta-\gamma+1}{1-\beta\gamma}\Big)\cr
&\times&\Big((1-a)
\ln\frac{a}{1-a}
+\frac{\beta\gamma-\beta-\gamma+1}{2\beta\gamma-\beta-\gamma}
\ln\frac{\beta\gamma-\beta-\gamma+1}{1-\beta\gamma}\Big)\cr
&+&\frac{(1-a)(1-\beta)(1-\gamma)}{1-a\beta-a\gamma-\beta\gamma+
2a\beta\gamma}\Big(\frac{\pi^2}{6}-Li_2\frac{-a}{1-a}+\ln(1-a)\ln\frac{a}{1-a}-\frac{1}{1-a}\Big)\cr
&+&\frac{(1-\beta\gamma)(1-\beta)(1-\gamma)}{(\beta+\gamma-2\beta\gamma)(1-a\beta-a\gamma-\beta\gamma+
2a\beta\gamma)}\Big(\frac{\pi^2}{6}-Li_2\frac{\beta\gamma-\beta-\gamma+1}{1-\beta\gamma}\cr
&-&\ln\frac{\beta+\gamma-2\beta\gamma}{1-\beta\gamma}
\ln\frac{\beta\gamma-\beta-\gamma+1}{1-\beta\gamma}\Big)
\eea
and
\bea
\int_0^1\,d\rho\frac{\mathcal X^{(1)}_{\rho 
\alpha 0}-\mathcal X^{(1)}_{\rho 00}+\mathcal{X}^{(1)}_{0\alpha 0}}{\rho}&=&\frac{(1-\alpha)^2}{\alpha}\ln(1-\alpha)\Big(\ln(1-\alpha)-1\Big)
-(1-\alpha)\Big(\frac{\pi^2}{6}-1\Big)\cr
&+&
\frac{1-\alpha}{\alpha}\Big(\ln\alpha\ln(1-\alpha)+Li_2(1-\alpha)
-\frac{\pi^2}{6}\Big)+\frac{\pi^2}{6}-1
\eea
where
\beq
Li_2(x)=\sum_{k=1}^\infty \frac{x^k}{k^2},\quad Li_2(1)=\frac{\pi^2}{6},\quad Li_2(-1)
=-\frac{\pi^2}{12},\quad Li_2(0)=0.
\eeq

Let us immediately emphasize that the above solution is related to the coupling constant $\lambda_1$. To establish the full solution of  the two-point functions of our model, which takes into account the three coupling constants $\lambda_\rho,\,\,\rho=1,2,3$ we must use 
 the symmetry condition of  proposition \ref{symnewcal}.  The end result is given by the sum of the three equations \eqref{zon},\eqref{zon1} and \eqref{zon2}.  Therefore the two-point functions $G^{Sym}_{\alpha\beta\gamma}$ of $3D$ tensor model is  given by the relation
\bea\label{solutionfull1}
G_{\alpha\beta\gamma}^{sym}&=&1+
\lambda'_1\mathcal{X}^{(1)}_{\alpha\beta\gamma}
+\lambda'^2_1\Big\{
\frac{\pi^2}{6}-\frac{3}{2}+\frac{(1-\alpha)(1-\beta)(1-\gamma)}{1-\alpha\beta-\alpha\gamma-\beta\gamma+
2\alpha\beta\gamma}\Big[\mathcal X^{(1)}_{\alpha\beta\gamma}\Big(\ln\frac{(1-\alpha)^2}{\alpha}-1\Big)\cr
&+&\int_0^1\,
d\rho\frac{(2\beta\gamma-\beta-\gamma)\mathcal{X}^{(1)}_{\rho\beta\gamma}}{1-\beta\rho-\gamma\rho-\beta\gamma
+2\beta\gamma\rho}+
\int_0^1\,d\rho\frac{\mathcal{X}^{(1)}_{\rho\alpha 0}-\mathcal{X}^{(1)}_{\rho 00}}{1-\rho}-\alpha\int_0^1\,d\rho\,\frac{\mathcal{X}_{\rho\alpha 0}^{(1)}}{1-\alpha\rho}\cr
&
-&\frac{\pi^2}{6}+\frac{3}{2}-\int_0^1\,d\rho\frac{\mathcal X^{(1)}_{\rho \alpha 0}-\mathcal X^{(1)}_{\rho 00}+\mathcal{X}^{(1)}_{0\alpha 0}}{\rho}
-\mathcal{X}^{(1)}_{0\alpha 0}\ln\frac{(1-\alpha)^2}{\alpha}\Big]\Big\}
+O(\lambda_1'^3)\cr
&+&
\lambda'_2\mathcal{X}^{(1)}_{\beta\gamma\alpha}
+\lambda'^2_2\Big\{
\frac{\pi^2}{6}-\frac{3}{2}+\frac{(1-\alpha)(1-\beta)(1-\gamma)}{1-\alpha\beta-\alpha\gamma-\beta\gamma+
2\alpha\beta\gamma}\Big[\mathcal X^{(1)}_{\beta\gamma\alpha}\Big(\ln\frac{(1-\beta)^2}{\beta}-1\Big)\cr
&+&\int_0^1\,
d\rho\frac{(2\alpha\gamma-\alpha-\gamma)
\mathcal{X}^{(1)}_{\rho\gamma\alpha}}{1-\alpha\rho-\gamma\rho-\alpha\gamma
+2\alpha\gamma\rho}+
\int_0^1\,d\rho\frac{\mathcal{X}^{(1)}_{\rho\beta 0}-\mathcal{X}^{(1)}_{\rho 00}}{1-\rho}-\beta\int_0^1\,d\rho\,\frac{\mathcal{X}_{\rho\beta 0}^{(1)}}{1-\beta\rho}\cr
&
-&\frac{\pi^2}{6}+\frac{3}{2}-\int_0^1\,d\rho\frac{\mathcal X^{(1)}_{\rho \beta 0}-\mathcal X^{(1)}_{\rho 00}+\mathcal{X}^{(1)}_{0\beta 0}}{\rho}
-\mathcal{X}^{(1)}_{0\beta 0}\ln\frac{(1-\beta)^2}{\beta}\Big]\Big\}
+O(\lambda'^3_2)\cr
&+&
\lambda'_3\mathcal{X}^{(1)}_{\gamma\beta\alpha}
+\lambda'^2_3\Big\{
\frac{\pi^2}{6}-\frac{3}{2}+\frac{(1-\alpha)(1-\beta)(1-\gamma)}{1-\alpha\beta-\alpha\gamma-\beta\gamma+
2\alpha\beta\gamma}\Big[\mathcal X^{(1)}_{\gamma\beta\alpha}\Big(
\ln\frac{(1-\gamma)^2}{\gamma}-1\Big)\cr
&+&\int_0^1\,
d\rho\frac{(2\alpha\beta-\alpha-\beta)
\mathcal{X}^{(1)}_{\rho\alpha\beta}}{1-\alpha\rho-\beta\rho-\alpha\beta
+2\alpha\beta\rho}+
\int_0^1\,d\rho\frac{\mathcal{X}^{(1)}_{\rho\gamma
 0}-\mathcal{X}^{(1)}_{\rho 00}}{1-\rho}-\gamma\int_0^1\,d\rho\,\frac{\mathcal{X}_{\rho\gamma 0}^{(1)}}{1-\gamma\rho}\cr
&
-&\frac{\pi^2}{6}+\frac{3}{2}-\int_0^1\,d\rho\frac{\mathcal X^{(1)}_{\rho \gamma 0}-\mathcal X^{(1)}_{\rho 00}+\mathcal{X}^{(1)}_{0\gamma 0}}{\rho}
-\mathcal{X}^{(1)}_{0\gamma 0}\ln\frac{(1-\gamma)^2}{\gamma}\Big]\Big\}
+O(\lambda'^3_3)
\eea
where $\lambda'_\rho=2\lambda_\rho/m^2;\,\, \rho=1,2,3$,  $\alpha,\beta,\gamma\in(0,1)$ and $G_{000}=1.$
Noting that the solution \eqref{solutionfull1} satisfies the condition \eqref{ben} if and only if  we set $\lambda_1'=\lambda_2'=\lambda_3'$.
Let us also emphasize that the higher  order solution can be get pertubatively by iteration.

\section{Closed equation for two-point functions of rank 4 TGFT}\label{sec4}
The same method use in last section will be performed here to establish  the renormalized  two-point functions of rank $4$ tensor field firstly given in \cite{BenGeloun:2011rc}. We provide  the master equation of the two-point functions.  The action $S_{4D}$ of the model is also subdivided into two terms as
\bea
S_{4D}=S_{4D}^{\kin}+S_{4D}^{\inter}.
\eea
 The kinetic term $S_{4D}^{\kin}$ is given by
\bea
S_{4D}^{\kin}=\sum_{p_j\in\mathbb{Z}}\vp_{1234}
\Big(\sum_{i=1}^4p_i^2+m^2\Big)\bvp_{1234}.
\eea
Noting that in four dimensional case the renormalization is 	guaranteed by the presence of the propagator associated with the heat kernel
 \cite{Geloun:2011cy}: 
\beq
C([p])=\Big(\sum_{i=1}^4p_i^2+m^2\Big)^{-1}
=M_{1234}^{-1}.
\eeq
 $S_{4D}^{\inter}$ is related to the interaction, which is
divided into three fundamental contributions $V_{6,1}$, $V_{6,2}$ and $V_{4,1}$ given by
\bea \label{S62}
V_{6;1} &=& \sum_{p_j\in\mathbb{Z}}
\varphi_{1234} \,\bar\varphi_{1'234}\,\varphi_{1'2'3'4'} \,\bar\varphi_{1''2'3'4'} \,
\varphi_{1''2''3''4''} \,\bar\varphi_{12''3''4''}+ \text{permutations }  \\
V_{6;2} &=& \sum_{p_j\in\mathbb{Z}}
\varphi_{1234} \,\bar\varphi_{1'2'3'4}\,\varphi_{1'2'3'4'} \,\bar\varphi_{1''234'}\,
\varphi_{1''2''3''4''} \,\bar\varphi_{12''3''4''}+ \text{permutations }\\
V_{4;1} & =& \sum_{p_j\in\mathbb{Z}}\varphi_{1234} \,\bar\varphi_{1'234}\,\varphi_{1'2'3'4'} \,\bar\varphi_{12'3'4'}\, 
+ \text{permutations } 
\eea
and  an anomalous  term, namely $V_{4,2}$
\bea
V_{4;2}&=&
\Big(\sum_{p_j\in\mathbb{Z}} \bar\varphi_{1234} \,
\varphi_{1234} \Big)
\Big(\sum_{p_j\in\mathbb{Z}} \bar\varphi_{1'2'3'4'}\, 
\varphi_{1'2'3'4'}\Big) .
\eea
This last vertex is not taken into account  in the computation of the correlation functions due to the fact that it is disconnected and  does not contribute to the melonic Feynman graph of the theory. This vertex could be interpreted as the generation of a scalar matter field out of pure gravity \cite{BenGeloun:2011rc}. The vertices are represented in figure \ref{fig:Vertex4d}.

 Let us immediately emphasize that the vertices of the type $V_{6,1}$ and $V_{4,1}$ are parametrized by four indices $\rho\in\{1,2,3,4\}$,  and  the vertices contributing to $V_{6,2}$ are parametrized by  six  index values  $\rho\rho'\in\{1.2, 1.3,1.4, 2.3,2.4, 3.4\}$.  The couple $\rho\rho'$ will be totally    symmetric i.e., $\rho\rho'=\rho'\rho$.
\begin{figure}[htbp]
\begin{center}
$V_{6,1;1}$\includegraphics[scale=0.6]{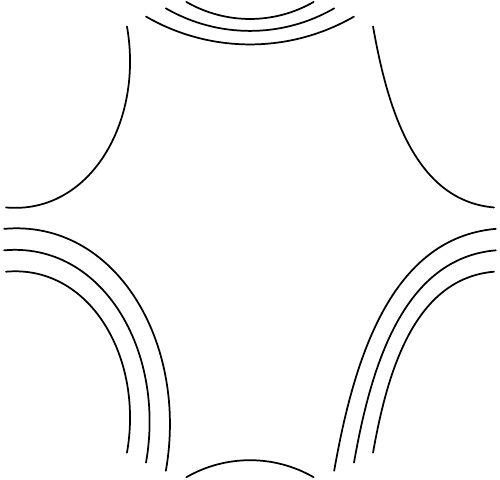}\hspace{0.2cm}
$V_{6,2;14}$\includegraphics[scale=0.6]{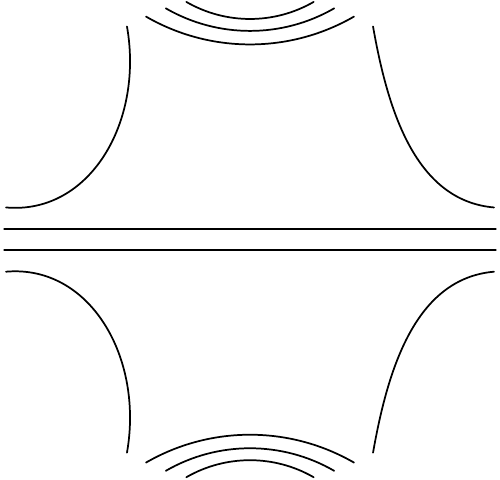}\hspace{0.2cm}
$V_{4,1;1}$\includegraphics[scale=0.6]{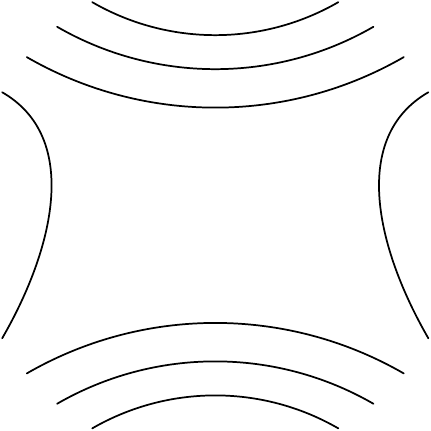}\hspace{0.2cm}
$V_{4,2;1}$\includegraphics[scale=0.6]{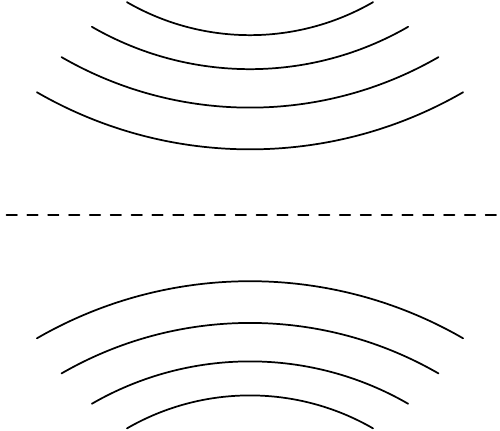}\hspace{0.2cm}
\end{center}
 \caption{Vertex representation of $4D$ tensor model}
  \label{fig:Vertex4d} 
\end{figure} 
One can check that these interaction are invariant under $U(N_a)$ transformations. Then the same procedure of finding the Ward-Takahashi identities applies. 
The  Ward-Takahashi identities  of the equation \eqref{Ward1} is re-expressed as
\bea
\big(M_{m234}-M_{n234}\big)\langle [\vp_{m}\bvp_{n}]_{234}\vp_{n234}\bvp_{ m234}\rangle_c=\langle
\vp_{n234}\bvp_{n234}\rangle_c-
\langle\bvp_{m234}\vp_{m234}\rangle_c\\
\big(M_{1m34}-M_{1n34}\big)\langle [\vp_{m}\bvp_{n}]_{134}\vp_{1n34}\bvp_{ 1m34}\rangle_c=\langle
\vp_{1n34}\bvp_{1n34}\rangle_c-
\langle\bvp_{1m34}\vp_{1m34}\rangle_c\\
\big(M_{12m4}-M_{12n4}\big)\langle [\vp_{m}\bvp_{n}]_{124}\vp_{12n4}\bvp_{ 12m4}\rangle_c=\langle
\vp_{12n4}\bvp_{12n4}\rangle_c-
\langle\bvp_{12m4}\vp_{12m4}\rangle_c\\
\big(M_{123m}-M_{123n}\big)\langle [\vp_{m}\bvp_{n}]_{123}\vp_{123n}\bvp_{ 123m}\rangle_c=\langle
\vp_{123n}\bvp_{123n}\rangle_c-
\langle\bvp_{123m}\vp_{123m}\rangle_c.
\eea
\begin{figure}[htbp]
\begin{center}
\includegraphics[scale=0.45]{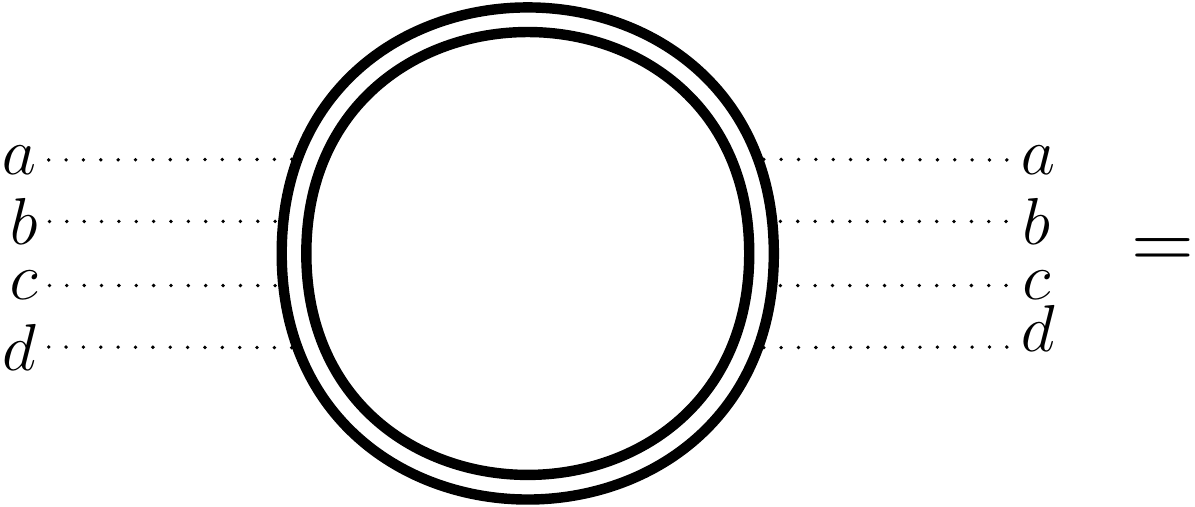}\put(-200,30){$\Gamma_{abcd}=$}\\
$\sum_{\rho}\Big(\Gamma_{abcd}^{6,1;\rho}+\Gamma_{abcd}^{4,1;\rho}\Big)+\sum_{\rho\rho'}\Gamma_{abcd}^{6,2;\rho\rho'}$
\end{center}
 \caption{Schwinger-Dyson equation of rank $4$ tensor model}
  \label{fig:Schwinger2} 
\end{figure}
The figure \ref{fig:Schwinger2}  gives the Schwinger-Dyson equation of the two-point functions.  This figures  collects the 1PI two-point functions.
\begin{figure}[htbp]
\begin{center}
\includegraphics[scale=0.2]{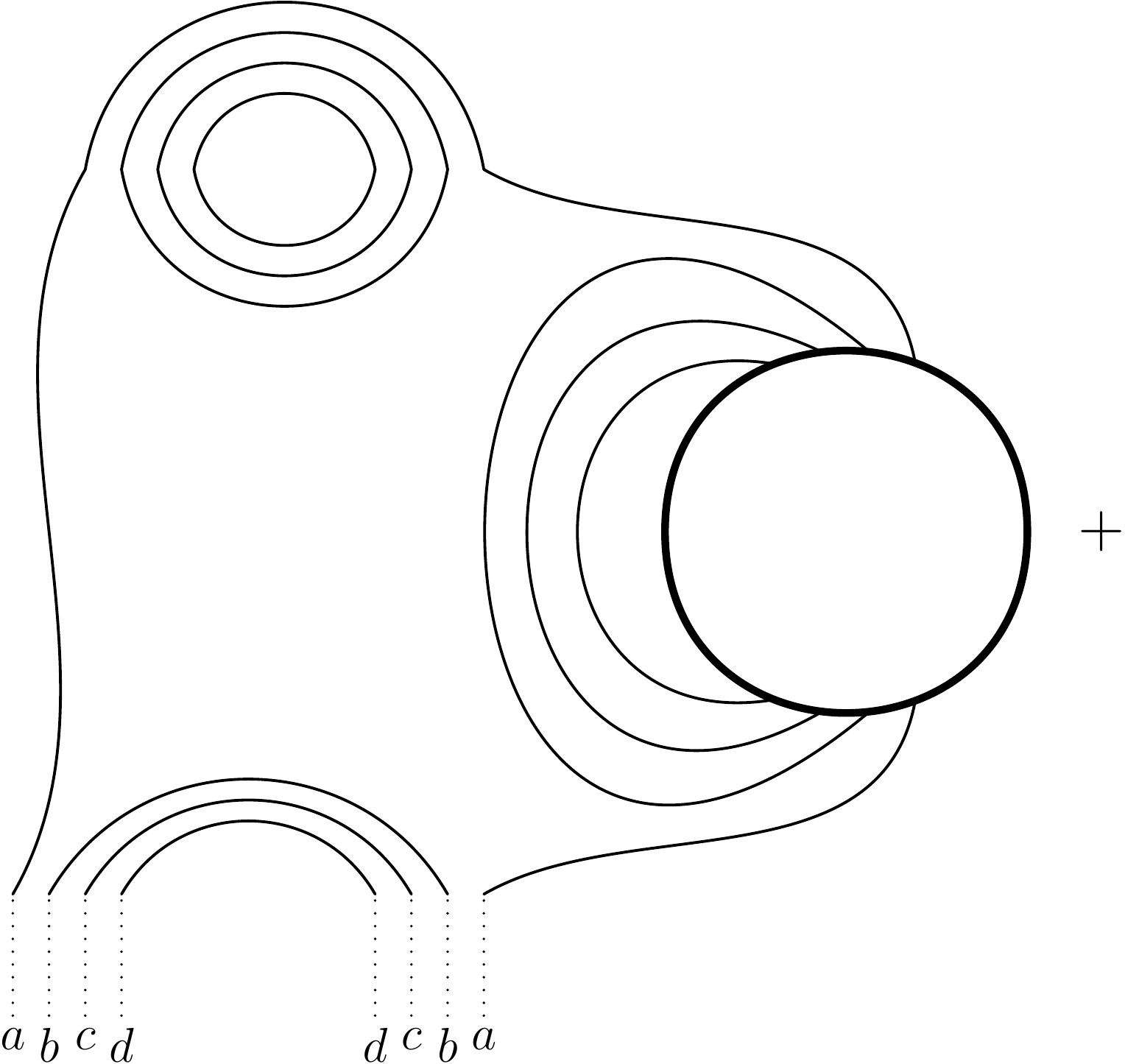}
\includegraphics[scale=0.2]{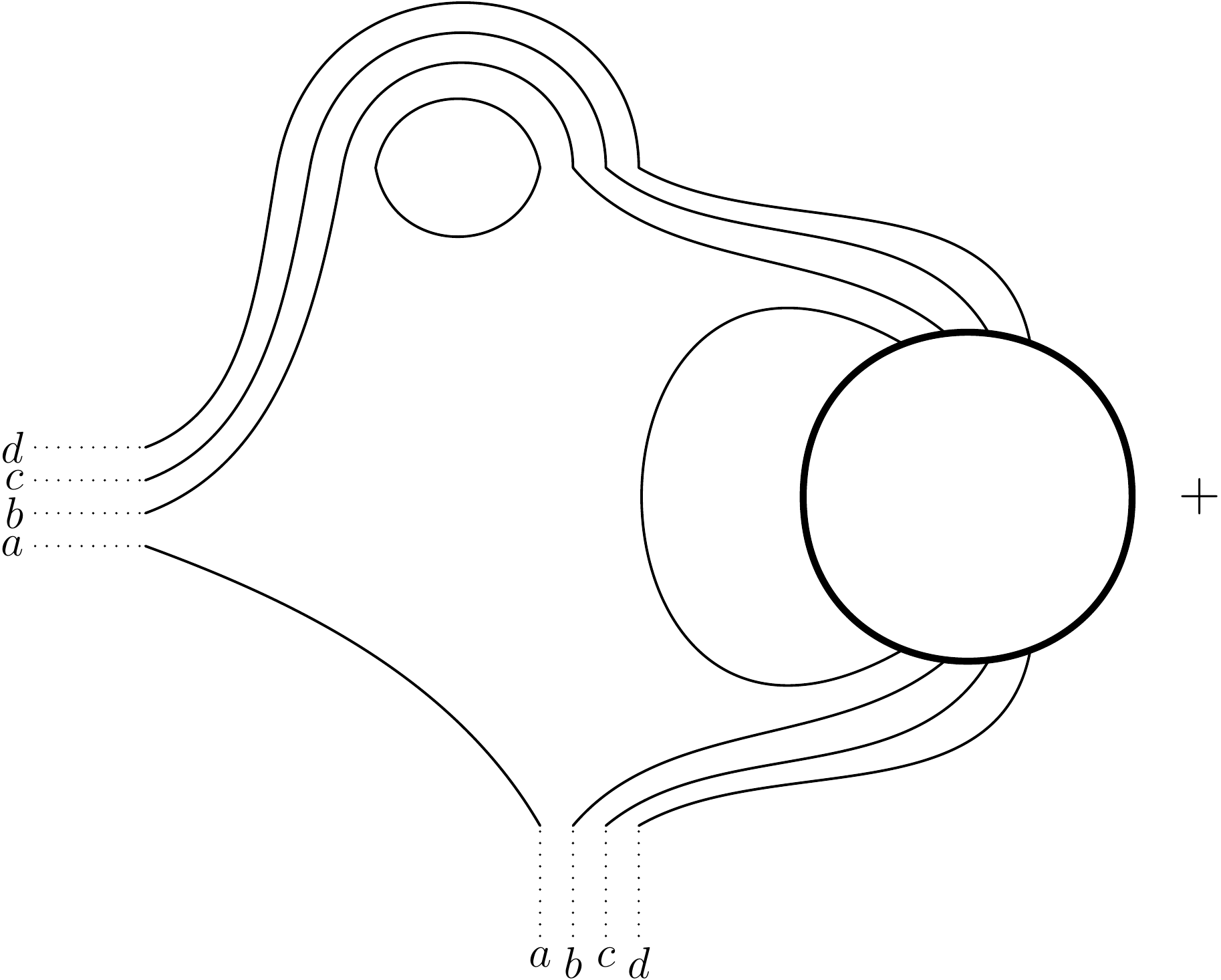}
\includegraphics[scale=0.2]{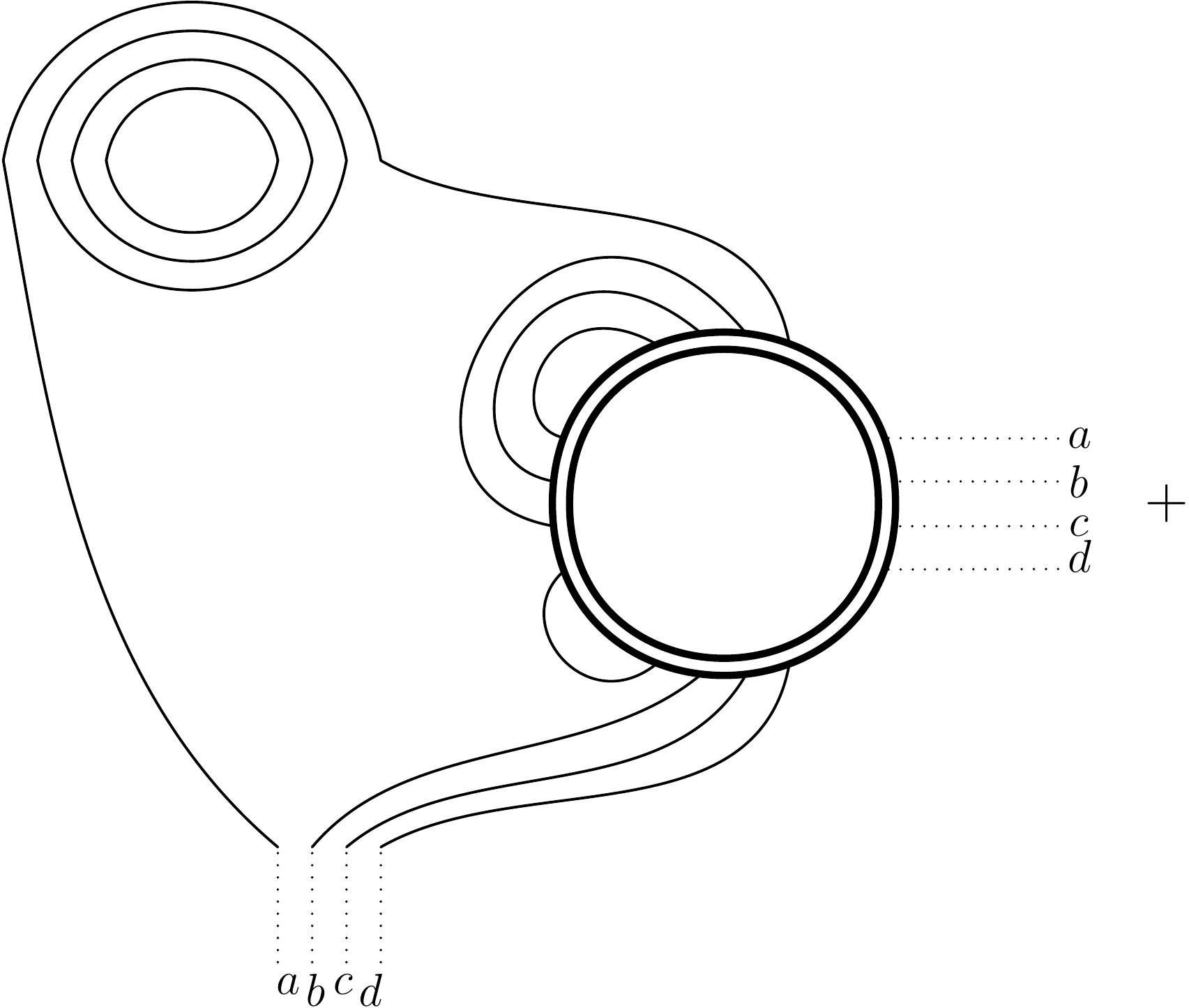}
\put(10,38){ $3$ Permutation $\rho$}
\put(-340,38){$\Gamma_{abcd}^{6,1}=$}
\end{center}
 \caption{}
  \label{fig:Schwinger0} 
\end{figure}
\begin{figure}[htbp]
\begin{center}
\includegraphics[scale=0.2]{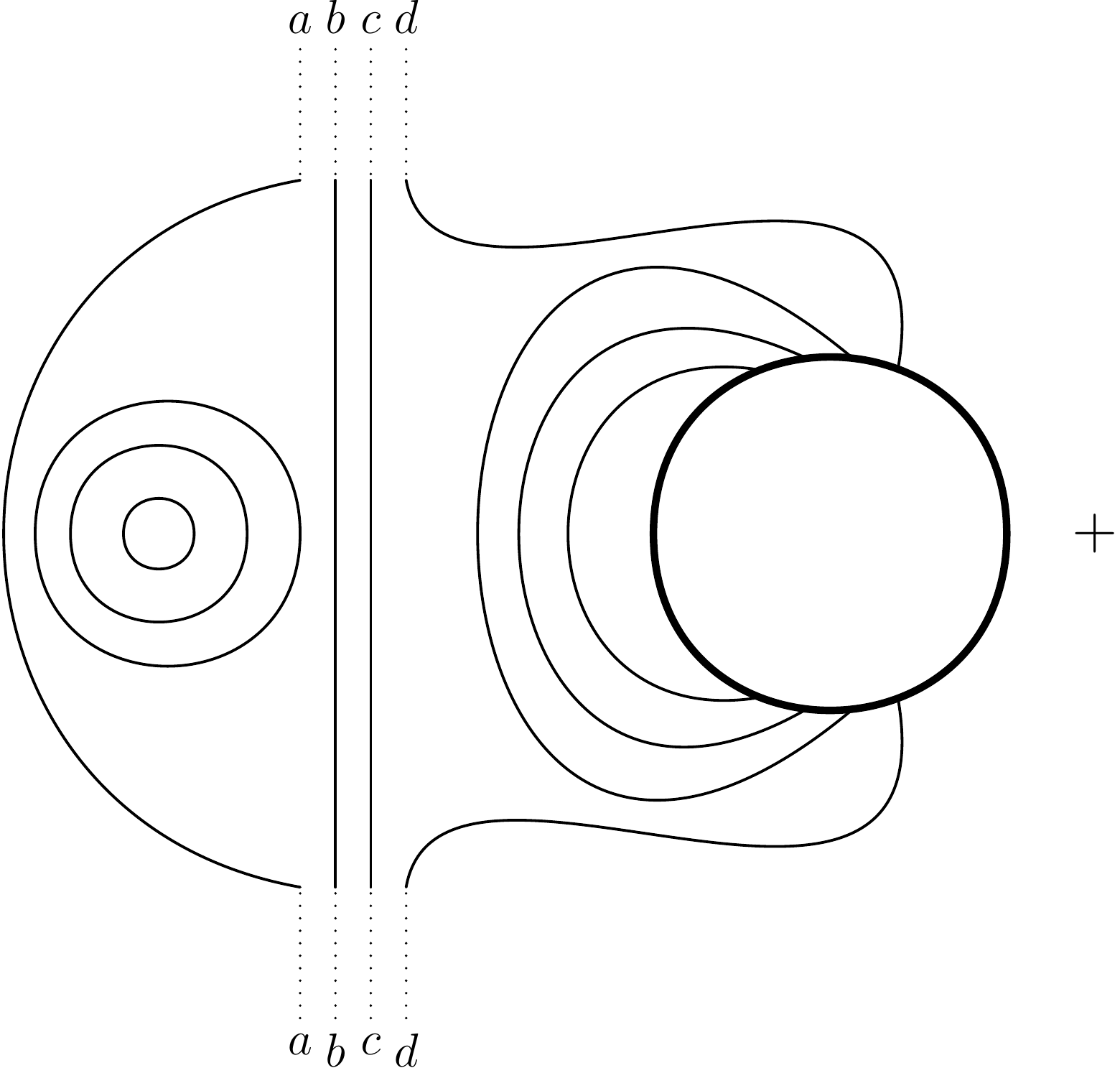}
\includegraphics[scale=0.15]{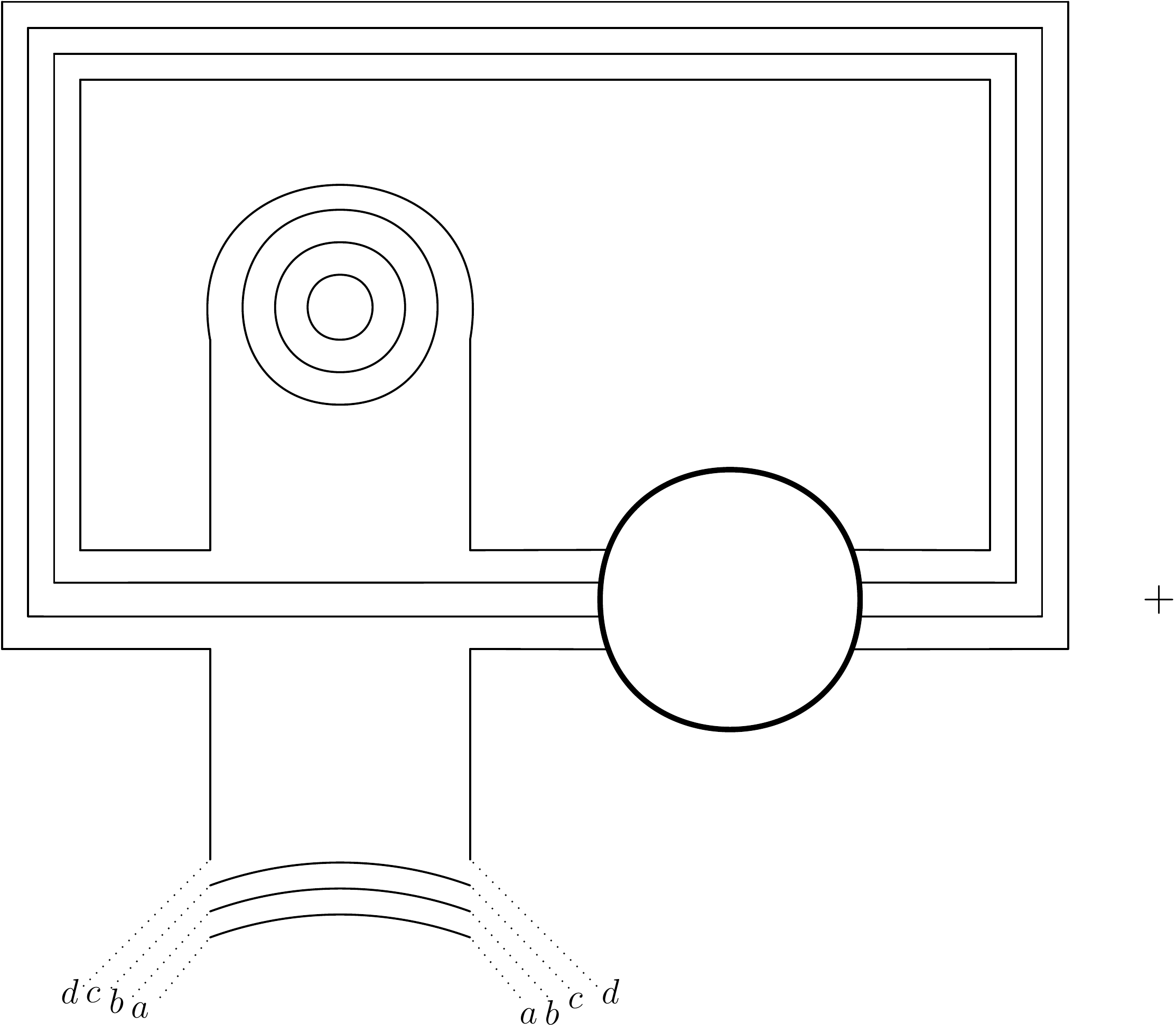}
\includegraphics[scale=0.2]{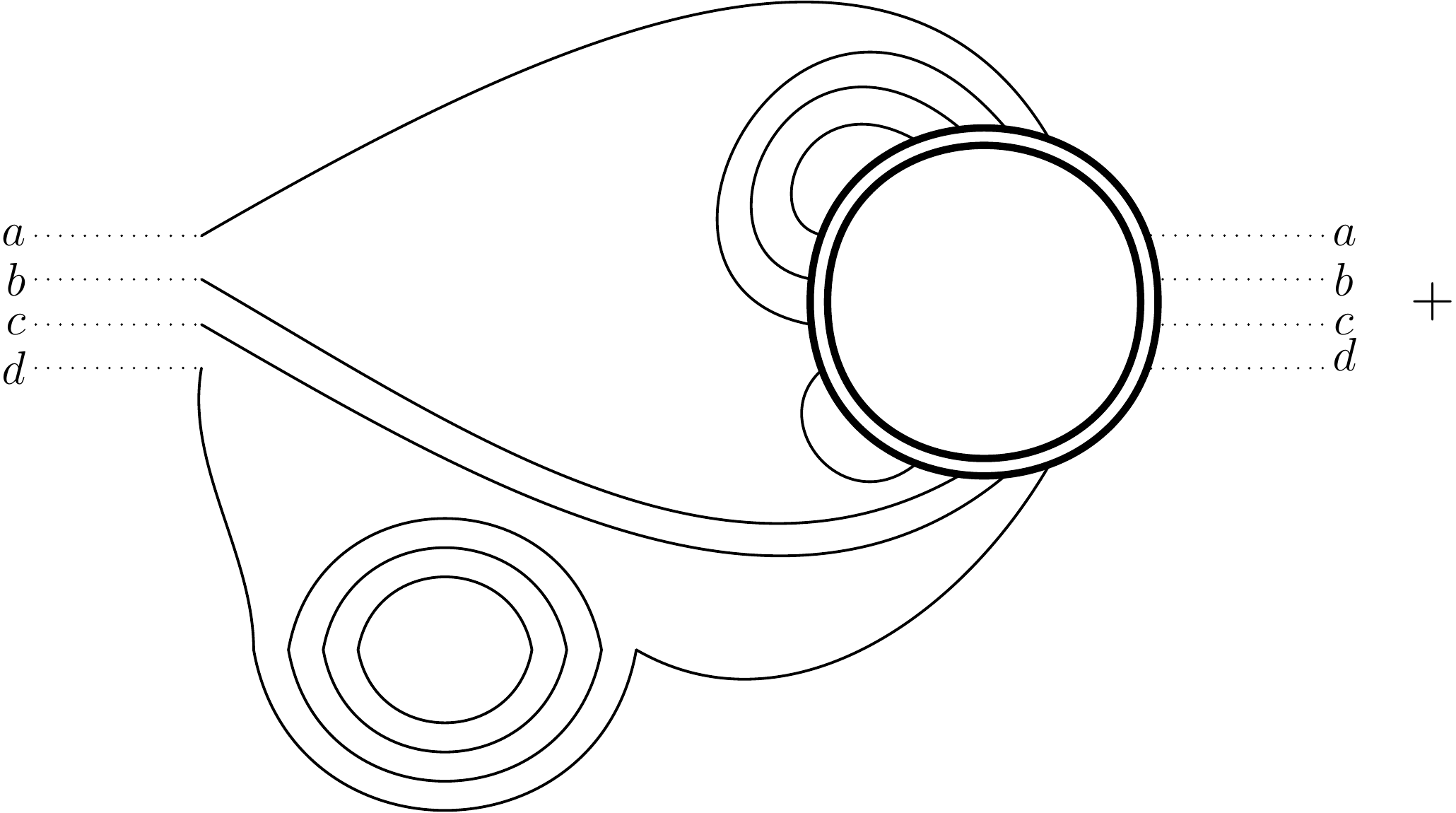}
\put(10,38){ $5$ Permutation $\rho\rho'$}
\put(-360,38){$\Gamma_{abcd}^{6,2}=$}
\end{center}
 \caption{}
  \label{fig:Schwinger00} 
\end{figure}
 \begin{figure}[htbp]
\begin{center}
\includegraphics[scale=0.2]{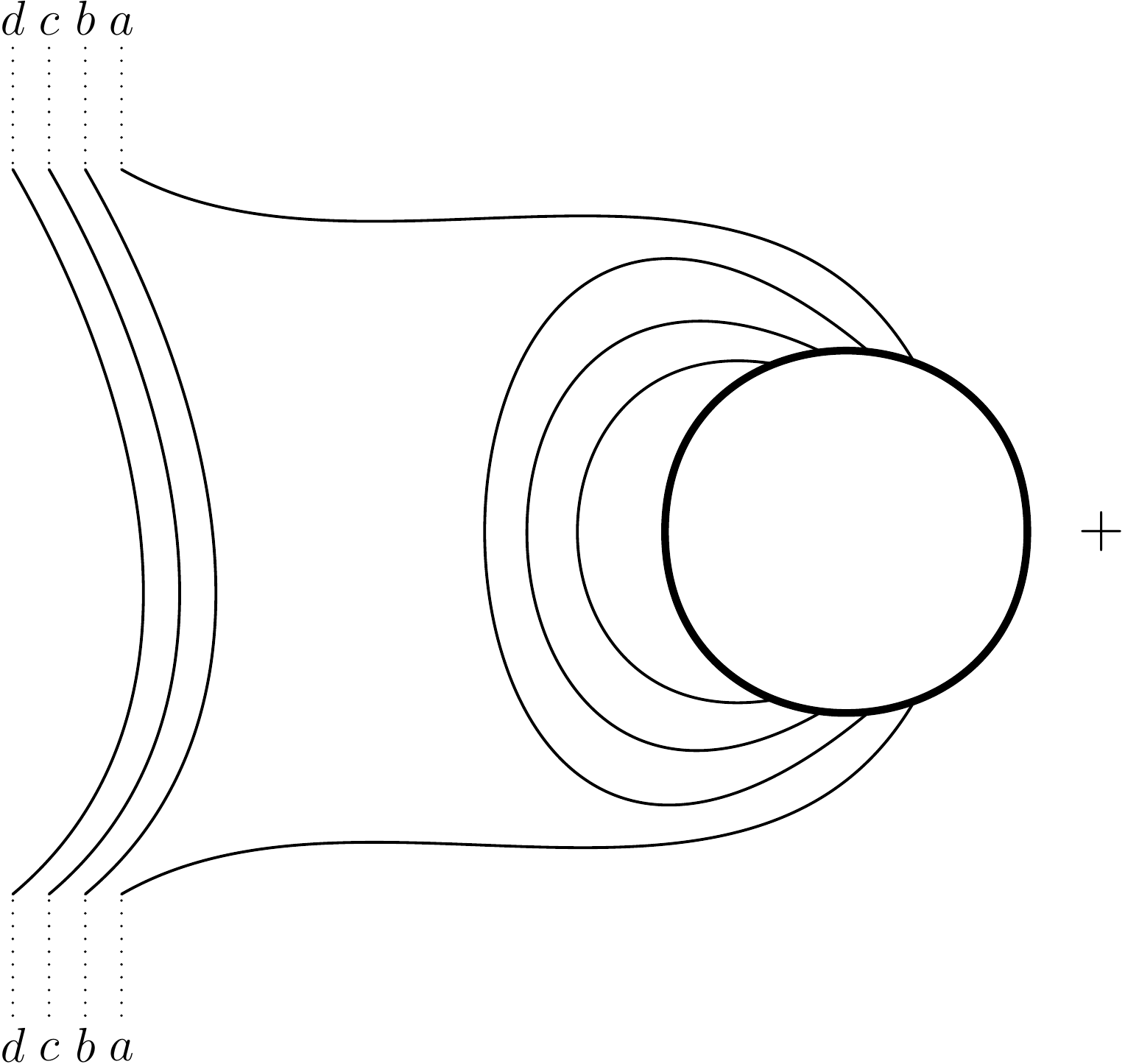}
\includegraphics[scale=0.2]{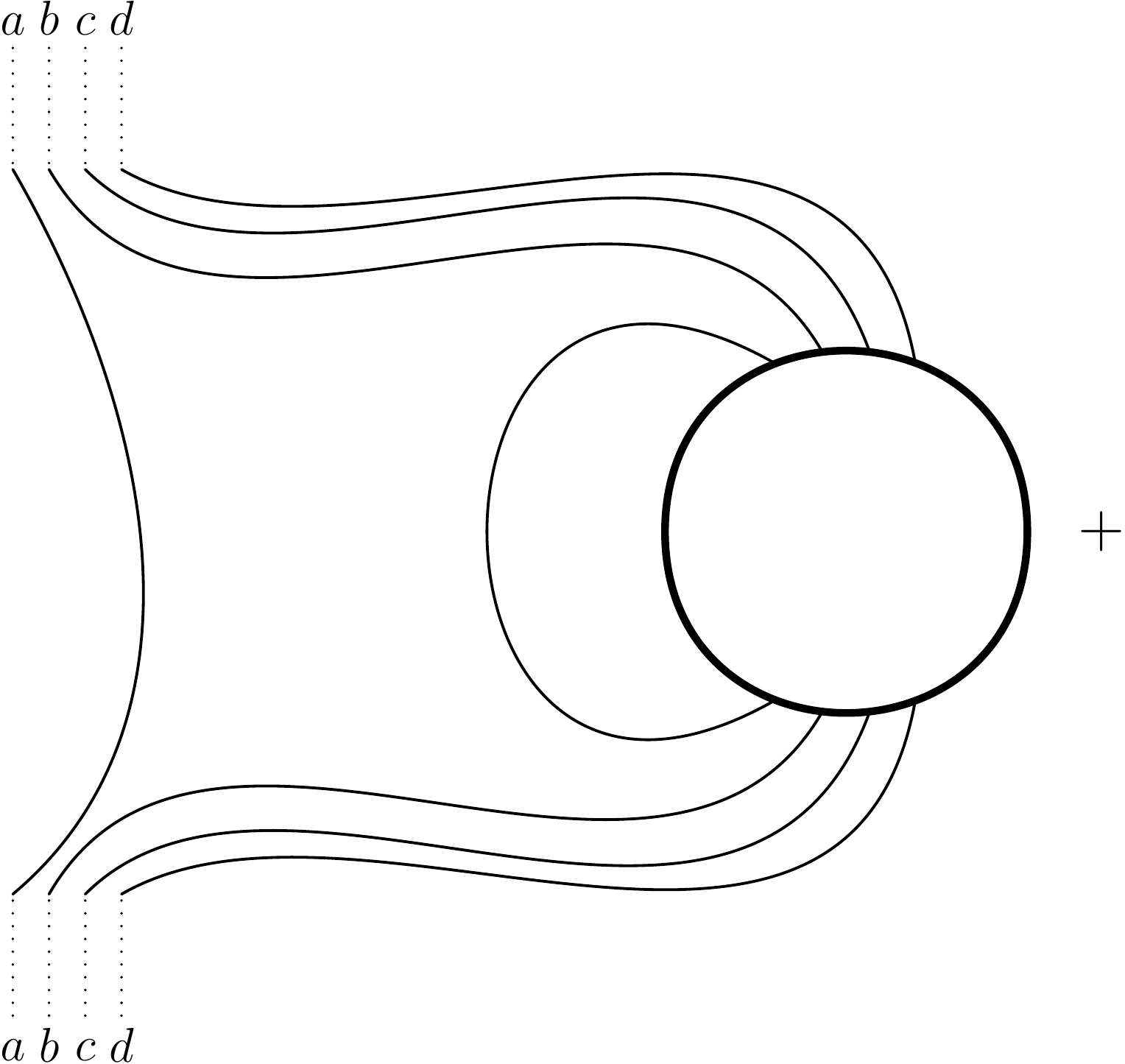}
\includegraphics[scale=0.2]{ward-26.pdf}
\put(10,38){ $3$ Permutation $\rho$}
\put(-350,38){$\Gamma_{abcd}^{4,1}=$}
\end{center}
 \caption{}
  \label{fig:Schwinger000} 
\end{figure}
  Let us discuss  the  contributions in this figure.  The graphs of figure \ref{fig:Schwinger0}  are related to the graphs made with vertex $V_{6,1}$. The first graph of this figure is denoted by $T^{6,1}_{abcd}$ and the sum of the other two  is  $\Sigma^{6,1}_{abcd}$.  The graphs of figure \ref{fig:Schwinger00}  are related to the graphs  built with the vertex $V_{6,2}$.  The first graph of this figure  is called $T^{6,2}_{abcd}$ and the sum of the other two is $\Sigma^{6,2}_{abcd}$.   In the same manner, the graphs of figure \ref{fig:Schwinger000}  take into account the graphs built with vertex $V_{4,1}$. The first graph is called $\Sigma^{4,1}_{abcd}$ and the sum of the over two is $T^{4,1}_{abcd}$.  Then the relations given in figures  \ref{fig:Schwinger0}, \ref{fig:Schwinger00} and \ref{fig:Schwinger000}
 are re-expressed simply as 
\bea
\Gamma_{abcd}^{6,1}=\sum_{\rho}\Gamma_{abcd}^{6,1;\rho},\quad \Gamma_{abcd}^{6,2}=\sum_{\rho\rho'}\Gamma_{abcd}^{6,2;\rho\rho'}\quad \Gamma_{abcd}^{4,1}=\sum_{\rho}\Gamma_{abcd}^{4,1;\rho} 
\eea
with
\bea
\Gamma_{abcd}^{6,1;\rho}=T_{abcd}^{6,1;\rho}+
\Sigma_{abcd}^{6,1;\rho},\quad
\Gamma_{abcd}^{6,2;\rho\rho'}=T_{abcd}^{6,2;\rho\rho'}+
\Sigma_{abcd}^{6,2;\rho\rho'},\quad
\Gamma_{abcd}^{4,1;\rho}=T_{abcd}^{4,1;\rho}+
\Sigma_{abcd}^{4,1;\rho}.
\eea
Therefore the equation on figure \ref{fig:Schwinger2}  takes the form
\bea
\Gamma_{abcd}=\Gamma_{abcd}^{6,1}+
\Gamma_{abcd}^{4,1}+\Gamma_{abcd}^{6,2}.
\eea
All of the above quantities are obtained  by using the following symmetry properties:
\begin{proposition}
\begin{itemize}
\item $\Gamma_{abcd}^{6,1;2}$  can be obtained using $\Gamma_{abcd}^{6,1;1}$
and replaced $a\rightarrow b$ and $b\rightarrow a$. \item $\Gamma_{abcd}^{6,1;3}$ is obtained using $\Gamma_{abcd}^{6,1;1}$
and replaced $a\rightarrow c$, $b\rightarrow a$ and  $c\rightarrow b$.
\item  $\Gamma_{abcd}^{6,1;4}$ is obtained using $\Gamma_{abcd}^{6,1;1}$
and replaced $a\rightarrow d$, $b\rightarrow a$, $c\rightarrow b$ and $d\rightarrow c$.
\end{itemize}
 This above symmetries is well satisfied for $\Gamma_{abcd}^{4,1;\rho}$.  In the case of  $\Gamma_{abcd}^{6,2;\rho\rho'}$ we get:
\begin{itemize}
\item $\Gamma_{abcd}^{6,2;13}$ can be obtained using $\Gamma_{abcd}^{6,2;14}$ and by replaced $a\rightarrow b$ and $b\rightarrow a$.
\item $\Gamma_{abcd}^{6,2;12}$ is obtained by replaced in $\Gamma_{abcd}^{6,2;14}$, $a\rightarrow c$, $b\rightarrow a$ and $c\rightarrow b$.
\item $\Gamma_{abcd}^{6,2;23}$ is obtained by replaced in $\Gamma_{abcd}^{6,2;14}$, $a\rightarrow b$, $b\rightarrow a$,  $c\rightarrow d$ and $d\rightarrow c$.
\item  $\Gamma_{abcd}^{6,2;24}$ is obtained by replaced in $\Gamma_{abcd}^{6,2;14}$, $c\rightarrow d$ and $d\rightarrow c$. 
\item $\Gamma_{abcd}^{6,2;34}$ is obtained by replaced in $\Gamma_{abcd}^{6,2;14}$, $b\rightarrow c$, $c\rightarrow d$ and $d\rightarrow b$.
\end{itemize}
\end{proposition}
We then focus our attention to $\Gamma_{abcd}^1=(\Gamma_{abcd}^{6,1;1}+
\Gamma_{abcd}^{4,1;1})+ \Gamma_{abcd}^{6,2;14}$.  We also call  $G_{[mn]abc}^{ins}$  the two-point functions with insertion (1,2,3) wherein  the momentum indices $p_1,p_2,p_3$  are summed i.e.
\bea
G_{[mn]abc}^{ins}=\sum_{p_1,p_2,p_3}
\langle\vp_{m123}\bvp_{n123}\vp_{nabc}\bvp_{ mabc}\rangle_c.
\eea
The following relations are satisfied:
\bea
&&\Sigma_{abcd}^{6,1;1}=Z^2\lambda_{6,1;1}C_{abcd}\sum_{p}
G_{abcd}^{-1}G_{[ap]bcd}^{ins},\quad T_{abcd}^{6,1;1}=Z^2\lambda_{6,1;1}C_{abcd}
\sum_{p,q,r}G_{pqra},\\
&&\Sigma_{abcd}^{6,2;14}=Z^2\lambda_{6,2;14}C_{abcd}\sum_{p}
G_{abcd}^{-1}G_{[ap]bcd}^{ins},\quad T_{abcd}^{6,2;14}=Z^2\lambda_{6,2;14}C_{abcd}
\sum_{p,q,r}G_{pqra},\\
&&\Sigma_{abcd}^{4,1;1}=Z^2\lambda_{4,1;1}\sum_{p}
G_{abcd}^{-1}G_{[ap]bcd}^{ins},\quad T_{abcd}^{4,1;1}=Z^2\lambda_{6,1;1}
\sum_{p,q,r}G_{pqra}
\eea
and then
\bea
\Gamma_{abcd}^1&=&Z^2C_{abcd}\lambda_{6,1;1}\Big[\sum_{p}
G_{abcd}^{-1}G_{[ap]bcd}^{ins}+
\sum_{p,q,r}G_{pqra}\Big]+Z^2C_{abcd}\lambda_{6,2;14}\Big[\sum_{p}
G_{abcd}^{-1}G_{[ap]bcd}^{ins}\cr
&+&\sum_{p,q,r}G_{pqra}\Big]+Z^2\lambda_{4,1;1}\Big[\sum_{p}
G_{abcd}^{-1}G_{[ap]bcd}^{ins}+
\sum_{p,q,r}G_{pqra}\Big].
\eea
We set $\lambda_{6,1;\rho}=\lambda_{6,1}$, 
$\lambda_{6,2;\rho\rho'}=\lambda_{6,2}$ and $\lambda_{4,1;\rho}=\lambda_{4,1}$. Noting that the connected to point function can be expressed as $G_{abcd}^{-1}=M_{abcd}-\Gamma_{abcd}$. Then we get
\bea\label{dine}
\Gamma_{abcd}^1&=&Z^2M^{-1}_{abcd}\lambda_{6,1}\Big[\sum_{p}
G_{abcd}^{-1}\frac{G_{pbcd}-G_{abcd}}{Z(a^2-p^2)}+
\sum_{p,q,r}G_{pqra}\Big]\cr
&+&Z^2M^{-1}_{abcd}\lambda_{6,2}\Big[\sum_{p}
G_{abcd}^{-1}\frac{G_{pbcd}-G_{abcd}}{Z(a^2-p^2)}+\sum_{p,q,r}G_{pqra}\Big]\cr
&+&Z^2\lambda_{4,1}\Big[\sum_{p}
G_{abcd}^{-1}\frac{G_{pbcd}-G_{abcd}}{Z(a^2-p^2)}+
\sum_{p,q,r}G_{pqra}\Big]\cr
&=&Z^2M^{-1}_{abcd}\lambda_{6,1}\Big[\sum_{p}\Big(\frac{1}{M_{pbcd}-\Gamma_{pbcd}}-\frac{1}{M_{pbcd}-\Gamma_{pbcd}}
\frac{\Gamma_{abcd}-\Gamma_{pbcd}}{Z(a^2-p^2)}\Big)
+
\sum_{p,q,r}\frac{1}{M_{pqra}-\Gamma_{pqra}}\Big]\cr
&+&Z^2M^{-1}_{abcd}\lambda_{6,2}\Big[\sum_{p}\Big(\frac{1}{M_{pbcd}-\Gamma_{pbcd}}
-\frac{1}{M_{pbcd}-\Gamma_{pbcd}}
\frac{\Gamma_{abcd}-\Gamma_{pbcd}}{Z(a^2-p^2)}\Big)+\sum_{p,q,r}\frac{1}{M_{pqra}-\Gamma_{pqra}}\Big]\cr
&+&Z^2\lambda_{4,1}\Big[\sum_{p}\Big(\frac{1}{M_{pbcd}-\Gamma_{pbcd}}
-\frac{1}{M_{pbcd}-\Gamma_{pbcd}}
\frac{\Gamma_{abcd}-\Gamma_{pbcd}}{Z(a^2-p^2)}\Big)+
\sum_{p,q,r}\frac{1}{M_{pqra}-\Gamma_{pqra}}\Big].
\eea
Now we use the Taylor expansion that allows  to pass to the renormalized quantity as
\beq
\Gamma_{abcd}^1=Zm_{bar}^2-m^2+(Z-1)(a^2+b^2+c^2+d^2)+\Gamma_{abcd}^{phys},
\eeq 
with condition $\Gamma_{0000}=0$ and $\partial\Gamma_{0000}=0$. This implies that 
\beq
G_{abcd}^{-1}=a^2+b^2+c^2+d^2+m^2-\Gamma_{abcd}^
{ren}.
\eeq
Then we get the following proposition
\begin{proposition}
The closed equation of the two-point functions of four dimension tensor model is given by
\bea\label{closedp}
&&\quad\quad\quad\quad\quad\quad\quad\quad(Z-1)(a^2+b^2+c^2+d^2)+\Gamma_{abcd}^{phys}\cr
&&=M^{-1}_{abcd}\lambda_{6,1}\Big\{\sum_{p}\Big[\frac{Z}{p^2+b^2+c^2+d^2+m^2-\Gamma_{pbcd}^{phys}}-\frac{1}{m^2}
\frac{M_{abcd}}{(p^2+m^2-\Gamma_{p000}^{phys})}\cr
&&-\frac{Z}{p^2+b^2+c^2+d^2+m^2-\Gamma_{pbcd}^{phys}}
\frac{\Gamma_{abcd}^{phys}-\Gamma_{pbcd}^{phys}}{(a^2-p^2)}+\frac{1}{m^2}\frac{M_{abcd}}{(p^2+m^2-\Gamma_{p000}^{phys})}
\frac{\Gamma_{p000}^{phys}}{p^2}\Big]\cr
&+&
\sum_{p,q,r}\Big[\frac{Z^2}{p^2+q^2+r^2+a^2+m^2-\Gamma_{pqra}^{phys}}-
\frac{1}{m^2}\frac{ZM_{abcd}}{(p^2+q^2+r^2+m^2-\Gamma_{pqr0}^{phys})}\Big]\Big\}
\cr
&&+M^{-1}_{abcd}\lambda_{6,2}\Big\{\sum_{p}\Big[\frac{Z}{p^2+b^2+c^2+d^2+m^2-\Gamma_{pbcd}^{phys}}-\frac{1}{m^2}\frac{M_{abcd}}{(p^2+m^2-\Gamma_{p000}^{phys})}\cr
&-&\frac{Z}{p^2+b^2+c^2+d^2+m^2-\Gamma_{pbcd}^{phys}}
\frac{\Gamma_{abcd}^{phys}-\Gamma_{pbcd}^{phys}}{(a^2-p^2)}+\frac{1}{m^2}\frac{M_{abcd}}{(p^2+m^2-\Gamma_{p000}^{phys})}
\frac{\Gamma_{p000}^{phys}}{p^2}\Big]\cr
&+&
\sum_{p,q,r}\Big[\frac{Z^2}{p^2+q^2+r^2+a^2+m^2-\Gamma_{pqra}^{phys}}-
\frac{1}{m^2}\frac{ZM_{abcd}}{(p^2+q^2+r^2+m^2-\Gamma_{pqr0}^{phys})}\Big]\Big\}
\cr
&&+\lambda_{4,1}\Big\{\sum_{p}\Big[\frac{Z}{p^2+b^2+c^2+d^2+m^2-\Gamma_{pbcd}^{phys}}-\frac{Z}{p^2+m^2-\Gamma_{p000}^{phys}}\cr
&-&\frac{Z}{p^2+b^2+c^2+d^2+m^2-\Gamma_{pbcd}^{phys}}
\frac{\Gamma_{abcd}^{phys}-\Gamma_{pbcd}^{phys}}{(a^2-p^2)}+\frac{Z}{p^2+m^2-\Gamma_{p000}^{phys}}
\frac{\Gamma_{p000}^{phys}}{p^2}\Big]\cr
&&+
\sum_{p,q,r}\Big[\frac{Z^2}{p^2+q^2+r^2+a^2+m^2-\Gamma_{pqra}^{phys}}-\frac{Z^2}{p^2+q^2+r^2+m^2-\Gamma_{pqr0}^{phys}}\Big]\Big\}.
\eea
\end{proposition}
\begin{proof} The equation \eqref{closedp} can be simply obtained using the relation of  $Zm_{bar}^2-m^2$ in the same way of the last section as 
\bea\label{samary}
&&Zm_{bar}^2-m^2=ZM^{-1}_{0000}\lambda_{6,1}\Big[\sum_{p}\Big(\frac{1}{p^2+m^2-\Gamma_{p000}^{phys}}-\frac{1}{p^2+m^2-\Gamma_{p000}^{phys}}
\frac{\Gamma_{p000}^{phys}}{p^2}\Big)\cr
&&+
\sum_{p,q,r}\frac{Z}{p^2+q^2+r^2+m^2-\Gamma_{pqr0}^{phys}}\Big]
+ZM^{-1}_{0000}\lambda_{6,2}\Big[\sum_{p}\Big(\frac{1}{p^2+m^2-\Gamma_{p000}^{phys}}\cr
&&-\frac{1}{p^2+m^2-\Gamma_{p000}^{phys}}
\frac{\Gamma_{p000}^{phys}}{p^2}\Big)+\sum_{p,q,r}\frac{Z}{p^2+q^2+r^2+m^2-\Gamma_{pqr0}^{phys}}\Big]
+Z\lambda_{4,1}\Big[\sum_{p}\Big(\frac{1}{p^2+m^2-\Gamma_{p000}^{phys}}
\cr
&&-\frac{1}{p^2+m^2-\Gamma_{p000}^{phys}}
\frac{\Gamma_{p000}^{phys}}{p^2}\Big)+
\sum_{p,q,r}\frac{Z}{p^2+q^2+r^2+m^2-\Gamma_{pqr0}^{phys}}\Big],\quad M_{0000}^{-1}=\frac{1}{Zm^2}
\eea
Then \eqref{closedp} takes the form by replacing the relation \eqref{samary} into the right hand side of equation \eqref{dine}.
\end{proof}
Let us remark that the continuous limit of the equation \eqref{closedp}  can be  built. We identify the sum  as $\sum_p=2\int_0^\infty\,dp$ and $\sum_{p,q,r}=2\int_0^\infty\,p^2dp$. We also impose the cutoff $p_\Lambda$ in the $UV$ and changing the variables as
\bea
a^2=m^2\frac{\alpha}{1-\alpha},\quad b^2=m^2\frac{\beta}{1-\beta},\quad c^2=m^2\frac{\gamma}{1-\gamma},\cr d^2=m^2\frac{\epsilon}{1-\epsilon},\quad p^2=m^2\frac{\rho}{1-\rho},\quad p^2_\Lambda=m^2\frac{\Lambda}{1-\Lambda}.
\eea
Now let us define the two quantities $s(\alpha,\beta,\gamma,\epsilon)$ and $p(\alpha,\beta,\gamma,\epsilon)$
as
\beq
s(\alpha,\beta,\gamma,\epsilon)=1-\alpha\beta-\alpha\gamma-\alpha\epsilon
-\beta\gamma-\beta\epsilon-\gamma\epsilon+2\alpha\beta\gamma
+2\alpha\beta\epsilon+2\alpha\gamma\epsilon+2\beta\gamma\epsilon
-3\alpha\beta\gamma\epsilon
\eeq
and
\beq
p(\alpha,\beta,\gamma,\epsilon)=(1-\alpha)(1-\beta)(1-\gamma)(1-\epsilon).
\eeq
The equation \eqref{closedp} is re-expressed  as
\bea\label{cloclo}
&& m^2(Z-1)\frac{p(\alpha,\beta,\gamma,\epsilon)}{(1-\alpha)(1-\beta)(1-\gamma)(1-\epsilon)}+m^2\frac{\Gamma_{\alpha\beta\gamma\epsilon}}{(1-\alpha)(1-\beta)(1-\gamma)(1-\epsilon)}\cr
&=&2M^{-1}_{\alpha\beta\gamma\epsilon}(\lambda_{6,1}+\lambda_{6,2})\Big\{\int_0^\Lambda\,\frac{1}{2m}\sqrt{\frac{1-\rho}{m\rho}}\frac{d\rho}{(1-\rho)^2}\Big[\frac{Z(1-\rho)(1-\beta)(1-\gamma)(1-\epsilon)}{s(\rho,\beta,\gamma,\epsilon)-\Gamma_{\rho\beta\gamma\epsilon}}\cr
&-&\frac{1}{m^2}\frac{M_{\alpha\beta\gamma\epsilon}(1-\rho)}{1-\Gamma_{\rho 000}}-\frac{Z(1-\rho)}{(s(\rho,\beta,\gamma,\epsilon)-\Gamma_{\rho\beta\gamma\epsilon})}\frac{(1-\rho)
\Gamma_{\alpha\beta\gamma\epsilon}-(1-\alpha)
\Gamma_{\rho\beta\gamma\epsilon}}{(\alpha-\rho)}\cr
&+&\frac{1}{m^2}\frac{M_{\alpha\beta\gamma\epsilon}(1-\rho)}{(1-\Gamma_{\rho 000})}\frac{\Gamma_{\rho 000}}{\rho}\Big]
+\int_0^\Lambda\frac{1}{2m}\sqrt{\frac{m\rho}{1-\rho}}\frac{d\rho}{(1-\rho)^2}\Big[\frac{Z^2(1-\rho)^3(1-\alpha)}{s(\rho,\rho,\rho,\alpha)-\Gamma_{\rho\rho\rho\alpha}}\cr
&-&\frac{1}{m^2}\frac{ZM_{\alpha\beta\gamma\epsilon}(1-\rho)^3}{(2\rho^3-3\rho^2+1-\Gamma_{\rho\rho\rho 0})}\Big]\Big\}\cr
&+&\lambda_{4,1}\Big\{\int_0^\Lambda\,\frac{1}{2m}\sqrt{\frac{1-\rho}{m\rho}}\frac{d\rho}{(1-\rho)^2}\Big[\frac{Z(1-\rho)(1-\beta)(1-\gamma)(1-\epsilon)}{s(\rho,\beta,\gamma,\epsilon)-\Gamma_{\rho\beta\gamma\epsilon}}\cr
&-&\frac{Z(1-\rho)}{1-\Gamma_{\rho 000}}-\frac{Z(1-\rho)}{(s(\rho,\beta,\gamma,\epsilon)-\Gamma_{\rho\beta\gamma\epsilon})}\frac{(1-\rho)
\Gamma_{\alpha\beta\gamma\epsilon}-(1-\alpha)
\Gamma_{\rho\beta\gamma\epsilon}}{(\alpha-\rho)}\cr
&+&\frac{Z(1-\rho)}{(1-\Gamma_{\rho 000})}\frac{\Gamma_{\rho 000}}{\rho}\Big]
+\int_0^\Lambda\frac{1}{2m}\sqrt{\frac{m\rho}{1-\rho}}\frac{d\rho}{(1-\rho)^2}\Big[\frac{Z^2(1-\rho)^3(1-\alpha)}{s(\rho,\rho,\rho,\alpha)-\Gamma_{\rho\rho\rho\alpha}}\cr
&-&\frac{Z^2(1-\rho)^3}{(2\rho^3-3\rho^2+1-\Gamma_{\rho\rho\rho 0})}\Big]\Big\}.
\eea
The wave function $Z$ can be also  deduced  as
\bea\label{Znew}
Z=\frac{1-\frac{2}{m^2}(\lambda_{6,1}+\lambda_{6,2})\int_0^\Lambda\,\frac{d\rho}{2m^3}\sqrt{\frac{1-\rho}{m\rho}}\Big(G_{\rho 000}+\frac{G'_{\rho 000}}{\rho}\Big)}{1+\lambda_{4,1}\int_0^\Lambda\,\frac{d\rho}{2m^3}\sqrt{\frac{1-\rho}{m\rho}}\Big(G_{\rho 000}+\frac{G'_{\rho 000}}{\rho}\Big)},
\eea
where
\bea\label{changevar}
s(\alpha,\beta,\gamma,\epsilon)-\Gamma_{\alpha\beta\gamma\epsilon}=\frac{s(\alpha,\beta,\gamma,\epsilon)}{G_{\alpha\beta\gamma\epsilon}},
\eea
and 
\beq
M_{\alpha\beta\gamma\epsilon}=Zm^2\frac{s(\alpha,\beta\gamma,\epsilon)}{p(\alpha,\beta\gamma,\epsilon)}.
\eeq
Finally by replacing  the expressions \eqref{Znew} and \eqref{changevar}  in the equation \eqref{cloclo}, we obtain the closed equation in the continuous limit, which will also be   fully addressed in forthcoming  work.

%%%%%%%%%%%%%%%%%%%%%%%%%%%%

%%%%%%%%%%%%%%%%%%%%%%%%%%%

\section{ Conclusion}\label{sec5}
In the present paper, we have presented a perturbative calculation of two-point correlation functions of  rank $3$ TGFT. As discussed earlier the correlation functions are given  by
combining  Ward-Takahashi identities and Schwinger-Dyson equations that allows  to establish the appropriate closed equation.
The closed equation in the $4D$ case is also given.

 In this work, we proved that the nonperturbative techniques as developed in \cite{Grosse:2013iva}\cite{Grosse:2012uv}\cite{Grosse:2009pa}  can be reported to the tensor situation. Indeed, although, we only solve our closed form equations for the two-point functions
at initial orders, it is very promising to see that we can obtain even solutions in this highly combinatoric case. 
As future investigations, we can now undertake a calculation of the general solution at all orders of the coupling constants for both rank 3 and 4 models.

%%%%%%%%%%%%%%%%%%%%%%%%%%

%%%%%%%%%%%%%%%%%%%%%%%%%%

\section*{Acknowledgements}
Discussions with Joseph Ben Geloun, Raimar Wulkenhaar and Vincent Rivasseau are gratefully acknowledged. 
This research was
supported in part by Perimeter Institute for Theoretical Physics and Fields Institute  for Research in Mathematical Sciences (Toronto). Research at Perimeter Institute is supported by the
Government of Canada through Industry Canada and by the Province of Ontario through the Ministry of Research
and Innovation.

%%%%%%%%%%%%%%%%%%%%

%%%%%%%%%%%%%%%%%%%%%

\end{document}